\documentclass[12pt, a4paper]{extarticle}
\usepackage[margin=2cm]{geometry}
\usepackage{extarrows}
\usepackage[english]{babel}
\usepackage{graphicx} 
\usepackage{wrapfig}
\usepackage{amsmath,amsthm,amssymb,scrextend}
\usepackage{booktabs}            
\usepackage{multirow}            
\usepackage{amsfonts}            
\usepackage{multicol}
\usepackage[shortlabels]{enumitem}
\usepackage{cutwin}
\usepackage{float}
\usepackage{csquotes}
\usepackage{enumitem}
\usepackage{tikz}
\usepackage{hyperref}

\newtheorem{thm}{Theorem}
\newtheorem{lm}{Lemma}
\newtheorem{cor}{Corollary}
\newtheorem{st}{Statement}
\newtheorem*{thm*}{Theorem}
\newtheorem*{lm*}{Lemma}
\newtheorem*{st*}{Statement}

\theoremstyle{definition}
\newtheorem{defin}{Definition}[section]

\theoremstyle{remark}
\newtheorem{ex}{Example}[section]
\newtheorem{rk}{Remark}[section]

\renewenvironment{proof}{\begin{addmargin}[1em]{0em}\begin{newproof}}{\end{newproof}\end{addmargin}\qed}

\def \bf#1 {\textbf{#1 }}

\def \rN {\mathbb{N}}
\def \rR {\mathbb{R}}
\def \rZ {\mathbb{Z}}
\def \sub {\subset}
\def \ifof {\emph{iff. }}
\def \NB {\noindent\textbf{NB! }}
\def \imp {\implies}
\def \b {\beta}
\def \a {\alpha}
\def \l {\lambda}
\def \k {\kappa}
\def \d {\partial}
\def \eps {\epsilon}
\def \R {\mathbb{R}}
\def \j {\textbf{j}}
\def \cN {\mathcal{N}}

\DeclareMathOperator{\rank}{rk}
\DeclareMathOperator{\tr}{tr}
\DeclareMathOperator{\dd}{d}
\DeclareMathOperator{\pp}{p}
\DeclareMathOperator{\low}{low}
\DeclareMathOperator{\DFS}{DFS}
\DeclareMathOperator{\vol}{vol}
\DeclareMathOperator{\len}{len}
\DeclareMathOperator{\dist}{dist}

\newcommand{\sumt}{\sum\limits}
\newcommand{\maxt}{\max\limits}
\newcommand{\mint}{\min\limits}

\begin{document}

\title{Introduction to graph theory and basic algorithms}
\author{Mikhail Tuzhilin%
  \thanks{Affiliation: Moscow State University, Electronic address: \texttt{mtu93@mail.ru}};
  Dong Zhang%
  \thanks{Affiliation: Peking University, Electronic address: \texttt{dongzhang@math.pku.edu.cn};}
  }
\date{}
\maketitle

This book collects the lectures about graph theory and its applications which were given to students of mathematical departments of Moscow State University and Peking University. Graph theory is a very wide field with a lot of applications in almost every scientific area: in many branches of mathematics, computer science, physics, chemistry, biology and also in psychology, arts, philosophy and many others. Nowadays, graph theory becomes especially more important because of the rapid development of molecular biology, neural networks and AI fields. One of the aims of writing this book was to give students thorough knowledge about graphs to understand modern scientific fields more deeply. Here we tried to give classical and modern theorems and algorithms in more understandable and simple way. We spent many time to rewrite them and close the gaps in several \textquote{simplest} well-known proofs to provide more precise and accurate material for students.

The book starts with the basic definitions and assumptions which are required for the further material and slowly increases the complexity. Note that the book's narrative is sequential: earlier theorems and definitions are used in later material. After theoretical parts usually goes the part with algorithms and examples. The last three sections are highly connected with modern fields such as Social Network Analysis and Spectral Graph Theory. 

The first author wrote Sections 1-14 and 15.1, and the second author, Dong Zhang, wrote Sections 15.2, 15.3 and 16. 
We would like to thank Sino-Russian mathematical center and especially professor Fan Huijun for assistance and attention to our work. We acknowledge support from National Key R and D Program of China (Grant No. 2020YFE0204200).

\tableofcontents
\newpage

\section{Basic definitions}

\begin{defin}
A set of pairs of any objects is called \bf{graph} $G = (V, E)$, where

\bf{V (vertices, or nodes, or points)} --- the set current of objects,

\bf{E (edges)} --- the set of pairs.
\end{defin}\noindent
\NB In this book the vertex set is considered as finite $V\sub\rN$. Also we denote by \bf{graph $(n,m)$} a graph with $n$ vertices and $m$ edges.

\begin{defin}
A graph is called \bf{(un)directed} \ifof the set of pairs is (un)directed set respectively.
\end{defin}

\begin{ex}
Consider a graph of the cube:\\ \par
$V = \{1, 2, 3, 4, 5, 6, 7, 8\},$\\ \par
$E = \bigl\{(1, 2), (1, 4), (1, 5), (2, 6), (2, 3), \\ 
(3, 7), (3, 4), (5, 6), (5, 8), (6, 7), (7, 8)\bigr\}.\\$

\begin{figure}[h]
\vspace{-125pt}
\hspace{260pt}
\begin{minipage}{8cm}
\hspace{50pt}
	\includegraphics[width=0.5\textwidth]{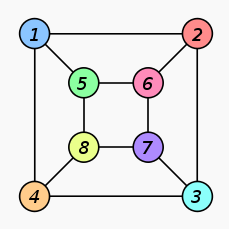}
	\caption{Cube graph representation.}
	\label{fig:sample}
 \end{minipage}
\end{figure}

\end{ex}

\begin{defin}
Two vertices $i, j$ in undirected graph connected with an edge are called \bf{adjacent} and are denoted by $i\sim j$.
\end{defin}

\begin{defin}
    An edge that consists of the same elements is called  \bf{loop}. 
\end{defin}

\begin{defin}
    A subset of edges that consists of the same elements is called \bf{multiple edges (or multi-edge)}.
\end{defin}

\begin{ex}~\par
    $(1, 1)$ --- loop, \par
    $\{(1, 2), (1, 2), (1, 2)\}$ --- multiple edges.
\end{ex}
\NB Now and further we will consider only \bf{simple graphs} --- graphs without loops and multiple edges unless otherwise specified.\\

\NB There is important difference between simple directed graph and simple undirected graph. Be aware of the case $\{(1, 2), (2, 1)\}$. 

\begin{defin}
    Directed graph is called \bf{oriented graph} if there are no two-side edges between any two vertices of the graph.
\end{defin}

\begin{defin}
    \bf{Subgraph} is the part of the graph that is a graph.
\end{defin}

\NB In the definition of graph we don't include \emph{isolated vertices}. This definition can be expanded by adding the set of vertices to edges, but this or similar definitions will give rise many problems further so we omit this case.

\begin{defin}
    \bf{Induced subgraph} on the current subset vertices of a graph is the maximal by inclusion of edges subgraph whose set of vertices is the current set.
\end{defin}

\begin{defin}
    Two graphs are called \bf{isomorphic} if one can re-number the vertices of one graph to obtain another.  
\end{defin}

\begin{figure}[!htbp]
    \centering
	\includegraphics[width=0.8\textwidth]{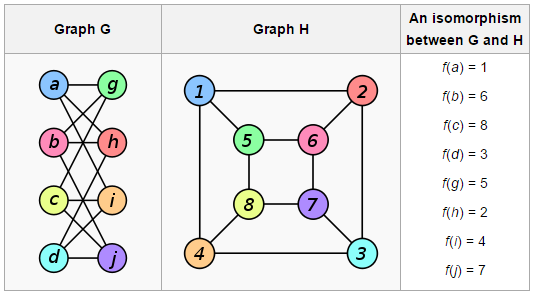}
	\caption{Example of isomorphism.}
\end{figure}

\begin{figure}[!htbp]
    \centering
	\includegraphics[width=0.8\textwidth]{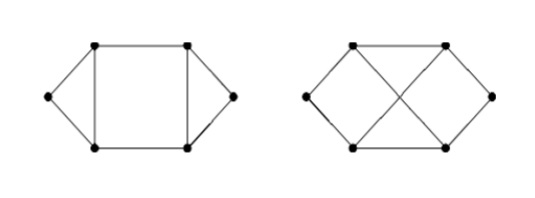}
	\caption{Example of non-isomorphic graphs.}
 \label{fig:nonisog}
\end{figure}

\begin{rk}
To prove that two graphs in figure~\ref{fig:nonisog} are non-isomorphic compare the structures of their induced subgraphs. 
\end{rk}

\begin{figure}[!htbp]
    \centering
	\includegraphics[width=0.8\textwidth]{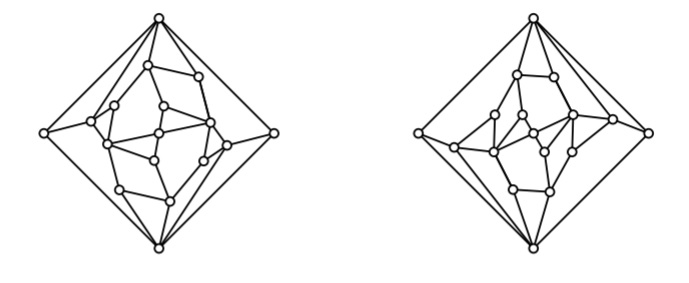}
	\caption{Are these graphs isomorphic or not?}
 \label{fig:isoprob}
\end{figure}

It is simple to prove the graph isomorphism in special cases like in figure~\ref{fig:nonisog} but in general this proof is NP-hardness problem. Try to solve the problem in figure~\ref{fig:isoprob}.\\

\NB Isomorphic graphs are also called the same graphs.

\begin{defin}
    \bf{Degree of vertex} $\mathbf{v_i}\ (\deg{v_i})$ is the number of edges which are coming out from the vertex $v_i$.
\end{defin}

\begin{defin}
    A graph is called \bf{regular} \ifof all vertices have equal degrees.
\end{defin}

In different books and articles definition of degree may vary: it can be defined by the number of in-edges and also by total edges (sum of in- and out-edges), but here we use the out-edges definition. 

\begin{defin}
    A vertex is called \bf{even (odd)} if it has even (odd) degree.
\end{defin}

\begin{lm} (Handshaking lemma). 
$$\sum_{v_i\in V}{\deg{v_i}} = 2 \|E\|.$$
\end{lm}
\begin{proof}
    Note that in $\sum_{v_i\in V}{\deg{v_i}}$ every edge occurs 2 times.
\end{proof}

\begin{ex}
    Consider the bit representation of natural numbers \\
    from 0 to $2^n-1$. Let's Define the graph of n-dimensional cube by \\
    following:
    \begin{enumerate}
        \item $V = \{\text{i for i from 0 to } 2^n-1\}$,
        \item  $E = \{\text{(i,j) iff. i and j differ only in one bit}\}.$
    \end{enumerate}
\emph{How many edges have n-dimensional cube for any n?}
\end{ex}

The simplest cases $n = 1, 2, 3$ are well-known: it is segment, \\
square and cube with 1, 4 and 12 edges respectively.\\
For another cases let's use Handshaking lemma: \\\par
$2\|E\| = \sum_{v_i\in V}{\deg{v_i}} = \|V\| \deg{v_i} = n\,2^n$.\\\\
Thus $\|E\| = n\,2^{n-1}$.

\begin{figure}[h]
\vspace{-245pt}
\hspace{273pt}
\begin{minipage}{8cm}
\hspace{60pt}
	\includegraphics[width=0.5\textwidth]{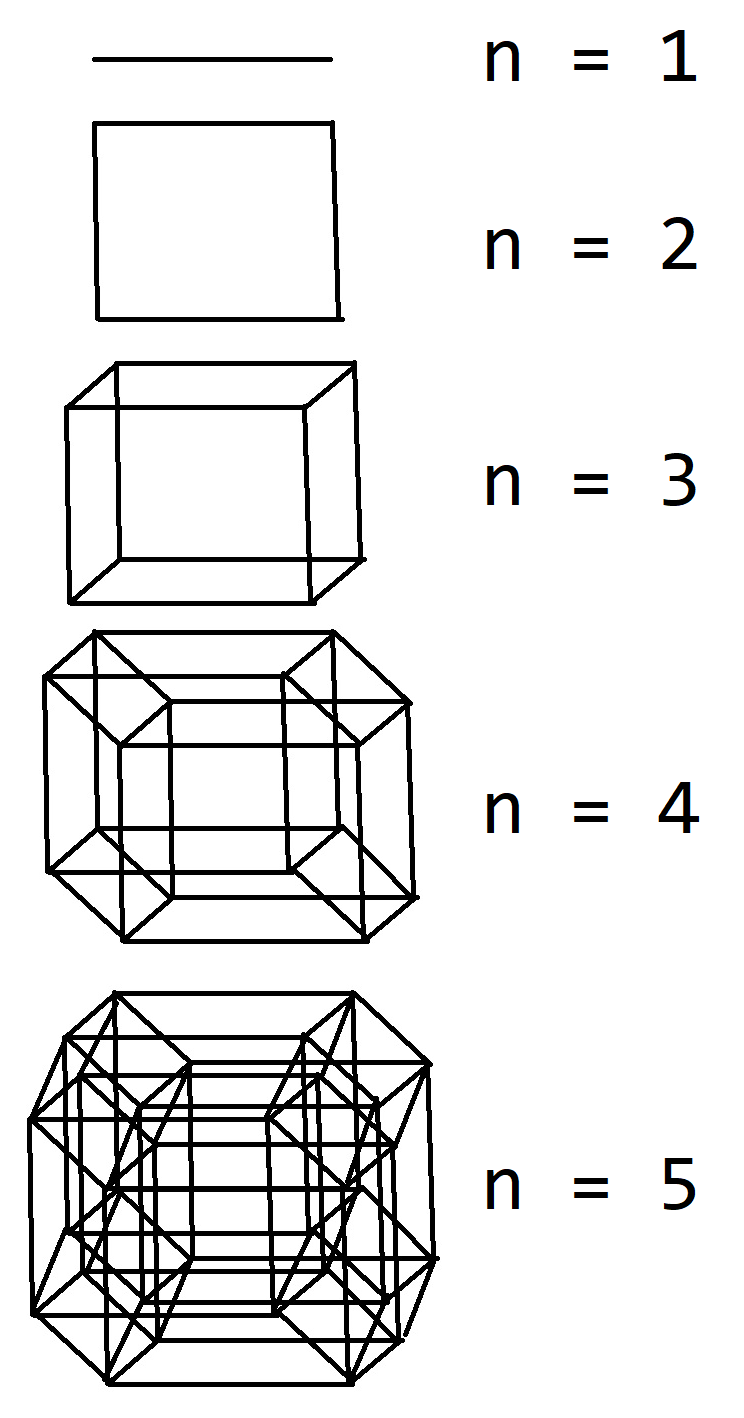}
	\caption{Graphs of hypercubes.}
	\label{fig:sample1}
 \end{minipage}
\end{figure}

\begin{defin}
    \bf{Path} is the sequence of edges in which ending and beginning vertices of consecutive edges are coincide.    
\end{defin}

\begin{defin}
    Graph is called \bf{weighted} if each edge $e_{ij}$ corresponds to a number $w_{ij}\in \rR$.
\end{defin}

One can consider edges with zero weight but this correspondence may lead to many problems in further definitions and algorithms e.g. adjacency matrix, incidence matrix and so on thus we suppose that $w_{ij}\not = 0$ for $i,j\in V$. Also sometimes we will denote a vertex $v_i$ by $i$ for short.


\section{Two matrices corresponded to a graph}

\subsection{Adjacency matrix}

\begin{defin}
    \bf{Adjacency matrix} $A(G)$ for unweighted directed graph $G$ is defined by following:
    \begin{equation*}
  A(G) = \{a_{ij}\}:=
    \begin{cases}
      1 & \text{if there exists edge from vertex $v_i$ to vertex $v_j$,}\\
      0 & \text{otherwise.}
    \end{cases}      
    \end{equation*}
\end{defin} 

\NB For weighted directed graph adjacency matrix (or \emph{weight adjacency matrix}) can be defined in the same way by replacing 1's by a weight $w_{ij}$ for each edge $e_{ij}$.\\ 

\begin{figure}[H]
\vspace{-12pt}
    \centering
	\includegraphics[width=0.83\textwidth]{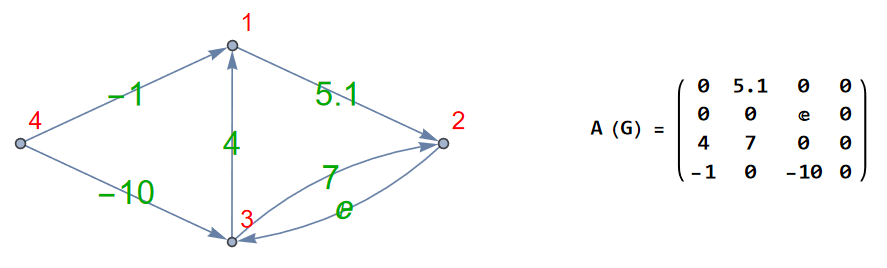}
	\caption{Example of a directed graph and its weight adjacency matrix.}
 \label{isoprob}
\end{figure}

\noindent
\emph{Properties of adjacency matrix:}
\begin{enumerate}
    \item Adjacency matrix is symmetrical for undirected graph,
    \item Adjacency matrix diagonal elements equal to 0 (no loops),
    \item Adjacency matrix always square matrix,
    \item\label{AGdeg} For unweighted graph adjacency matrix consists of 0 and 1 and
    $$\deg{v_i} = \{\text{the sum elements in $i$ row of adjacency matrix}\},$$
    \item  \textquote{Adjacence} correspondence between weighted directed graphs and matrices with properties 2-3 is one-to-one if 0 weights are excluded.
\end{enumerate}

\begin{thm}
    The number of paths from vertex $v_i$ to $v_j$ of length $k$ is equal to $(A(G)^k)_{ij}$ for any unweighted graph G.
\end{thm}
\begin{proof}
Let's consider the induction by the length $k$:
\begin{enumerate}
    \item The number of paths from vertex $v_i$ to $v_j$ of the length 1 corresponds to existence of the edge $e_{ij}$.
    \begin{equation*}
  \{a_{rj}\} =
    \begin{cases}
      1 & \text{if there is edge from vertex $v_r$ to $v_j$,}\\
      0 & \text{otherwise.}
    \end{cases}       
    \end{equation*}
    \item 
    Suppose the induction statement holds for any given $k$, then
    \begin{gather*}
    (A^{k+1})_{ij} = (A^{k})_{ir} a_{rj} =
    \sum_{v_r\in V}{\{
    \text{Number of paths from vertex $v_i$ to $v_r$ with length $k$\}*}}\\
    \text{\qquad\qquad\qquad\qquad\qquad\qquad\qquad*\{1 if there is edge from $v_r$ to $v_j$, 0 --- overwise\} = }\\
    \hspace{-120pt}
    \text{= \{Number of paths from vertex $v_i$ to $v_j$ of length $k+1$\}.}
    \end{gather*}
\end{enumerate}

\end{proof}

\subsection{Incidence matrix}

Consider unweighted directed graph $G$ with some numeration of edges.

\begin{defin}
    \bf{Incidence matrix} for unweighted directed graph $G$ and its edges numeration is defined by following:
    \begin{equation*}
    B(G) = \{b_{ij}\} =
    \begin{cases}
      1 & \text{if edge $j$ starts from vertex $v_i$,}\\
      -1 & \text{if edge $j$ ends in vertex $v_i$,}\\
      0 & \text{otherwise.}
    \end{cases}.       
    \end{equation*}
\end{defin}

\NB For unweighted undirected graph incidence matrix can be defined equivalently, but instead of all -1's will be 1's.\\ 

\NB To generalize incidence matrix definition for weighted graphs (or weighted incidence matrix) multiply each column by weight of corresponding edge.

\begin{figure}[H]
\vspace{-12pt}
    \centering
	\includegraphics[width=0.85\textwidth]{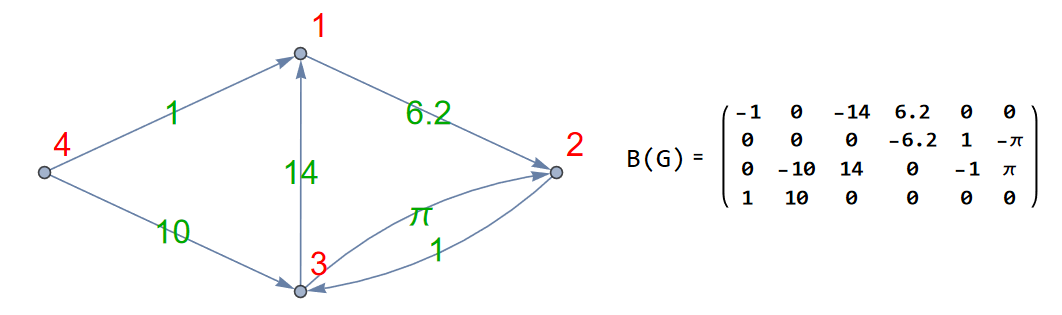}
	\caption{Example of graph and its weight incidence matrix.}
 \label{isoprob2}
\end{figure}

\noindent
\emph{Properties of incidence matrix:}
\begin{enumerate}
    \item In each column of incidence matrix there are only two non-zero elements,
    \item Incidence matrix is rectangular matrix with dimensions $\|V\|\times \|E\|$,
    \item The sum of elements in each incidence matrix column equals to 0 for directed graphs,\label{prop:incsum}
    \item For unweighted undirected graph incidence matrix consists of 0 and 1 and 
    $$\deg{v_i} = \{\text{the sum of elements in $i$ row of incidence matrix}\},$$\label{prop:incdeg}
    \item\textquote{Incidence} correspondence between weighted directed graphs and matrices with properties 2-4 is one-to-one up to column permutations if non-positive weights are excluded. 
\end{enumerate}


\section{Connectivity and trees}

\begin{defin}
    Graph is called \bf{connected} \ifof there exists path between any two vertices.
\end{defin}

\begin{defin}
    \bf{Connected component of a graph} is connected subgraph that is not part of any larger connected subgraph.
\end{defin}

For directed graph there exist several connectivity definitions: \emph{weakly connectivity} (if replacing of directed edges to undirected edges produces a connected undirected graph), \emph{semiconnectivity} (if there exists one-side path between any two vertices) and \emph{strong connectivity} (if there exists path between any two vertices from one to another and vise versa). For each of these connectivity one can define corresponded connected components, but we will consider properties and prove theorems only for undirected case here and in the next section. \\

Let's consider some simplest necessary and sufficient conditions of connectivity. 

\begin{lm}[Odd vertices lemma]
    The number of odd vertices in each connected component of undirected graph is even. 
\end{lm}
\begin{proof}
    Use Handshaking lemma for each connected component.
\end{proof}

\begin{thm} (Connectivity necessary condition).
    Consider an undirected graph $G(n, m)$. If $m > C_{n-1}^2$, then graph G is connected.
\end{thm}
\begin{proof}
    Assume the contrary: there exist greater or equal than two connected components. Denote one component by $G_1$ and all the rest by $G_2$. Hence, $G = G_1\sqcup G_2$.\\
    Consider two cases:
    \begin{enumerate}
        \item Let $\|V(G_1)\| = n-1, \|V(G_2)\| = 1 \imp \|E(G)\| = \|E(G_1)\| \leq C_{n-1}^2$ --- the maximum number of edges in undirected graph with $n-1$ vertices.\\ 
        The same for $\|V(G_1)\| = 1, \|V(G_2)\| = n-1$.
        \item Let $2 \leq \|V(G_1)\| = m \leq n-2 \imp $
        $$\|E(G)\| = \|E(G_1)\| + \|E(G_2)\|\leq C_m^2 + C_{n-m}^2 = \frac{2m^2+n^2-2nm-n}{2} \vee C_{n-1}^2 = \frac{n^2-3 n+2}{2}$$
        $$n-1+m^2-nm\vee 0$$
        $$(m-1)(m-n+1) \leq 0$$
        This gives a contradiction, since $\|E(G)\| > C_{n-1}^2$.
    \end{enumerate}
\end{proof}

\begin{defin}
    \bf{Cycle} is a path consisted from distinct edges where the first and last vertices coincide. 
\end{defin}

\begin{defin}
    \bf{Simple path} is a path where any edges and vertices are distinct except the first and the last vertices.
\end{defin}

\begin{defin}
    \bf{Simple cycle} is a simple path where the first and last vertices coincide.
\end{defin}

\begin{lm}\label{lm:cycle}
    A graph contains a cycle iff. it contains a simple cycle. 
\end{lm}
\begin{proof}
\begin{enumerate}
    \item If a graph contains a simple cycle then it obviously contains a cycle.
    \item Consider a graph that contains a cycle. Let's construct a simple cycle: start from any vertex of the cycle and walk through the cycle and delete visited edges. Since the cycle is finite, some visited vertex will be reached after several steps. The cycle corresponded to this walk to the first visited vertex is simple cycle.
\end{enumerate}
\end{proof}

\begin{defin}
    \bf{Tree} is undirected connected acyclic (without cycles) graph.
\end{defin}

Another definition of tree holds from the lemma~\ref{lm:cycle}: \emph{tree is undirected connected graph without simple cycles}. \\

\begin{defin}
    \bf{Polytree} is an oriented graph whose underling undirected graph is a tree.
\end{defin}

\begin{defin}
    \bf{Tree leaf} is a vertex of degree 1.
\end{defin}

\begin{defin}
    \bf{Tree root} is any highlighted vertex of the graph.
\end{defin}

\NB The root is any highlighted vertex of the graph, however the leafs, in general, are not highlighted as roots. 

\begin{defin}
    \bf{Arborescence} of a directed graph is the polytree that consist of all vertices of the graph that can be reached from the root.
\end{defin}

\begin{lm}[Two leafs lemma]\label{lm:2leafs}
    Every tree contains two leafs.
\end{lm}
\begin{proof}
    The method is the same as in previous lemma. Let's start walking from any vertex and deleting visited edges. Since the graph is finite, a dead end will be reached after several steps. The end corresponds to first leaf. Starting this procedure again from the leaf will provide the second.
\end{proof}

\begin{lm}[Uniqueness path lemma]\label{lm:upath}
    There exists unique path between any two vertices of tree. 
\end{lm}
\begin{proof}
    Hint: assume the contrary and find cycle.
\end{proof}


\section{Spanning trees and algorithms}

\subsection{Spanning trees of unweighted graphs}

\begin{defin}
    \bf{Spanning tree} of undirected graph is a tree that contains all vertices of a graph.
\end{defin}

\NB Spanning trees exist only for undirected connected graphs. For not connected graph there is another definition --- \emph{spanning forest} --- the collection of spanning trees corresponded to each connected component.\\\\
Consider the base algorithms for constructing spanning trees:
\begin{enumerate}
    \item \bf{Depth-first search (DFS)} is the algorithm for searching elements in tree-structured data. Here we discuss it from constructing spanning trees point of view. 

    \emph{Description:}
    \begin{enumerate}
    \item Start from the root. The root is assigned first number.
    \item Walk through non-assigned vertices and assign visited vertices by consecutive numbers (from second) as long as you can continue. \label{enum:step12}
    \item  If it is dead end, return to the previous vertex and do the step~\ref{enum:step12}.
    \end{enumerate}
    The assigned numbers form \bf{depth-first order}.\\\\
    To construct spanning tree, you should save edges on the step~\ref{enum:step12} from previous assigned vertex to the new one that has to be assigned (see figure~\ref{fig:dbalg}).

    \item \bf{Breadth-first search (BFS)} is also the algorithm for searching elements in tree-structured data. Here we discuss it from constructing spanning trees point of view.

    \emph{Description:}
    \begin{enumerate}
    \item Start from the root. The root is assigned first number.
    \item Assign all non-assigned adjacent vertices by consecutive numbers (from second). \label{enum:step22}
    \item  Go to the next (up to assigning) not visited vertex and do the step~\ref{enum:step22}.
    \end{enumerate}
    The assigned numbers form \bf{breadth-first order}.\\\\
    To construct spanning tree, you should save edges on the step~\ref{enum:step22} from the current vertex to new non-assigned vertices which have to be assigned (see figure~\ref{fig:dbalg}).
    
\end{enumerate}

\NB It is easy to see, that the complexity of these algorithms equals to $\mathbf{O\bigl(\|V\|+\|E\|\bigr)}$.

\begin{figure}[H]
\vspace{-12pt}
    \centering
	\includegraphics[width=1\textwidth]{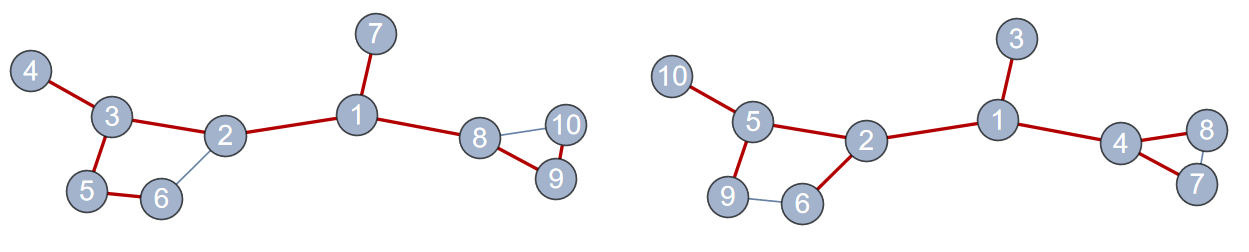}
	\caption{Depth-first search (left) and breadth-first search (right) algorithms, orders and corresponding spanning trees (red).}
 \label{fig:dbalg}
\end{figure}

In general situation depth-first search and breadth-first search algorithms produce different spanning trees. Spanning tree from depth-first search algorithm is called \emph{Trémaux tree}, while from breadth-first search just breadth-first tree. Also there is a freedom to choose a walking direction and thus different directions produce different spanning trees as well.

For not connected graphs these algorithms also can be used. In this case they will construct spanning trees for the connected component of the root. By changing root between components one can construct spanning trees for each component so-called \emph{spanning forest}.\\

\NB These algorithms can by applied for checking connectivity condition of the graph.

\begin{thm}[Necessary and sufficient tree condition 1]\label{thm:nsctree1}
    A graph with $n$ vertices is tree iff. it is connected and consists of $n-1$ edges.
\end{thm}
\begin{proof}\ 
    \begin{enumerate}
        \item \emph{Necessarity.} Consider the induction by the number of vertices: 
        \begin{enumerate}
            \item For graph $(1, 0)$ holds.
            \item Let for every tree with $n$ vertices the condition holds. Consider tree $T$ with $n+1$ vertices. This tree contains two leafs by two leafs lemma~\ref{lm:2leafs}. By deleting one of these leafs the graph will become a tree $T'$ with $n$ vertices and edges decreased by 1. The tree $T'$ consists of $n-1$ edges by inductive hypothesis, therefore $T$ consists of $n$ edges.
        \end{enumerate}
        \item \emph{Sufficiency.} Construct the spanning tree of the graph by one of algorithms above. The spanning tree consists of $n-1$ edges.
    \end{enumerate}
\end{proof}

\begin{thm}[Necessary and sufficient tree condition 2]\label{thm:nsctree2}
    A graph with $n$ vertices is tree iff. it is acyclic and consists of $n-1$ edges.
\end{thm}
\begin{proof}\ 
    \begin{enumerate}
        \item \emph{Necessarity.} Holds from theorem~\ref{thm:nsctree1}.
        \item \emph{Sufficiency.} Let the number of connected components be $k$.
        Construct the spanning tree for every component. Each component is tree, therefore the total number of edges in the graph equals $n-k$ ($-1$ for every component)$\imp k = 1$ and the graph is connected.
    \end{enumerate}
\end{proof}

\NB By these theorems, the complexity of depth-first search and breadth-first search algorithms for connected graphs can be rewritten as $\mathbf{O\bigl(\|V\|+\|E\|\bigr)} = \mathbf{O\bigl(\|E\|\bigr)}$.

\subsection{Minimal spanning trees of weighted graphs}

For weighted graph there exists another type of spanning tree:
\begin{defin}
    \bf{Minimal spanning tree} is the spanning tree with minimum total edge weight.
\end{defin}

Consider two basic algorithms for constructing minimal spanning tree. These algorithms belong to large area called \emph{greedy algorithms} (in which making the locally optimal choice at each stage leads to  globally optimal solution).

\begin{enumerate}
    \item \bf{Prim's algorithm.} \\
    \emph{Description:}
    \begin{enumerate}
        \item Start from arbitrarily vertex. This vertex is the tree $T$.
        \item Let the $K(T)$ --- edges $\not\in T$ adjacent to any vertex $\in T$.
        Add to the tree $T$ new edge with minimum weight in $K(T)$ that is not produce the cycle in $T$ and vertices incident to this edge.\label{enum:p2}
        \item Do \ref{enum:p2} step $\|V\|-1$ times.
    \end{enumerate}
    
    \item \bf{Kruskal's algorithm.} \\
    \emph{Description:}
    \begin{enumerate}
        \item Start from arbitrarily vertex. This vertex is the tree $T$.
        \item Let the $K(T)$ --- all edges $\not\in T$.
        Add to the tree $T$ new edge with minimum weight in $K(T)$ that is not produce the cycle in $T$ and vertices incident to this edge.\label{enum:k2}
        \item Do \ref{enum:k2} step $\|V\|-1$ times.
    \end{enumerate}
\end{enumerate}

\NB The Prim's and Kruskal's algorithms complexity is $\mathbf{O\bigl(\|E\|\,log(\|V\|)\bigr)}$.\\

Minimal spanning trees from Prim's and Kruskal's algorithms depend of the first selected vertex and also of edges from the equal weights subset.  

\begin{rk}
    The proof of correctness Prim's and Kruskal's algorithms holds from theorem~\ref{thm:nsctree1} and theorem~\ref{thm:nsctree2} respectively or from greedy property. 
\end{rk}

Each of these algorithms has its pros and cons: Prim's algorithm produces connected graph for each step but binary heap and adjacency list should be used to reach $O\Bigl(\|E\|\,\log\bigl(\|V\|\bigr)\Bigr)$ complexity. Kruskal's algorithm reach this complexity by using simpler structures. 


\section{Laplacian matrix and Kirchhoff's theorem}

As in previous two sections let's consider undirected unweighted graph $G$.

\begin{defin}
    \bf{Laplacian matrix (also called Kirchhoff matrix)} for undirected unweighted graph $G$ is defined by following:
    \begin{equation*}
    L(G) = \{l_{ij}\} =
    \begin{cases}
      \deg{v_i} & \text{if $i = j$,}\\
      -1 & \text{if there exists edge from vertex $v_i$ to vertex $v_j$,}\\
      0 & \text{otherwise.}
    \end{cases}       
    \end{equation*}
\end{defin}

Or equivalently:
$$L(G) = \mathrm{Diag}\{\deg{v_i}\,|\,v_i\in V\} - A(G).$$

\begin{figure}[H]
\vspace{-12pt}
    \centering
	\includegraphics[width=0.9\textwidth]{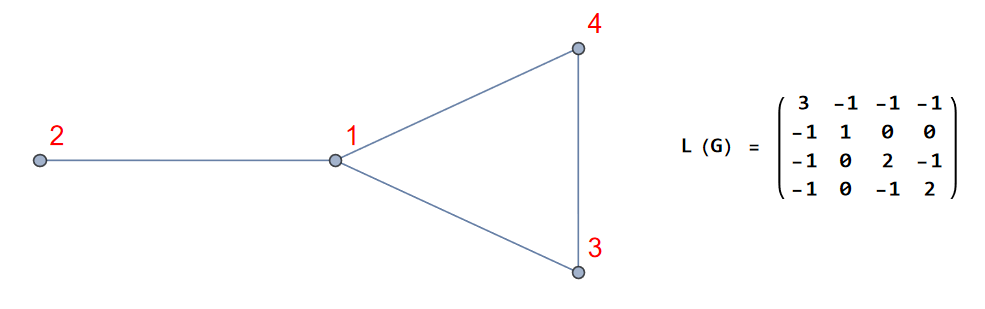}
	\caption{Example of a graph and its laplacian matrix.}
 \label{isoprob1}
\end{figure}

\noindent
\emph{Properties of Laplacian matrix:}
\begin{enumerate}
    \item Laplacian matrix non-diagonal part consists of 0 and -1 elements, 
    \item Laplacian matrix is symmetrical matrix,
    \item Laplacian matrix is square matrix,
    \item The sum of elements in each Laplacian matrix column (and row) equals to 0,\label{prop:sum} 
    \item Laplacian matrix is degenerate matrix,
    \item  \textquote{Laplacian} correspondence between undirected unweighted graphs and matrices with properties 1-4 is one-to-one.
\end{enumerate}

\begin{lm}[Laplacian matrix algebraic complements lemma]\label{lm:lapcomp}
    Algebraic complements of all Laplacian matrix elements are equal.
\end{lm}
\begin{proof}
    Denote the number of vertices of a graph $G$ by $n$ and the algebraic complement to the element $l_{ij}$ of Laplacian matrix by $L_{ij}$. Consider two cases:
    \begin{enumerate}
        \item Let $\rank\bigl(L(G)\bigr) < n-1 \imp L_{ij} = 0, \text{for any } i,j$.
        \item $\rank\bigl(L(G)\bigr) = n-1$. Let's denote the column of all 1's by $\j$, the matrix of complements $\{L_{ij}\}$ by $M$ and the columns of $M^T$ by $\{m_j\}$. Hence, $L\j = 0$ due to Laplacian matrix property~\ref{prop:sum} and
        $$L M^T = \det(L)\cdot I = 0 \imp L m_j = 0, \text{for } j = 1, 2,...,n.$$\ \ \ \ \ 
        Consider a linear equation $L x = 0$. Since the $\rank\bigl(L(G)\bigr) = n-1$, solutions of the equation form a one-dimensional linear space. The vector $\j$ is a solution of the equation, hence all other solutions ($m_1, m_2, ..., m_n$) are proportional to $\j$ and therefore consist of equal elements. From the symmetry property of the Laplacian matrix the proposition holds.
    \end{enumerate} 
\end{proof}

\begin{lm}[Laplacian and incidence matrix lemma]\label{lm:lapinc}
    $L = B\,B^T$, where $B$ is the incidence matrix for any orientation of the graph and numeration of the edges.
\end{lm}
\begin{proof}
    Consider any orientation and numeration of the edges of undirected graph and let $B$ be its corresponding incidence matrix. Denote the row $i$ of the matrix $B$ by $b_i$ and the number of vertices of the graph by $n$.
    $$B\, B^T = \{m_{ij}\} = \{b_i\cdot b_j\}.$$
    \begin{enumerate}
        \item For $i = j$, $m_{ii} = b_i^2 = \{\text{the number of non-zero elements in } b_i\} =$\\$= \{\text{the sum of elements in $i$ row in the incidence matrix for undirected graph}\} = \deg{v_i}$, due to incidence matrix property~\ref{prop:incdeg}.
        \item For $i \neq j$, 
    \begin{equation*}
    b_{ik}\ b_{jk} =
    \begin{cases}
      -1 & \text{if $v_i$ adjacent to $v_j$ by the edge $k$,}\\
      0 & \text{otherwise.}
    \end{cases}.       
    \end{equation*} 
    Since the graph has no loops,
        $$m_{ij} =  b_i\cdot b_j =\sum_{k=1}^n {b_{ik}b_{jk}} = l_{ij}.$$
    \end{enumerate}
\end{proof}

\begin{lm}[Graph $(n, n-1)$ lemma]\label{lm:nnm1}
    Let $B(G)$ --- the incidence matrix for some orientation and numeration edges of the graph $G(n, n-1)$ and $D(B)$ be maximal principal minor of the matrix $B$. It follows, that
    \begin{enumerate}
        \item $D(B) = \pm1$, if $G$ is connected (tree),
        \item $D(B) = 0$, if $G$ is not connected (not tree).
    \end{enumerate}
\end{lm}
\begin{proof}
\begin{enumerate}
    \item Let the graph $G$ be connected. The graph $G$ is tree by the theorem~\ref{thm:nsctree1} and therefore it has at least two leafs by lemma~\ref{lm:2leafs}.

    For the graph $G(n, n-1)$ the matrix $B(G)$ is also $n\times (n-1)$ matrix, hence the maximal principal minor is the determinant of matrix $B(G)$ without some row. Suppose that this row corresponds to last vertex $v_n$ without loss of generality (if it is not re-number vertices). 
    \begin{enumerate}
        \item\label{lm:nn-1} Consider the leaf $\neq v_n$ and the edge adjacent to the leaf. Re-number vertices and edges such that this leaf becomes first vertex $v_1$ and its adjacent edge becomes the first edge $e_1$. This re-numbering procedure corresponds to row and column permutations in the matrix $B(G)$ and in the determinant $D(B)$, therefore $D(B)$ is not changing. By this re-numbering the first row of the matrix $B(G)$ becomes equal to $(\pm 1, 0, 0, ... , 0)$.
        \item Consider the induced subgraph on $V(G)\setminus \{v_1\}$ vertices and do the same procedure as in~\ref{lm:nn-1} for this subgraph. Thus, the corresponding vertex and the edge become the second vertex and the second edge and the second row of the matrix $B(G)$ become to look like $(*, \pm 1, 0, 0, ... , 0)$, where \textquote{*} means some number. Let's continue this procedure.
        \item \ \\ In the end, the matrix $B(G)$ become to look like
        $\begin{pmatrix}
        \pm 1 & 0 & 0 &... & 0 \\
        * & \pm 1 & 0 &... & 0 \\
        * & * & \pm 1 &... & 0 \\
        & & ... & & \\
        * & * & * &... & \pm 1 \\
        * & * & * &... & *
        \end{pmatrix}$, and $D(B)$ will be the determinant of the upper part of the matrix. Therefore, $D(B) = \pm 1$.
    \end{enumerate}
    \item  Suppose the graph $G$ is not connected. Denote by $G_1$ the connected component without vertex $v_n$ and $G_2$ all the rest. 
    
    Consider the sub-matrix $K\sub B(G)$ corresponding to vertices of $G_1$. The incidence matrix $B(G_1)$ is the sub-matrix of the matrix $K$. The sum in each column of the matrix $K$ equals to 0. Indeed, if edge belongs to $G_1$ then the sum equals to 0 by property~\ref{prop:incsum} of incidence matrix $B(G_1)$ and if edge belongs to $G_2$ then all elements in corresponding column equals 0. Hence, the sum of rows of the matrix $K$ equals 0. Since $K$ is the part of the matrix corresponding to $D(B)$, the determinant $D(B) = 0$.
\end{enumerate}
\end{proof}

\begin{thm}[Kirchhoff matrix tree theorem]
    The number of trees in undirected unweighted connected graph $G$ equals to an algebraic complement of any Laplacian matrix element. 
\end{thm}
\begin{proof}
    From Lemma~\ref{lm:lapinc}: $L(G) = B\,B^T$. Let's denote the number of vertices of graph $G$ by $n$. Algebraic complements of all  Laplacian matrix elements are equal by the lemma~\ref{lm:lapcomp}. Consider the algebraic complement $L_{nn}$ of the last bottom-right element $l_{nn}$. Let's denote $\hat{B}$ the matrix $B(G)$ without last row. Therefore, $L_{nn} = \hat{B}\,\hat{B}^T$.

    Since the graph is connected, the number of edges is $\geq n-1$ and hence the principal minor $D(\hat{B})$ of the matrix $\hat{B}$ has $(n-1)\times(n-1)$ dimensions. Let $T(B)$ be the sub-matrix of $B$ consisted of the columns corresponding to $D(\hat{B})$ (see Figure~\ref{fig:lbb}). The sub-matrix $T(B)$ is obtained by adding to $D(\hat{B})$ last row and hence has $n\times(n-1)$ dimensions. Therefore, the sub-matrix $T(B)$ corresponds to a subgraph $\hat{G} (n, n-1)\sub G$. Moreover, this matrix equals to incidence matrix of the subgraph $B(\hat{G})$ (to prove this part more accurately, one should either consider the definition of subgraph with isolated vertices or prove this part in terms of sub-matrices). By the lemma~\ref{lm:nnm1}: 
    \begin{equation*}
    D\bigl(T(B)\bigr) = D\bigl(B(\hat{G})\bigr) = 
    \begin{cases}
      \pm 1 & \text{if $\hat{G}$ is tree,}\\
      0 & \text{if it is not,}
    \end{cases},       
    \end{equation*} 
    and it is also equal to $D(\hat{B})$.

\begin{figure}[H]
\vspace{-12pt}
    \centering
	\includegraphics[width=0.7\textwidth]{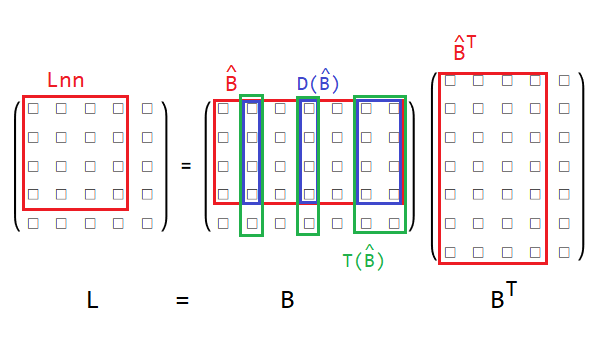}
	\caption{Designations in the Kirchhoff matrix theorem proof.}
\label{fig:lbb}
\end{figure}
    
    The last step is to use a formula from algebra:
\begin{lm*}[Binet-Cauchy formula]
    Consider the matrix equation:
    $$\begin{matrix}
        P & \times & Q & = & C \\
        s\times t & &t\times s & & s\times s
    \end{matrix}$$
    for two rectangular matrices $P$ and $Q$ of $s\times t$ and $t\times s$ dimensions and square $s\times s$ matrix $C$ for any $s, t\in \rN$.
    
    Let the columns of matrix $P$ be called \textquote{corresponded} to rows of matrix $Q$ iff. they are consist of the same sets of indexes. Denote the minors consisted of the corresponding columns and rows by $P_{i_1 i_2 ... i_s}$ and $Q_{i_1 i_2 ... i_s}$ respectively, for indexes $i_1, i_2, ... , i_s \in \{1...t\}$ with repetitions.
    
    It follows that
    $$\det{C} = \sum_{i_1, i_2, ... , i_s = 1}^t {P_{i_1 i_2 ... i_s}\,Q_{i_1 i_2 ... i_s}}$$
\end{lm*}
\begin{proof}
    The complete proof won't be provided here, just the main ideas:
    \begin{enumerate}
        \item For case $s = t$, it is well-known determinant of two matrices product: $\det{(P\,Q)} = \det{P}\det{Q}$.
        \item For case $s>t$, the determinants $P_{i_1 i_2 ... i_s}$ and  $Q_{i_1 i_2 ... i_s} = 0$ for any $i_1 i_2 ... i_s$ because of repetitions and $\rank(P\,Q)\leq\rank(P)\leq t\imp \det{C} = 0$.
        \item For case $s<t$, it follows from two algebraic facts:
        \begin{enumerate}
            \item The coefficient of $z^{t-k}$ in the polynomial $\det{(z I_t + C)}$ is the sum of the $k\times k$ principal minors of $C$ for any $k = 1,2,..,t$, where $I_t$ is the identity matrix of $t\times t$ dimensions. 
            \item If $s\leq t$, then $\det{(z I_t + Q\,P)} = z^{t-s}\det{(z I_s + P\,Q)}$, for matrices $P$ and $Q$ from the statement.
        \end{enumerate}
    \end{enumerate}
\end{proof}

Using the Binet-Cauchy formula for matrices $L_{nn} = \hat{B}\,\hat{B}^T$:
$$
L_{nn} = \sum_{i_1, i_2, ..., i_{n-1}} {\hat{B}_{i_1 i_2 ...  i_{n-1} }\,\hat{B}_{i_1 i_2 ...  i_{n-1} }^T} = \sum_{i_1, i_2, ...,  i_{n-1}} {D_{i_1 i_2 ...  i_{n-1}}^2(\hat{B})} = \sum_{\hat{G}(n,n-1)\sub G} {D^2\bigl(B(\hat{G})\bigr)},
$$
and by the lemma~\ref{lm:nnm1} it equals to the number of spanning trees of $G$.  

\end{proof}

\begin{defin}
    \textbf{The complete graph} $K_n$ is the graph where any two vertices are connected with an edge.
\end{defin}

\begin{thm}[Cayley's formula]\label{thm:cayley}
    The number of spanning trees of the complete graph $K_n$ equals to $n^{n-2}$.
\end{thm}
\begin{proof}
    The Laplacian matrix for the complete graph $K_n$ is
    \begin{equation*}
        L(G) = 
        \begin{pmatrix}
            n-1 & -1 & -1 & ... & -1 \\
            -1 & n-1 & -1 & ... & -1 \\
             & & & ... & \\
             -1 & -1 & ... & -1 & n-1
        \end{pmatrix}, 
    \end{equation*} of dimensions $n\times n$.

    The algebraic complement to the last element $l_{nn}$ is the determinant of the matrix of dimensions $(n-1)\times (n-1)$. Let's find this algebraic complement.
    \begin{equation*}
        L_{nn} = \det
        \begin{pmatrix}
            n-1 & -1 & -1 & ... & -1 \\
            -1 & n-1 & -1 & ... & -1 \\
             & & & ... & \\
             -1 & -1 & ... & -1 & n-1
        \end{pmatrix} \stackrel{\text{adding all rows to first}}{=} \det
        \begin{pmatrix}
            1 & 1 & 1 & ... & 1 \\
            -1 & n-1 & -1 & ... & -1 \\
             & & & ... & \\
             -1 & -1 & ... & -1 & n-1
        \end{pmatrix}=
    \end{equation*}
    \begin{equation*}
    \stackrel{\text{adding first row to others}}{=} \det
        \begin{pmatrix}
            1 & 1 & 1 & ... & 1 \\
            0 & n & 0 & ... & 0 \\
             & & & ... & \\
             0 & 0 & ... & 0 & n
        \end{pmatrix} = n^{n-2}.
    \end{equation*}
\end{proof}


\section{Bipartite graph}

\begin{defin}
    Connected undirected graph is called \bf{bipartite} \ifof one can divide the vertices into two groups such that any two vertices from one group are not adjacent.
\end{defin}

Consider the \bf{bipartite property testing algorithm}:
\begin{enumerate}
    \item Choose any vertex.
    \item Start depth-first or breadth-first algorithms and divide vertices to 0 or 1 groups by putting to them corresponded marks:
    \begin{enumerate}
        \item For depth-first: sequentially alternate marks for depth-first walk.
        \item For breadth-first: put the same marks if the vertices are on the same breadth and change marks otherwise.
    \end{enumerate}
    \item If any two vertices from one group are not adjacent, the graph is bipartite and not bipartite otherwise.
\end{enumerate}

\NB The complexity of this algorithm is the same as for depth-first and breadth-first algorithms for connected graphs and equals to $\mathbf{O\bigl(\|E\|\bigr)}$.

\begin{ex}
    Consider the result of the testing algorithm:
    \begin{figure}[H]\label{fig:bipart}
    \centering
	\includegraphics[width=1\textwidth]{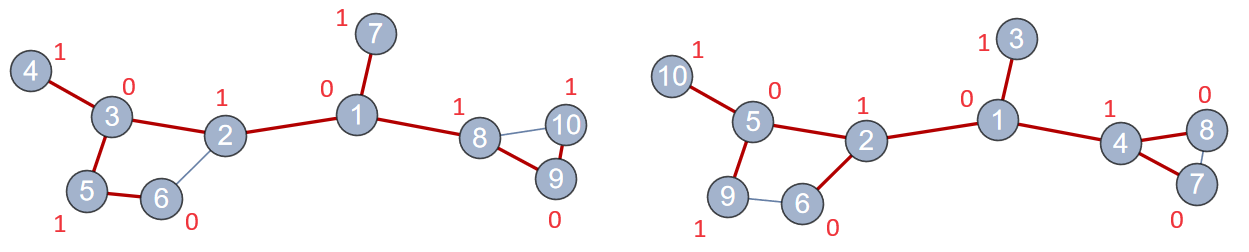}
	\caption{Group division (red marks) by depth-first algorithm (left) and breadth-first algorithm (right).}
\end{figure}
This graph is not bipartite because after the algorithm there are adjacent vertices from the same group.
\end{ex}

\begin{thm}[Necessary and sufficient bipartite condition] Connected undirected graph is bipartite iff. it doesn't contain odd length cycles. 
\end{thm}
\begin{proof}\ 
    \begin{enumerate}
        \item \emph{Necessarity.} Assume the contrary: the graph contains odd length cycle. Let's start to divide vertices from this cycle to two groups by the rule: two vertices from one group should not be adjacent. By this rule vertices will sequentially alternate to each other. Since the length of cycle is odd, the first and the last vertices will be from the same group. This gives a contradiction with bipartite condition.
        \item \emph{Sufficiency.} Consider bipartite property testing algorithm with spanning tree $T$ construction. It is easy to see that any tree is a bipartite graph. Let's start to add remaining edges. Denote first edge by $(v, w)$. 
        
        Assume that vertices $v$ and $w$ corresponds to the same group by the testing algorithm. By the lemma~\ref{lm:upath} there exists unique path from $v$ to $w$ in $T$. Since the marks alternate to each other along this path by the algorithm, this path with the edge $(v, w)$ form an odd length cycle. This gives a contradiction. Hence, all remaining edges connect vertices from different groups.
    \end{enumerate}
\end{proof}

\begin{defin}
    Bipartite graph is called \bf{complete bipartite graph} $K_{n,m}$ \ifof any two vertices from different groups are connected with an edge.
\end{defin}

\begin{thm}[Cayley's formula for bipartite graphs]
    The number of spanning trees of the complete bipartite graph $K_{n, m}$ equals to $n^{m-1} m^{n-1}$.
\end{thm}
\begin{proof}
    The Laplacian matrix for the complete bipartite graph consists of 4 blocks:
    \begin{equation*}
        L(G) = 
        \begin{pmatrix}
            m & 0 & 0 & ... & 0 & -1 & -1 & ... & -1 & -1 \\
            0 & m & 0 & ... & 0 & -1 & -1 & ... & -1 & -1 \\
             &  & & ... & & & & ... & & \\
             0 & 0 & ... & 0 & m & -1 & -1 & ... & -1 & -1 \\
             -1 & -1 & ... & -1 & -1 & n & 0 & 0 & ... & 0 \\
             -1 & -1 & ... & -1 & -1 & 0 & n & 0 & ... & 0 \\
             &  & ... & & & & & & ... & \\
             -1 & -1 & ... & -1 & -1 & 0 & 0 & ... & 0 & n  \\
        \end{pmatrix}. 
    \end{equation*} 
    
    The first upper-left corner with $m$'s has $n\times n$ dimensions and the last lower-right corner with $n$'s has $m\times m$ dimensions.

    The algebraic complement to the last element $l_{nn}$ is the determinant of the matrix of dimensions $(n+m-1)\times (n+m-1)$. Let's find this algebraic complement.
    \begin{equation*}
        L_{nn} = \det
        \begin{pmatrix}
            m & 0 & 0 & ... & 0 & -1 & -1 & ... & -1 & -1 \\
            0 & m & 0 & ... & 0 & -1 & -1 & ... & -1 & -1 \\
             &  & & ... & & & & ... & & \\
             0 & 0 & ... & 0 & m & -1 & -1 & ... & -1 & -1 \\
             -1 & -1 & ... & -1 & -1 & n & 0 & 0 & ... & 0 \\
             -1 & -1 & ... & -1 & -1 & 0 & n & 0 & ... & 0 \\
             &  & ... & & & & & & ... & \\
             -1 & -1 & ... & -1 & -1 & 0 & 0 & ... & 0 & n  \\
        \end{pmatrix} \stackrel{\text{adding all rows to first}}{=}
    \end{equation*}
    \begin{equation*}
     = \det
        \begin{pmatrix}
            1 & 1 & 1 & ... & 1 & 0 & 0 & ... & 0 & 0 \\
            0 & m & 0 & ... & 0 & -1 & -1 & ... & -1 & -1 \\
             &  & & ... & & & & ... & & \\
             0 & 0 & ... & 0 & m & -1 & -1 & ... & -1 & -1 \\
             -1 & -1 & ... & -1 & -1 & n & 0 & 0 & ... & 0 \\
             -1 & -1 & ... & -1 & -1 & 0 & n & 0 & ... & 0 \\
             &  & ... & & & & & & ... & \\
             -1 & -1 & ... & -1 & -1 & 0 & 0 & ... & 0 & n  \\
        \end{pmatrix} \stackrel{\text{adding first row to the last $m-1$'s}}{=}
    \end{equation*}
     \begin{equation*}
     = \det
        \begin{pmatrix}
            1 & 1 & 1 & ... & 1 & 0 & 0 & ... & 0 & 0 \\
            0 & m & 0 & ... & 0 & -1 & -1 & ... & -1 & -1 \\
             &  & & ... & & & & ... & & \\
             0 & 0 & ... & 0 & m & -1 & -1 & ... & -1 & -1 \\
             0 & 0 & ... & 0 & 0 & n & 0 & 0 & ... & 0 \\
             0 & 0 & ... & 0 & 0 & 0 & n & 0 & ... & 0 \\
             &  & ... & & & & & & ... & \\
             0 & 0 & ... & 0 & 0 & 0 & 0 & ... & 0 & n  \\
        \end{pmatrix} = m^{n-1} n^{m-1}.
    \end{equation*}
\end{proof}

There is also more general definition of complete $k$-partite graph. We will consider such graphs in the subsection~\ref{sc:nls}. For these graphs it is easy to reformulate and proove this Cayley's theorem by analogy for more general case.


\section{Shortest path problems}

In this section there will be observed the basic algorithms of searching shortest path between vertices for weighted directed or undirected graphs.

\begin{defin}
    A cycle in weighted graph is called \bf{negative cycle} if the sum of weights along the cycle is negative.
\end{defin}

If the path between two vertices intersects negative cycle, the shortest path between them is not exist. Therefore, for shortest path algorithms we consider only graphs without negative cycles unless otherwise specified. 

These algorithms are very similar for both directed and undirected cases, so let's give combined definition:

\begin{defin}
    Spanning tree of weighted (directed) undirected graph $G$ is called \bf{(shortest path arborescence) shortest path tree} if the path from the root to any other vertex along the (poly)tree is the shortest path in $G$.
\end{defin}

\newpage
\noindent
Consider basic algorithms:
\begin{enumerate}
    \item \bf{Dijkstra's algorithm}.
    
    This algorithm finds the shortest path tree (shortest path arborescence) with selected vertex as a root for (directed) undirected graphs with positive weights.
    
    \emph{Description:} 
    \begin{enumerate}
        \item Initialization: mark all vertices unvisited, initialize two arrays: $\dd[v]$ --- array of minimum distances and $\pp[v]$ --- array of predecessors for shortest paths from the root to vertices $v$ and assign $\dd[v]$ by zero for the selected vertex and infinities for all the rest.
        \item\label{dijkstra} Find the unvisited vertex $v$ with the smallest value in $d$. For each its unvisided neighbour $w$ if the sum of $\dd[v]$ plus weight of the edge $(v, w)$ is less than value $\dd[w]$, change $\dd[w]$ to the sum and change the predecessor $\pp[w]$ for $v$. After checking all unvisited neighbours mark $v$ as visited.
        \item  If there exist unvisited vertices in the graph go to~\ref{dijkstra}. 
    \end{enumerate}

Edges (predecessors and corresponded vertices) form shortest path tree (see figure~\ref{fig:dij}).\\

\NB The complexity of Dijkstra's algorithm equals to $\mathbf{O\bigl(\|E\|+\|V\|^2\bigr) = O\bigl(\|V\|^2\bigr)}$ in the simplest realization. Using binary search tree or binary heap it equals to $\mathbf{O\Bigl(\bigl(\|E\|+\|V\|\bigr)\,log\,\|V\|\Bigr)}$.

\begin{figure}[H]
    \centering
	\includegraphics[width=0.9\textwidth]{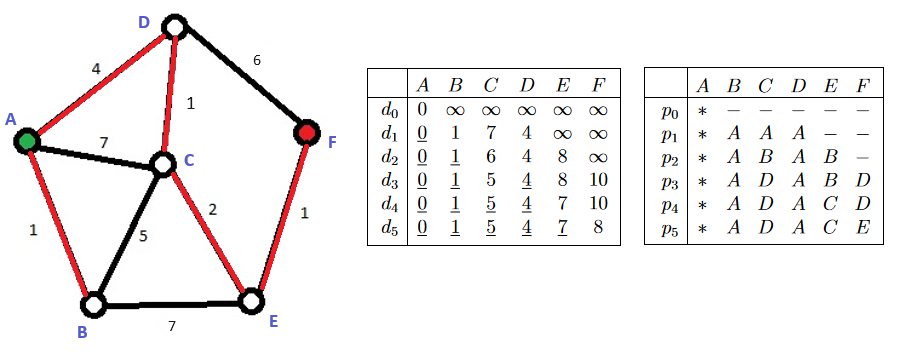}
	\caption{Shortest path tree (left, red) and sequential steps of Dijkstra's algorithm (right). The root vertex of the shortest path tree is denoted by \textquote{*}. Underlined numbers corresponded to visited vertices.}\label{fig:dij}
\end{figure}

\item \bf{Bellman–Ford algorithm.}

    This algorithm is similar to Dijkstra's algorithm. It finds the (shortest path arborescence) shortest path tree with selected vertex as a root for any (directed) undirected graphs.

    \emph{Description:} 
    \begin{enumerate}
        \item Initialize two arrays: $\dd[v]$ --- array of minimum distances and $\pp[v]$ --- array of predecessors for shortest paths from the root to vertices $v$ and assign $\dd[v]$ by zero for the selected vertex and infinities for all the rest.
        \item\label{bellford} For each edge $(v, w)$ if the sum of $\dd[v]$ plus weight of the edge $(v, w)$ is less than value $\dd[w]$, change $\dd[w]$ to the sum and change the predecessor $\pp[w]$ for $v$.
        \item  Do case~\ref{bellford} $\|V\|-1$ times. 
    \end{enumerate}

    Edges (predecessors and corresponded vertices) form shortest path tree similar to Dijkstra's algorithm (see figure~\ref{fig:bf}).\\

\NB The complexity of Bellman–Ford algorithm equals to $\mathbf{O\bigl(\|V\|\,\|E\|\bigr)}$.

\begin{figure}[H]
\vspace{-12pt}
    \centering
	\includegraphics[width=0.9\textwidth]{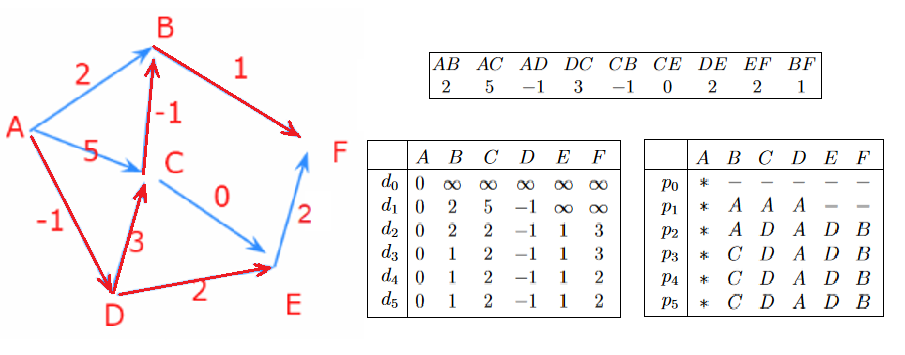}
	\caption{Shortest path arborescence (left, red), the edge set of the graph (right-top) and sequential steps of Bellman–Ford algorithm (right-bottom). The root vertex of the shortest path arborescence is denoted by \textquote{*}.}\label{fig:bf}
\end{figure}

\NB Dijkstra's and Bellman–Ford algorithms are also used for searching shortest path between two selected vertices for directed and undirected weighted graphs (walk along predecessors in reverse order to find the shortest path). Also, Dijkstra's algorithm can be stopped earlier when the final vertex was visited and Bellman–Ford algorithm when the minimum distances array was not changed. \\

Bellman–Ford algorithm also \bf{finds negative cycle} in a graph. If the graph contains negative cycle, the shortest path between any two vertices of the cycle doesn't exist and if one runs the algorithm one more time ($\|V\|$ times in total) the minimum distances array will change. Let's consider this algorithm (see figure~\ref{fig:bfc} for example):

\begin{enumerate}
    \item Do case~\ref{bellford} one more time.
    \item Find the vertex $v$, where $\dd[v]$ was changed.
    \item The negative cycle is obtained by saving vertices from this vertex $v$ along precursors array at previous $\|V\|$-1 step (in reverse order). 
\end{enumerate}

\begin{figure}[H]
\vspace{-12pt}
    \centering
	\includegraphics[width=\textwidth]{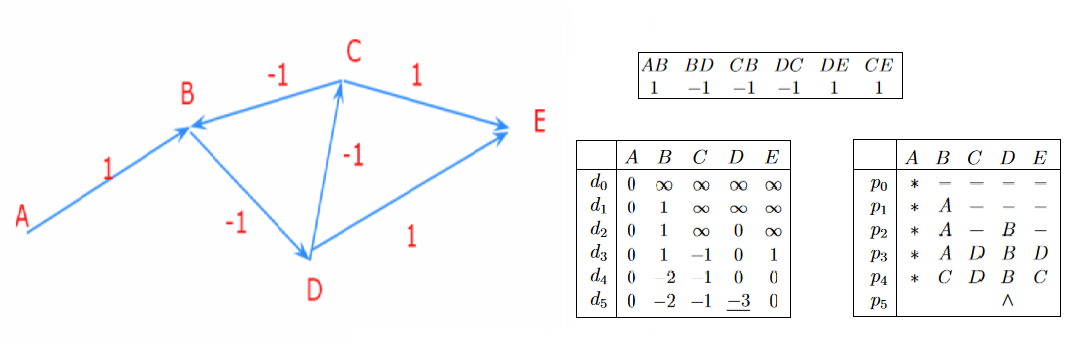}
	\caption{Sequential steps of Bellman–Ford algorithm (right) for the graph (left). The starting vertex is denoted by \textquote{*}. Column (vertex) $D$ with underlined number corresponds to the value that have changed at the additional $\|V\|$ step. The negative cycle $DCBD$ is constructed by walking along precursors set (in reverse order) from this vertex.}\label{fig:bfc}
\end{figure}

\item \bf{Floyd–Warshall algorithm.}

     This algorithm finds the shortest path matrix $\dd[v][w]$ (shortest path between any two vertices) for a graph.

    \emph{Description:} 
    \begin{enumerate}
        \item Initialize shortest path matrix $\dd[v][w]$: for every edge $(v, w)$ the element $\dd[v][w]$ is equal to the weight of $(v, w)$ and $\infty$ otherwise.
        \item\label{flowar} For each vertex $w$ check for each pair of vertices $v, u$ if the sum of $\dd[v][w]$ and $\dd[w][u]$ is less than $\dd[v][u]$. If it is then change $\dd[v][u]$ for this sum. 
    \end{enumerate}

\NB The complexity of Floyd–Warshall algorithm equals to $\mathbf{O\bigl(\|V\|^3\bigr)}$.

\item \bf{Johnson's algorithm.}

     This algorithm finds the shortest path matrix for a graph by using Bellman–Ford shortest path tree and Dijkstra's algorithm for weighted undirected connected graph and for directed graph if there exist paths from the root to any other vertex (spanning arborescence, see figure~\ref{fig:john} for example).

    \emph{Description:} 
    \begin{enumerate}
        \item Use Bellman–Ford algorithm, calculate the minimum distances array $\dd[v]$ and check if the graph contains negative cycle.
        \item\label{john} If the graph doesn't contain negative cycle, correct all weights using Bellman–Ford's values: change the weight of each edge $(v, w)$ to the sum of the weight and $\dd[v]-\dd[w]$.   
        \item For each vertex use Dijkstra's algorithm and calculate  shortest path matrix.
        \item Correct shortest path matrix by reverse transformation: change each $\dd[v][w]$ to the difference between $\dd[v][w]$ and $\dd[v]-\dd[w]$.
    \end{enumerate}

\NB The complexity of Johnson's algorithm equals to $\mathbf{O\Bigl(\|V\|\,\|E\|+ \bigl(\|E\|+\|V\|\bigr)\|V\|\,log\|V\|\Bigr)} = \mathbf{ O\Bigl(\bigl(\|E\|+\|V\|\bigr)\|V\|\,log\|V\|\Bigr)}$.

After changing the weights of the graph at the step~\ref{john} all weights become non-negative with zero-weights along Bellman–Ford's arborescence. Using Dijkstra's algorithm for each vertex it gives zero shortest path along the arborescence.

\begin{figure}[H]
\vspace{-12pt}
    \centering
	\includegraphics[width=0.9\textwidth]{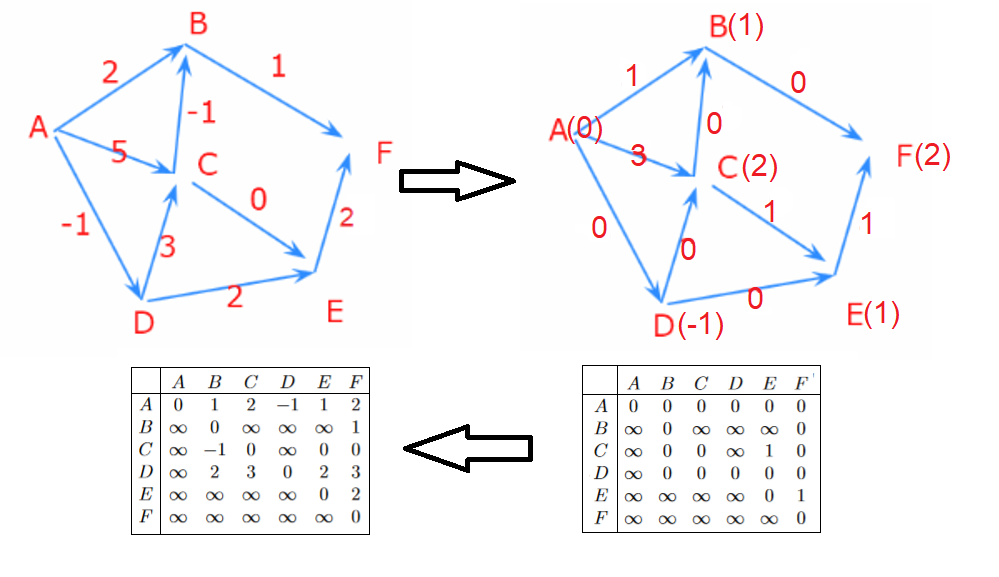}
	\caption{First two steps of Johnson's algorithm (top). Minimum distances lengths $\dd[v]$ for the first step are written near vertices in brackets. Weights of edges were changed at the second step. Shortest path matrix for the graph with new weights by using Dijkstra's algorithm (right-bottom). Result shortest path matrix (left-bottom).}\label{fig:john}
\end{figure}

\end{enumerate}


\section{Planar graph and polyhedrons}
\subsection{Planar graph}

\begin{defin}
    A graph is called \bf{planar} iff. there exists representation of the graph without intersections of edges.
\end{defin}

\begin{ex}
    \begin{enumerate}
        \item Every tree is planar.
        \item $K_2, K_3, K_4$ are planar.
        \item $K_5$ is not planar.
        \item $K_{1,2}, K_{1,3}, K_{2,2}, K_{2,3}$ are planar.
        \item $K_{3,3}$ is not planar.
    \end{enumerate}
\end{ex}
\begin{proof}
    \begin{enumerate}
        \item It is easy to see that tree is planar (use depth-first or breadth-first orders to prove it).
        \item Just try to draw these graphs on plane without intersections.
        \item\label{proof:k5} To prove that $K_5$ is not planar let's use classical theorem:
        \begin{thm*}[Jordan curve]
            Any simple closed curve $C$ divides the plane for two disjoint open subsets and any curve starting from a point of the first subset and ending in a point of the second subset intersect $C$.
        \end{thm*}
        These subsets are denoted by $int(C)$ --- inner subset and $out(C)$ --- outer subset. 
        
        Consider the contrary. The cycle $C = v_1v_2v_3v_1$ is closed curve thus it divides the plane by two parts. Let's suppose without loss of generality that $v_4\in int(C)$. Consider three cycles $C_1 = v_2v_3v_4v_2, C_2 = v_3v_4v_1v_3, C_3 = v_1v_2v_4v_1$. Since $v_i\in out(C_i)$, the assumption that $v_5 \in int(C_i)$ gives contradiction (the edge $(v_5, v_i)$ should intersect $C_i$ by Jordan curve theorem) and therefore, $v_5 \in out(C_i)$ for $i = 1,2,3$. Hence, $v_5 \in out(C) = \bigcap\limits^3_{i = 1} out(C_i)$ and it gives contradiction because $v_4 \in int(C)$.
        \item $K_{1,2}, K_{1,3}$ are trees. For $K_{2,2}, K_{2,3}$ just draw this graph on plane without intersections.
        \item Use the same procedure like in~\ref{proof:k5}.
    \end{enumerate}
\end{proof}

\begin{defin}
    \bf{Inner face} of the planar graph is the inner $int(C)$ part of the plane bounded by a simple cycle $C$ of graph.
\end{defin} 

\begin{defin}
    \bf{Outer face} of the planar graph is the common outer part $out(G)$ for all simple cycles $C_i$ of graph $G: out(G) = \bigcap\limits_{i} out(C_i)$.
\end{defin}

\emph{Notations:} The set of all faces (inner and outer) is denoted by $F$.\\

\NB Any tree has only one outer face and no inner faces.

\begin{thm}[Euler]
    For any planar connected graph holds
    $$\|V\|-\|E\|+\|F\| = 2.$$
\end{thm}
\begin{proof}
    Use the induction by number of faces $k = \|F\|$:
    \begin{enumerate}
        \item If graph contains only one face ($k = 1$), this face should be outer face. If there are no inner faces then there are no cycles in the graph and therefore graph is tree. For tree $\|V\| = n, \|E\| = n-1$ and the condition holds.
        \item Let for any planar connected graph with $k$ faces condition holds. Consider a planar connected graph with $k+1$ faces. Since the graph is planar, there exists an edge (without intersections with any other edges) that separates the outer face and some inner face. If one delete this edge the number of faces will decrease at 1 and hence the formula will be satisfied. Therefore, the condition holds for the initial graph also.
    \end{enumerate}
\end{proof}

\begin{cor}
    Let $c$ be the number of connected components of a planar graph. 
    Then the formula holds:
    $$\|V\|-\|E\|+\|F\| = 1+c.$$
\end{cor}
\begin{proof}
    Hint: for each connected component $H:\|V(H)\|-\|E(H)\|+\|F(H)\| = \text{outer face}+1$. 
\end{proof}

\begin{cor}
    For a connected planar graph $\|E\|\leq3\|V\|-6$.   
\end{cor}
\begin{proof}
    Consider some representation of our planar graph $G$ (e.g. use Auslander-Parter algorithm in subsection~\ref{subsc:planartest}). For this representation every face corresponds to a simple cycle and consists of $\geq 3$ edges and every edge belongs to a boundary of two different faces. Hence, for the graph $G: 2\|E\| \geq 3\|F\|$. By the Euler's theorem:
    $$3\|F\| = 6-3\|V\|+3\|E\|\leq 2\|E\|\imp\|E\|\leq3\|V\|-6.$$
\end{proof}

\begin{cor}
   Every connected planar graph has a vertex of degree at most five.   
\end{cor}
\begin{proof}
    Assume the contrary: all vertices have degrees greater or equal than six. By Handshaking lemma $6\|V\|-12\geq 2\|E\| = \sum_{v_i\in V}{\deg{v_i}}\geq6\|V\|$ holds contradiction.
\end{proof}

Let's consider some operations of a planar graphs which preserve the planarity:

\begin{lm}\label{lm:plan1}
    Every subgraph of a planar graph is planar.
\end{lm}
\begin{proof}
    Holds from the definition.
\end{proof}

\begin{defin}
    \bf{Subdivision} of an edge is an operation when the edge is replaced by a path of the length two: the internal vertex is added to the graph. 
\end{defin}

\begin{lm}\label{lm:plan2}
    Every subdivision of the non-planar graph is non-planar.
\end{lm}
\begin{proof}
    Just consider the inverse operation.
\end{proof}

By these two lemma's we can easy prove the necessary condition of one of the most important theorem about planar graphs:

\begin{thm*}[Pontryagin-Kuratowski]
    A connected graph is planar iff. it doesn't contain subdivisions of complete graphs $K_5$ and $K_{3,3}$.
\end{thm*}

We will provide the full proof of this theorem in section~\ref{subsec:planconn}.

\subsection{Polyhedrons}

\begin{defin}
    \bf{Polyhedron} is a set of polygons such that
    \begin{enumerate}
        \item Each side of a polygon should be the side of one another polygon,
        \item\label{poly:ev} The intersection of any two polygons can be just either edge of vertex,
        \item There exists path through polygons between any two polygons.
    \end{enumerate}
\end{defin}\noindent
Let's introduce another theorem like Jordan curve theorem for hyper-surfaces in $\rR^n$:

\begin{thm*}[Jordan-Brouwer]
Any connected compact simple hypersurface $S^{n-1}$ divides $\rR^n$ for two disjoint open subsets and any curve starting from point of the first subset and ending in a point of the second subset intersect $S^{n-1}$.
\end{thm*}

We can consider each polyhedron $P$ like 2-dimensional surface in $\rR^3$ and thus it divides the 3-dimensional space by two disjoint parts inner $int(P)$ and outer $out(P)$. The boundary of polyhedron is denoted by $\d(P)$.

\begin{defin}
    Polyhedron is said to be \bf{convex} \ifof the edge from any point of polyhedron to any inner point doesn't intersect the polyhedron.  
\end{defin}

This definition can be reformulate as any point of polyhedron is visible from any inner point.

\begin{thm}[Euler for convex polyhedrons]\label{thm:eulerpoly}
    For convex polyhedrons the following holds:
    $$\|V\|-\|E\|+\|F\| = 2,$$
    where $V, E$ and $F$ are vertices, edges and faces of polyhedron respectively. 
\end{thm}
\begin{proof}
    By the convex definition there exists the bijection from polyhedron to a sphere around this polyhedron with the center in inner point. By this bijection the graph of polyhedron $G(P)$ is mapped to the graph on the sphere $G'$ with the same number of vertices, edges and faces. Consider the stereographic projection from the sphere with the center inside some face $F_i$ of graph $G'$ to the 2-dimensional plane. This graph $G'$ maps bijectively (because of stereographic projection property, see figure~\ref{fig:eulerpol}) to the graph on the plane $G''$ such that inner faces of $G'$ except $F_i$ maps to inner faces of $G''$ and $F_i$ maps to the outer face of $G''$. It is easy to see that the graph $G''$ is planar (consider contrary and use bijections properties) with the same number of vertices, edges and faces. Therefore by using Euler theorem for the graph $G''$ the condition holds.    
\end{proof}

\begin{figure}[H]
\vspace{-12pt}
    \centering
	\includegraphics[width=0.8\textwidth]{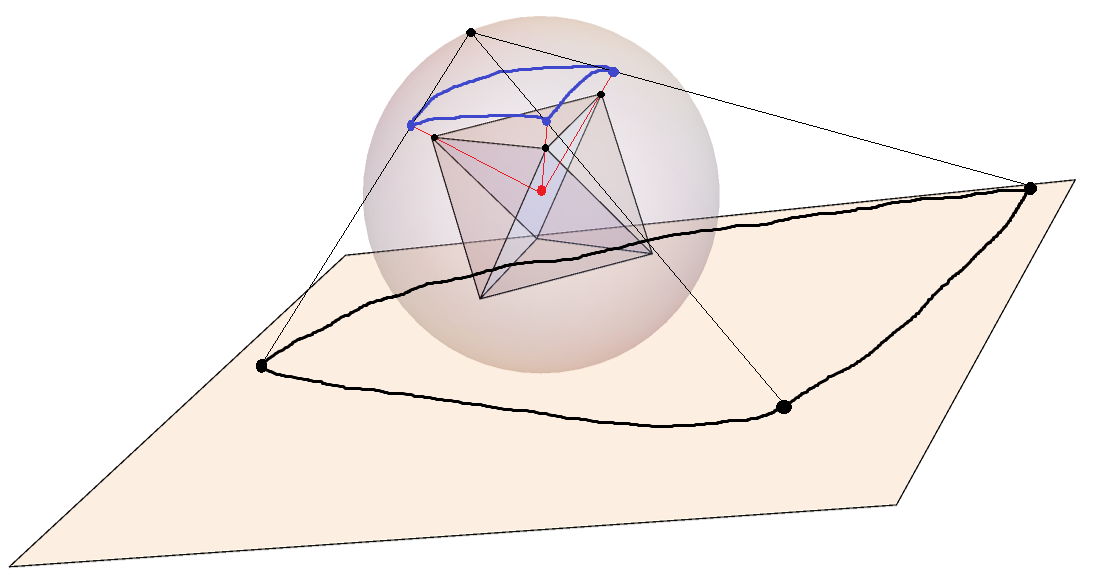}
	\caption{Two bijections from Euler's theorem about convex polyhedrons.}\label{fig:eulerpol}
\end{figure}

\begin{cor}\label{cor:outface}
    For any face of a planar graph there exists the representation where this face is outer face.    
\end{cor}
\begin{proof}
    Let's use the stereographic projection like in the theroem~\ref{thm:eulerpoly} two times. First, let's project the planar graph from the plane to sphere using inverse stereographic projection with any center. Second, choose the center of second stereographic projection inside the chosen face. By the composition of these projections this face becomes outer face.
\end{proof}

Since the Euler theorem holds for polyhedrons, upper estimation $\|E\|\leq3\|V\|-6$ also holds. Let's give lower estimation for polyhedrons. Since the intersection of any two polygons can be just either edge or vertex, any vertex degree must be $\geq 3$. Therefore, by using Handshaking lemma $2\|E\|\geq 3\|V\|$ and thus, $1.5\|V\|\leq\|E\|\leq3\|V\|-6$ holds for polyhedrons.

\begin{defin}
    \bf{Platonic solid} is a convex polyhedron, that consists of the same regular polygons and the vertices have the same degrees. 
\end{defin}

\begin{ex}
Let's $n$ be the number of vertices of polygons of a platonic solid and $deg$ be degree of any vertex. The following holds for platonic solids:
$$n\|F\| = deg \|V\| = 2\|E\|.$$

Let's multiply Euler formula on $2n$ and use the equation before.
$$2\,n\|V\|-2\,n\|E\|+2\,n\|F\| = 4\,n,$$
$$2\,n\|V\|-n\,deg\|V\|+2\,deg\|V\| = 4\,n,$$
$$\|V\|\bigl(2\,n-deg\,(n-2)\bigr) = 4\,n.$$

The number $n>0$ and $\|V\|>0$, thus $2\,n-deg\,(n-2)>0$ and therefore, $(n-2)\,deg<2\,n$.

Let's use this equation and this estimation to classify all platonic solids (see figure~\ref{fig:platonic}). Let's also note that the minimum number of $n = 3$ and the minimum degree $deg = 3$:
\begin{enumerate}
    \item $n = 3\imp deg<6,$
    \begin{enumerate}
        \item $deg  = 3\imp \|V\| = 4.$ This platonic solid is called \emph{Tetrahedron}.
        \item $deg  = 4\imp \|V\| = 6.$ This platonic solid is called \emph{Octahedron}.
        \item $deg  = 5\imp \|V\| = 12.$ This platonic solid is called \emph{Icosahedron}.
    \end{enumerate}
        \item $n = 4\imp 2\,deg<8\imp deg<4,$
        \begin{enumerate}
            \item $deg  = 3\imp \|V\| = 8.$ This platonic solid is called \emph{Cube}.
        \end{enumerate}
    \item $n = 5\imp 3\,deg<10\imp deg<4,$
        \begin{enumerate}
            \item $deg  = 3\imp \|V\| = 20.$ This platonic solid is called \emph{Dodecahedron}.
        \end{enumerate}
    \item $n \geq 6,$ Let's use the estimation again:
    $$0>n\,deg-2\,deg-2\,n= n\,(deg-2)-2\,deg\geq 4\,deg-12.$$
    This gives the contradiction, since $deg\geq 3$. 
    \end{enumerate}
\end{ex}

Using Handshaking lemma and Euler theorem~\ref{thm:eulerpoly} the number of edges and faces can be found. As a result the full classification is obtained:
\begin{center}
\begin{tabular}{|c|c c c c c|}
    \hline
    Name & n & deg & $\|V\|$ & $\|E\|$ & $\|F\|$ \\
    \hline
    Tetrahedron & 3 & 3 & 4 & 6 & 4 \\
    Octahedron & 3 & 4 & 6 & 12 & 8 \\
    Icosahedron & 3 & 5 & 12 & 30 & 20 \\
    Cube & 4 & 3 & 8 & 12 & 6 \\
    Dodecahedron & 5 & 3 & 20 & 30 & 12 \\
    \hline
\end{tabular}
\end{center}

\begin{figure}[H]
    \centering
    \vspace{-10pt}
	\includegraphics[width=0.77\textwidth]{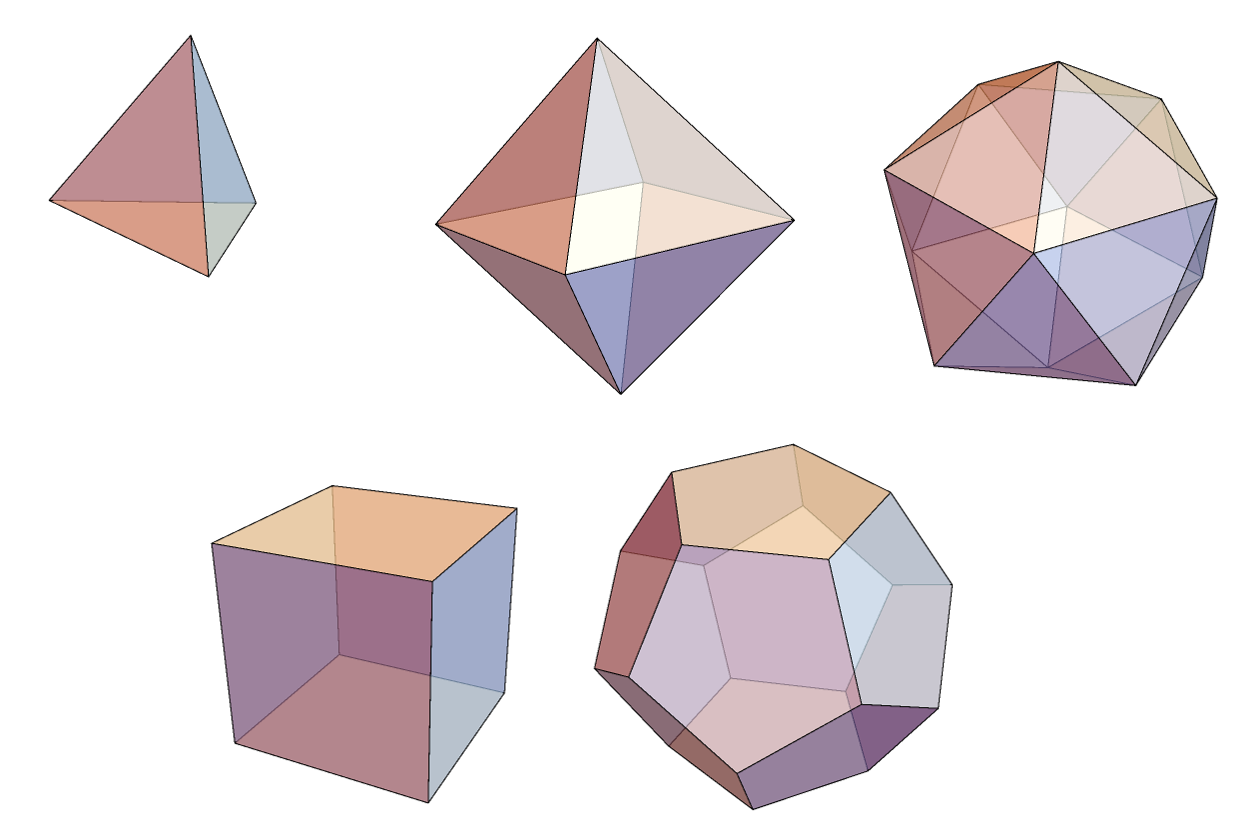}
	\caption{All platonic solids: Tetrahedron, Octahedron, Icosahedron (top), Cube, Dodecahedron (bottom).\label{fig:platonic}}
\end{figure}


\section{$K$-connectivity}

\subsection{$\k$-vertex and $\l$-edge connectivity}

\begin{defin}
    A set of vertices (edges) is called \bf{$k$-vertex (edge) cut} if a graph becomes not connected after the deletion of this set.
\end{defin}

\begin{defin}
    A graph is called \bf{$\k$-vertex (edge) connected} \ifof it doesn't have any $k-1$-vertex (edge) cuts. 
\end{defin}

\NB $\k$-vertex connectivity is also shortly called \emph{$\k$-connectivity}.

\begin{ex}
    \begin{enumerate}
        \item We assume that \emph{empty graph} and the graph with only one vertex are not connected. 
        \item 1-connected graph is just connected.
        \item If a graph is $k$-vertex (edge) connected then it is also $k-1$-vertex (edge) connected, $k-2$-vertex (edge) connected and so on.
        \item The complete graph $K_n$ is $n-1$-vertex and edge connected.
        \item The complete bipartite graph $K_{n, m}$ is $\min(n, m)$-vertex and edge connected.
    \end{enumerate}
\end{ex}

\begin{defin}
    The paths between two vertices are called \bf{k-vertex (edge) independent} \ifof there exist $k$ paths between these vertices which consist of disjoint sets of vertices (edges).   
\end{defin}

Consider another example:
\begin{st}\label{st:kdp1}
     Paths between any two vertices of complete graph $K_{d+1}$ are $d$-vertex (edge) independent.
\end{st}
\begin{proof}
    Consider two vertices $v_i, v_j$. Let's show $d$-vertex (edge) independent paths:
    \begin{enumerate}
        \item One path: $v_iv_j$.
        \item $d-1$ paths: $v_iv_kv_j$, for $v_k\in V\setminus\{v_i, v_j\}$.
    \end{enumerate}
\end{proof}

\begin{thm*}[Menger]
     A graph is $\k$-vertex (edge) connected iff. for any two vertices there exist $\k$-vertex (edge) independent paths.
\end{thm*}

We will prove this theorem in subsection~\ref{subsc:menger}. Let's consider the corollary of this theorem:

\begin{cor}\label{cor:3connminus}
    If a graph $G$ is 3-connected then $G\setminus(u,v)$ is 2-connected for any edge $(u,v)$.
\end{cor}

We will use this corollary later. Now let's consider the relation between vertex and edge connectivity:

\begin{lm}\label{lm:kld}
    Let $\k(G)$, $\l(G)$ are numbers of vertex and edge connectivity for a graph $G$ respectively. Let $d(G)$ be the minimum degree. The following holds: $$\k(G)\leq\l(G)\leq d(G).$$  
\end{lm}
\begin{proof}\ 
    \begin{enumerate}
        \item Case $\k(G)\leq\l(G)$. Since $\k(G)$-vertex independent paths are also edge independent, this estimation holds by the Menger's theorem.   
        \item Case $\l(G)\leq d(G)$. By deleting $d(G)$ edges at the vertex with minimum degree the graph becomes not connected. 
    \end{enumerate}
\end{proof}

\begin{thm}[($\k$, $\l$, d)-graph]
    For any $\k,\l,d \in \rN: \k\leq \l\leq d, $ there exists a graph with $\k$ --- vertex connectivity, $\l$ --- edge connectivity and $d$ --- minimum degree number.  
\end{thm}
\begin{proof}
    Consider two copies $G_1, G_2$ of complete graph $K_{d+1}$. Let's mark $\l$ vertices in the first graph $G_1$ and $\k$ vertices in the second $G_2$. Add $\l$ edges between all marked vertices such that each marked vertex should be adjacent to some added edge (see figure~\ref{fig:kld} for example). Denote the constructed graph by $G$.
    
\begin{figure}[H]
\vspace{-12pt}
    \centering
	\includegraphics[width=0.5\textwidth]{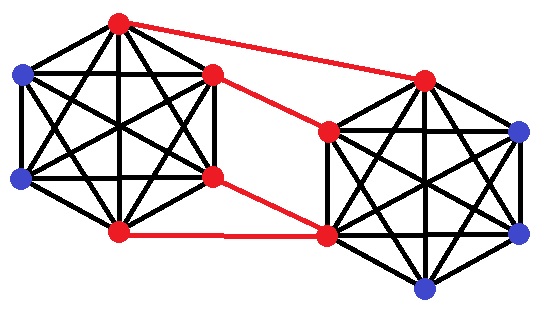}
	\caption{Example of (3, 4, 5)-graph. The red vertices are marked vertices, the red edges are added edges.}\label{fig:kld}
\end{figure}

Let's prove, that the graph $G$ is a graph with $\k$-vertex connectivity, $\l$-edge connectivity and minimum degree $d$.

\begin{enumerate}
    \item Since $\k\leq\l<d+1$, there exists unmarked vertex, that remains the same as for complete graph $K_{d+1}$. Hence the minimum degree of graph $G$ equals $d$.
    \item Let's prove that the collection of paths between any two vertices of the graph $G$ is $\k$-vertex and $\l$-edge independent. If these vertices belong to the same complete subgraph $K_{d+1}$ then the condition holds from the statement~\ref{st:kdp1}. If considered vertices belong to the different complete subgraphs $K_{d+1}$ then let's consider following paths: first vertex, any marked vertex from the same subgraph, adjacent marked vertex from the second subgraph and the second vertex. It is easy to see that they are vertex and edge independent and the number of paths equals to $\k$ and $\l$ for vertex and edge independence respectively.

    By using Menger's theorem the proof ends.
\end{enumerate}
\end{proof}

\subsection{Bridge detecting algorithm}\label{subsc:bridge}

\begin{defin}
    An edge of a connected graph is called \bf{bridge} \ifof the graph becomes not connected after the deletion of this edge.
\end{defin} 

\NB The bridges can be only in 1-edge connected graphs.\\\\
Let's consider the bridge detecting algorithm:

\bf{Tarjan's algorithm.}

For any vertex $v$ let's denote by $\DFS[v]$ --- the order of the vertex $v$, by $a_v$ --- ancestors of $v$, by $d_v$ --- descendants of $v$ corresponding to depth first search algorithm.

\emph{Description:}
\begin{enumerate}
    \item Use depth first search algorithm to define DFS order $\DFS[v]$ and construct the corresponding spanning tree of the graph.
    \item For each vertex $v$ in the reverse DFS order define:
    \begin{equation}\label{eq:lowv}
  \low[v] := \min
    \begin{cases}
      \DFS[v], & \\
      \DFS[a_v] & \text{if edge $(v, a_v)$ doesn't belong to DFS tree,} \\
      \low[d_v] & \text{if edge $(v, d_v)$ belongs to DFS tree.}
    \end{cases}       
    \end{equation}
    \item The edge $(v, d_v)$ is a bridge iff. $\DFS[v] < \low[d_v]$.
\end{enumerate}

\NB The complexity of this algorithm is the same as for depth-first search algorithm and equals to $\mathbf{O\bigl(\|E\|\bigr)}$ for connected graphs.\\

\NB If one add orientation of the current graph by following: orient DFS tree as Tremaux arborescence and orient other edges from descendants to ancestors, the $\low[v]$ will correspond to the lowest number in DFS order that can be reached from the vertex $v$.\\

Finally, Let's prove:
\begin{lm}
    Let the function $\low[v]$ is defined by equation~\ref{eq:lowv} for any vertex of a connected graph $G$. The edge $(v, d_v)$ is a bridge in $G$ iff. $\DFS[v] < \low[d_v]$ for a descendant $d_v$ of vertex $v$.
\end{lm}
\begin{proof}
    Let's denote the set of ancestors of $v$ by $A_v$ and the set of descendants of $v$ by $D_v$. First let's note that the edge $(v, d_v)$ is a bridge in $G$ iff. this edge divides vertices by two disjoint sets $A_v\bigcup\{v\}$ and $D_v$ such that there is no edges between these sets except $(v, d)$. Also the following holds:
    $$\DFS[a_v] < \DFS[v] < \DFS[d_v],\ \forall a_v \in A_v \text{ and } \forall d_v \in D_v,$$
    $$\low[v] \leq \low[d_v],\ \forall d_v \in D_v.$$
    \begin{enumerate}
        \item \emph{Necessarity.} Let's prove that $\DFS[v] < \low[d_v]$ for any descendant $d_v$ of vertex $v$ by induction corresponding to inverse DFS order:
        \begin{enumerate}
            \item Consider the descendant $d_v$ with the highest $n = \DFS[d_v]$. This vertex has no descendants and only can be adjacent to descendants of $v$. Since $\DFS[v] < \DFS[d],\  \forall d \in D_v \imp \DFS[v] < \DFS[a_{d_v}] \imp \DFS[v] < \low[d_v]$.
            \item Let we prove that $\DFS[v] < \low[d_v]$ for any $d_v:\DFS[d_v]\geq k$. Let's prove it for $d_v: \DFS[d_v] = k-1$. For this $d_v: \DFS[v] < \DFS[d_v].$ Since the vertex $d_v$ can only be adjacent to descendants of $v$ and $\DFS[v] < \DFS[d],\ \forall d \in D_v\imp \DFS[v] < \DFS[a_{d_v}]$. By the induction: $\DFS[v] < \low[d_{d_v}]$. Therefore, $\DFS[v] < \low[d_v]$.
        \end{enumerate}
        \item \emph{Sufficiency.} Assume the contrary: the edge $(v, d_v)$ is not a bridge. Thus, there exists an edge $(a_v, d_{1v})$ between two vertices $a_v\in A_v$ and $d_{1v}\in D_v: \text{$d_{1v}$ is descendant of $d_v$}$. Thus, $\low[d_v]\leq \low[d_{1v}] \leq \DFS[a_v] < \DFS[v]$ implies contradiction.
    \end{enumerate}
\end{proof}

\begin{cor}
    Let the function $\low[v]$ is defined by equation~\ref{eq:lowv} for any vertex of a connected graph $G$. The edge $(v, d_v)$ is a 1-cut vertex in $G$ iff. $\DFS[v] \leq \low[d_v]$ for a descendant $d_v$ of vertex $v$.
\end{cor}
\begin{proof}
    The proof is the same as in the previous lemma.
\end{proof}

\NB By this corollary Tarjan's algorithm is also can be used for detecting 1-cut vertices and thus by deleting them one can find all 2-connected components of a graph.


\section{2-connectivity}

First let's prove a theorem about outer face of 2-connected planar graph.

\begin{thm}
      Outer face of 2-connected planar graph $G$ is a simple cycle. 
\end{thm}
\begin{proof}
    Consider any planar representation of $G$. Let's denote the outer face of $G$ by $C_{out}$. Let's prove that this $C_{out}$ contains simple cycle $C$. Indeed, if one can started a walk by edges of $C_{out}$ this walk ends only in a visited vertex (and thus contains simple cycle) otherwise there previous vertex of this walk will be 1-cut vertex. 
    
    Now let's prove that $C = C_{out}$. Assume the contrary: there exists a vertex $v\in out(C)$. Consider any vertex $w\in C$. By Mengers theorem~\ref{thm:menger} there exist two vertex independent paths $P_1$ and $P_2$ from $w$ to $v$. Consider two points $x$ and $y$ --- the last points through $P_1$ and $P_2$ respectively such that $x,y\in C$. Let's denote two vertex independent paths in the subgraph $C$ from $w$ to $x$ and from $w$ to $y$ by $Q_{wx}$ and $Q_{wy}$ respectively. Let's denote a parts of paths $P_1$ and $P_2$ from $x$ to $v$ and from $y$ to $v$ by $P_{xv}$ and $P_{yv}$ respectively. The cycles $C_1 = Q_{wx}\bigcup Q_{wy}\bigcup P_{xv}\bigcup P_{yv}$ and $C_2 = \bigl(C\setminus (Q_{wx}\bigcup Q_{wy})\bigr)\bigcup P_{xv}\bigcup P_{yv}$ are simple cycles by construction. Since the graph $G$ is planar and $C$ is a simple cycle, either $int(C_1)\sub int(C_2)$ or $int(C_2)\sub int(C_1)$. Thus there exist an edge $e\in C\sub C_{out}$ such that the curve corresponding to the edge $e$ is a part of $int(C_1)\bigcap int(C_2)$ in the representation of $G$. Therefore, this curve does not intersect $out(C_1)\bigcap out(C_2)$ and it gives contradiction with the definition of $C_{out}$.
\end{proof}

Let's consider examples of graphs with 2-connectivity property.

\subsection{Eulerian graph}


Consider first 2-edge connected graph.

\begin{defin}
    \bf{Eulerian path} is a path that walks through every edge of the graph only one time.
\end{defin}

\begin{defin}
    \bf{Eulerian cycle} is an Eulerian path where the first and the last vertices coincide.
\end{defin}

\begin{defin}
    An undirected graph is called \bf{Eulerian} \ifof there exists an Eulerian cycle in the graph.
\end{defin}

\begin{lm}
    An Eulerian graph is 2-edge connected.
\end{lm}
\begin{proof}
    Assume the contrary: there exists a bridge $e$ in a graph $G$ such that $G\setminus e$ has at least two connected components. Hence, if there exists Eulerian cycle it should starting and therefore ending in the same component. This holds the contradiction. 
\end{proof}

\begin{thm}[Necessary and sufficient condition for Eulerian graph]\label{thm:ec}
    An undirected graph is Eulerian iff. it is connected and degrees of all vertices are even.
\end{thm}
\begin{proof}\ 
    \begin{enumerate}
        \item \emph{Necessarity.} The Eulerian graph is connected. Assume the contrary: there exists a vertex with odd degree. Let's delete the edges of the current graph one by one through the Eulerian cycle. Since we go in to the vertex by one edge and go out from this vertex by another edge, each time after visiting a vertex the degree of this vertex decreased by 2, and hence the parity of this vertex degree remains the same (note that it also holds to the starting vertex because of the cycle). Therefore, after the deleting all edges corresponding to Eulerian cycle (all edges in the graph) the parity of all vertices remains the same. It holds the contradiction with existence the vertex with odd degree.
        \item \emph{Sufficiency.} Let's start any walk in our graph $G$ with deleting visiting edges. This walk ends in our starting vertex (the first and last vertices in the walk have odd degrees every time and if vertex has odd degree there exists a unvisited edge adjacent to this vertex). Therefore, when the walk ends, this walk will form a cycle. Let's denote the set of edges of this cycle by $C_1$. 

        Consider the graph $G\setminus C_1$. Do the same procedure and find the set of edges $C_2$ of a new cycle in $G\setminus C_1$ and so on. Let's denote the set of this sets by $C = \{C_i\}$. Since all vertices have even degree, the set $C$ contains every edge of $G$. Let's construct the Eulerian cycle by using a stack $S$ and the set $C$: 
        \begin{enumerate}
            \item Start from any vertex of $C_1$. Take away $C_1$ from $C$ and add it to the stack $S$.
            \item\label{enum:ec1} Go through an edge corresponding to this cycle in the top of the stack $S$. If there is no edges corresponding to this cycle, delete the cycle from $S$ and do \ref{enum:ec1} again.
            \item If there exists an edge corresponding to a cycle $C_i$ from $C$, take away $C_i$ from $C$, add it to the stack $S$ and go to \ref{enum:ec1}.  
        \end{enumerate}
        This algorithm produces Eulerian cycle in the graph $G$.
    \end{enumerate}
\end{proof}

\begin{thm}[Eulerian path existence]
    An Eulerian path exists in the connected graph iff. there exist at most two vertices with odd degrees.
\end{thm}
\begin{proof}\ 
    \begin{enumerate}
        \item Assume that there are no odd vertices, then there exists Eulerian cycle by the theorem~\ref{thm:ec}.
        \item Assume that there are only two vertices with odd degrees $u$ and $v$. 
        \begin{enumerate}
            \item If there exists an edge $(u, v)$ then after deleting this edge the graph will have no odd vertices and contains Eulerian cycle by the theorem~\ref{thm:ec}. This Eulerian cycle with the edge $(u, v)$ forms the Eulerian path.
            \item If there are no edges between $u$ and $v$, add the edge $(u, v)$ to the graph. Therefore, the graph with $(u, v)$ will have no odd vertices and contains Eulerian cycle by the theorem~\ref{thm:ec}. This Eulerian cycle without the edge $(u, v)$ forms the Eulerian path.
        \end{enumerate} 
    \end{enumerate}
\end{proof}

\NB If there are only two vertices with odd degrees in the graph, an Eulerian path will starts and ends in these two vertices.\\

Consider basic \bf{Eulerian cycle searching algorithms} in connected graphs:

\begin{enumerate}
    \item \bf{Fleury's algorithm.}\\
    This algorithm used the algorithm of testing whether the graph will become not connected after the deleting current edge.\\
    \emph{Description:}
    \begin{enumerate}
        \item\label{enum:ec4} Walk in the graph with deleting visiting edges which keep the graph connected.
        \item If deleting all adjacent edges provided the graph to be not connected (all edges are bridges) go through one of these edges and delete. Do~\ref{enum:ec4} again. 
    \end{enumerate}

    \NB By using Tarjan’s algorithm for detecting bridges the complexity of Fleury's algorithm will be $\mathbf{O\bigl(\|E\|^2\bigr)}$.
    
    \item \bf{Using list structure.}\\
    This algorithm is based on the algorithm introduced in the theorem~\ref{thm:ec}.\\
    \emph{Description:}
    \begin{enumerate}
        \item Construct the set $C$ from the theorem~\ref{thm:ec} and mark all edges of the graph corresponding to cycles $C_i$.
        \item Add vertices of $C_1$ to the list $L$ and start from the top.
        \item\label{enum:ec2} Go to the next vertex in $L$. If there exist an edge $e$ adjacent to this vertex not in the list $L$ and corresponded to a cycle $C_i$, add all vertices of $C_i$ starting with this edge at the current position of the list $L$. Do~\ref{enum:ec2} again. 
    \end{enumerate}
    In the end the list will contain Eulerian cycle.
    
    \NB The complexity of this algorithm is $\mathbf{O\bigl(\|E\|\,\|V\|\bigr)}$.

    \item \bf{Using two stacks.}\\
    This algorithm used two stack structures $S_1$ and $S_2$.\\
    \emph{Description:}
    \begin{enumerate}
        \item\label{enum:ec3} Walk in the graph with deleting visiting edges and add consecutive vertices corresponded to these edges to the stack $S_1$ until it is possible.
        \item Transfer the vertex $v$ in top of the stack $S_1$ to the stack $S_2$. Do~\ref{enum:ec3} again starting with the new vertex in the top of the stack $S_1$.
    \end{enumerate}
    In the end the stack $S_2$ will contain Eulerian cycle.
    
    \NB The complexity of this algorithm is $\mathbf{O\bigl(\|E\|\,\|V\|\bigr)}$.
\end{enumerate}


\subsection{Hamiltonian graph}

An example of 2-vertex connected graph is a Hamiltonian graph.

\begin{defin}
    \bf{Hamiltoinian cycle} is a simple cycle that consist of all vertices of a graph.
\end{defin}

\begin{defin}
    A graph is called \bf{Hamiltonian} \ifof there exists a Hamiltoinian cycle in the graph. 
\end{defin}

\begin{lm}
    Planar graph $G$ is Hamiltonian iff. there exist a representation and a simple cycle $C$ such that the parts of the dual graph $D(G)$ which belong to $int(C)$ and $out(C)$ are trees.
\end{lm}
\begin{proof}
    We will prove this lemma in the subsection~\ref{subsec:dp}. 
\end{proof}

In the general case it is not very easy to give necessary and sufficient conditions for Hamiltoinian graphs, here we introduce the simplest ones:

\begin{lm}[Necessary Hamiltonian condition]
    A Hamiltonian undirected graph is 2-connected.
\end{lm}
\begin{proof}
    Assume the contrary: there exists a vertex $v$ in a graph $G$ such that $G\setminus v$ has at least two connected components. Hence, if there exists Hamiltonian cycle it should starting and therefore ending in the same component. This holds the contradiction. 
\end{proof}

\NB In general situation the construction of Hamiltonian cycle is NP-hardness problem. Thus knowledge of the Hamiltonian cycle in a graph is used in cryptography namely in the zero-knowledge protocol. 

\begin{lm}
    Graphs of platonic solids are Hamiltonian.
\end{lm}
\begin{proof}
    The proof is in figure~\ref{fig:ham}.
\end{proof}

\begin{figure}[H]
    \centering
    \vspace{-10pt}
	\includegraphics[width=0.8\textwidth]{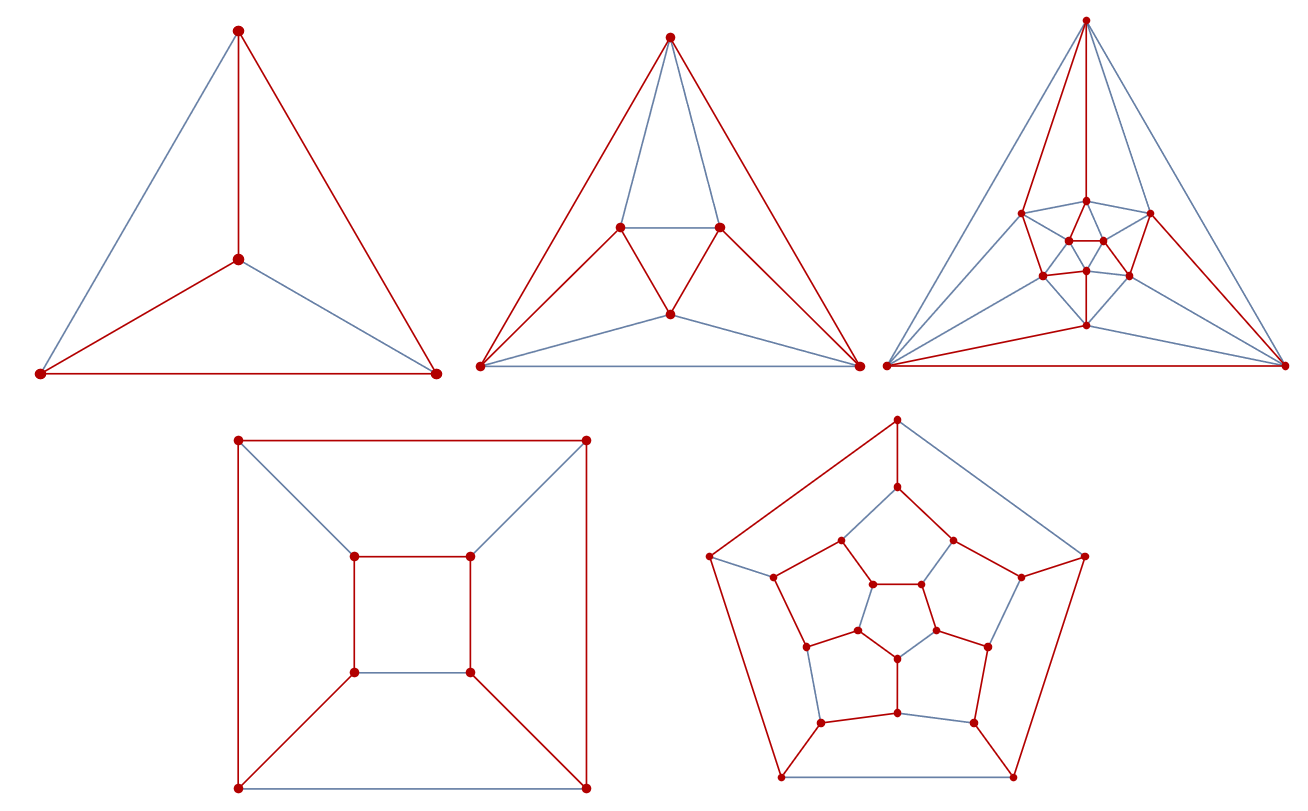}
	\caption{Hamiltonian cycles in planar representations of Platonic solids graphs.}\label{fig:ham}
\end{figure}

Let's introduce several sufficient conditions of Hamiltonian graphs:

\begin{thm}[Ore]\label{thm:ore} 
    Let's $\|V(G)\| = n$, then
    $$\deg(u)+\deg(v)\geq n,\,\text{for all not adjacent } u,v\in V(G)\imp G \text{ is Hamiltonian.}$$
\end{thm}
\begin{proof}
    Let's assume the contrary: there exists not Hamiltonian graph $G$ that is satisfied the condition. By the adding edges to the graph $G$ one can obtain a Hamiltonian graph. Let $G'$ be the boundary case: $G'$ is not Hamiltonian graph and $G'\bigcup (u,v)$ --- Hamiltonian. The condition of this theorem also holds for the graph $G'$, and the edge $(u,v)$ is the edge of Hamiltonian cycle in $G'\bigcup (u,v)$ (otherwise $G'$ is Hamiltonian).
    Let $w_n$ be the next vertex for $w$ by the Hamiltonian cycle (for some orientation of the cycle). Let $V'\sub V(G')$ be the set of vertices $w\in V'$ such that $(v, w)$ is an edge in $G'$ and $U'\sub V(G')$ be the set of vertices $w\in U'$ such that $(u, w_n)$ is an edge in $G'$.
    $$\|U'\| = \deg(u),\ \|V'\| = \deg(v),$$
    $$U'\cup V'\sub V(G')\setminus\{v\}\imp \|U'\cup V'\|\leq n-1,$$
    $$\|U'\cup V'\| = \|U'\|+\|V'\|+\|U'\cap V'\|,$$
    $$\|U'\|+\|V'\| = \deg(u)+\deg(v)\geq n\imp \|U'\cap V'\| \geq 1.$$
    Thus there exists a vertex $w$ such that the edges $(v,w)$ and $(u, w_n)$ are in the graph $G$. This holds contradiction, see figure~\ref{fig:hamc}.
    
\end{proof}

\begin{figure}[H]
    \centering
    \vspace{-10pt}
	\includegraphics[width=0.3\textwidth]{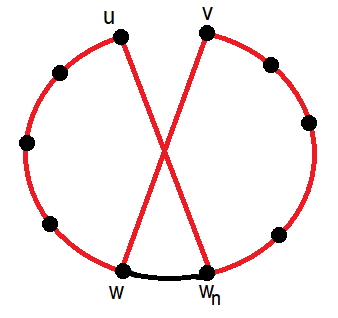}
	\caption{Hamiltonian cycle in the case $\|U'\cap V'\| \geq 1$ in the Ore's theorem~\ref{thm:ore}.}\label{fig:hamc}
\end{figure}

\NB This theorem has been proved for undirected graphs but it can be proved also for directed case in the same way.

\begin{thm}[Dirac]
    Let's $\|V(G)\| = n$, then
    $$\deg(u)\geq \frac n 2,\,\forall u\in V(G)\imp G \text{ is Hamiltonian.}$$
\end{thm}
\begin{proof}
    It simply holds from previous theorem.
\end{proof}

\NB For directed graph this theorem also the same.\\

To prove the next Bondy–Chvátal theorem first several definitions and lemmas are required.\\ \noindent
Let's denote the number of vertices of our graph $G$ by $n$.
\begin{defin}
    The sequence of degrees $d_i = \deg(v_i)$ is called \bf{ordered sequence of degrees} \ifof $d_1\leq d_2 \leq ... \leq d_n$ for some re-numeration of the vertices.
\end{defin}

\begin{defin}
    An ordered sequence of degrees $d_i$ is \bf{majorized} by an ordered sequence of degrees $d'_i$ \ifof $d_i\leq d'_i$ for $i = 1,2,...,n$
\end{defin}

From this point let consider the numeration of $G$ such that the sequence of degrees $G$ is ordered.

\begin{st}[Majorized sequance]
    Let a graph $G'$ be obtained from a graph $G$ by adding one edge. Then the ordered sequence of degrees $G'$ majorizes the ordered sequence of degrees $G$.
\end{st}
\begin{proof}
    Let's the added edge be $(v_i, v_j), i<j$. Then new degrees of vertices $v_i$ and $v_j$ will be $d_i+1$ and $d_j+1$. Let's prove that after changing $d_i$ to $d_i+1$ one can re-numerate vertices such that new ordered sequence $d'_i$ will majorize previous one. 

    Let the number $p\in \rZ^+$ is defined as following: $d_i = d_{i+1} = ... = d_{i+p} < d_{i+p+1}$. 
    
     Now let's produce the re-numeration:
    $$d'_k = d_k,\ k = 1, 2, ..., \widehat{i+p}, ... n \quad \text{ and } \quad d'_{i+p} = d_{i+p}+1.$$
    
    The proof ends by doing the same procedure for $d_j$. 
\end{proof}

\begin{lm}[Upper estimation]
    $$d_k\leq k\iff \|{v\in V(G): \deg(v)\leq k}\|\geq k.$$
\end{lm}
\begin{proof}
    \begin{enumerate}
        \item $(\Rightarrow).$
        $$d_k\leq k\Rightarrow d_1\leq d_2\leq ... \leq d_k\leq k\Rightarrow \|{v\in V(G): \deg(v)\leq k}\|\geq k.$$
        \item $(\Leftarrow).$
        $$\|{v\in V(G): \deg(v)\leq k}\| = k+p,\ p\geq 0\Rightarrow d_1\leq d_2\leq ... \leq d_k\leq ... \leq d_{k+p}\leq k\Rightarrow d_k\leq k.$$
    \end{enumerate}
\end{proof}

The same poof for
\begin{lm}[Lower estimation]
    $$d_{n-k}\geq n-k\iff \|{v\in V(G): \deg(v)\geq n-k}\|\geq k+1.$$
\end{lm}
\begin{proof}
    \begin{enumerate}
        \item $(\Rightarrow).$
        $$d_{n-k}\geq n-k\Rightarrow n-k\leq d_{n-k}\leq d_{n-k+1} ... \leq d_n\Rightarrow \|{v\in V(G): \deg(v)\geq n-k}\|\geq k+1.$$
        \item $(\Leftarrow).$
        
        $$\hspace{-180pt}\|{v\in V(G): \deg(v)\geq n-k}\|\geq k+p+1,\ p\geq 0\Rightarrow$$
        $$\qquad\qquad\qquad\qquad\Rightarrow n-k\leq d_{n-k-p}\leq d_{n-k-p+1} ... \leq d_{n-k}\leq ... \leq d_n\Rightarrow d_{n-k}\geq n-k.$$
    \end{enumerate}
\end{proof}

\begin{lm}[Implication and majorized sequence]
    If the following implication holds for the sequence of degrees $d_i$:
    $$d_k\leq k<\frac n 2\imp d_{n-k}\geq n-k,$$
    then it also holds for a majorized sequence of degrees $d'_i$.
\end{lm}
\begin{proof}
    \begin{enumerate}
        \item If $d_k\leq k < d'_k$ then implication holds by the false first argument,
        \item If $d_k\leq d'_k\leq k$ then $n-k\leq d_{n-k}\leq d'_{n-k}$ by the majorization property.
    \end{enumerate}
\end{proof}

\begin{thm}[Bondy–Chvátal]
    If for an ordered sequence of degrees of a connected graph $G$ the implication holds: 
    $$d_k\leq k<\frac n 2\imp d_{n-k}\geq n-k,$$ 
    then $G$ is Hamiltonian.
\end{thm}
\begin{proof}
    Assume the contrary: there exists a graph $G$ that satisfies the condition of the theorem. Let $G'$ be the maximal non-Hamiltonian graph by adding the edges to the graph $G$. Thus $G'$ becomes Hamiltonian by adding any edge. Let's $u$ and $v$ be two not adjacent vertices such that $\deg(u)+\deg(v)$ is maximum in $G'$ and let $\deg(u)\leq\deg(v)$ without loss of generality.
    
    The graph $G'\bigcup(u,v)$ is Hamiltonian. Let like in the Ore's theorem~\ref{thm:ore} $w_n$ be the next vertex for $w$ by the Hamiltonian cycle. Let $V'\sub V(G')$ be the set of vertices $w\in V'$ such that $(v, w)$ is an edge in $G'$ and $U'\sub V(G')$ be the set of vertices $w\in U'$ such that $(u, w_n)$ is an edge in $G'$. Then,
    $$\|U'\| = \deg(u),\ \|V'\| = \deg(v),$$
    $$U'\cup V'\sub V(G')\setminus\{v\}\imp \|U'\cup V'\|\leq n-1,$$
    \begin{enumerate}
        \item If $\|U'\cap V'\| \geq 1$, then it holds the contradiction like in the Ore's theorem~\ref{thm:ore} by figure~\ref{fig:hamc}.
        \item If $\|U'\cap V'\| = 0$. Let's denote $\deg(u)$ by $k$. 
        $$2\,k = 2\deg(u)\leq \deg(u)+\deg(v) = \|U'\cup V'\|\leq n-1 < n \imp k < \frac n 2.$$
        Then $\forall w\in U'$:
        \begin{enumerate}
            \item $w$ doesn't adjacent to $v$,
            \item $\deg(w)\leq \deg(u) = k\ \bigl(\text{otherwise }\deg(w)+\deg(v) > \deg(u)+\deg(v)\bigr).$
        \end{enumerate} 
        Therefore, 
        $$\|z\in V(G):\deg(z)\leq k\|\geq \|U'\| = \deg(u) = k\overset{\text{upper estimation}}{\imp}d_k\leq k < \frac n 2\overset{\text{cond. of theorem}}{\imp}$$
        $$d_{n-k}\geq n-k \overset{\text{lower estimation}}{\imp} \|z\in V(G):\deg(z)\geq n-k\|\geq k+1,$$
        and thus there exists a vertex $z$ not adjacent to $u$ with $\deg(w)\geq n-k$. For this vertex $\deg(z)+\deg(u) \geq n > \deg(v)+\deg(u)$. This holds a contradiction.
    \end{enumerate}
\end{proof}

\begin{thm}[Whitney]
    Any planar 4-connected graph is Hamiltonian.
\end{thm}
\begin{proof}
    We will prove this theorem later in the subsection~\ref{subsc:4conn} 
\end{proof}


\section{Planarity and $\k$-connectivity}\label{subsec:planconn}


\subsection{Planarity testing algorithms}\label{subsc:planartest}

By using Tarjan's algorithm we can delete all 1-cut vertices in a graph and consider just 2-connected components. Therefore, let's consider in this section a 2-connected graph without loss of generality. 

\begin{defin}
    \bf{$C$-component} of the graph $G$ corresponded to a simple cycle $C$ is either edge that doesn't belong to the cycle $C$ but connected vertices of $C$ or connected component of the graph $G\setminus C$ with all attachments from this subgraph to the cycle $C$.
\end{defin}

\begin{defin}
    Two $C$-components are \bf{skew} iff. they contains simple paths with starting points $x_1, x_2\in C$ and ending points $y_1, y_2 \in C$ such that these vertices are in the order $x_1, x_2, y_1, y_2$ corresponded to some orientation of the cycle $C$. 
\end{defin}

\begin{defin}
    \bf{Interplacement graph} with respect to cycle $C$ is the graph where vertices correspond to $C$-components of the cycle $C$ and edges correspond to the skew relationship between $C$-components (see figure~\ref{fig:apalg1} for example).
\end{defin}

\NB Interpacement graph can have isolated vertices.

Here we will consider the bipartitness property for not connected graphs: not connected graph is said to be bipartite if each connected component of the graph is bipartite (see figure~\ref{fig:apalg1}).

\begin{figure}[H]
    \centering
	\includegraphics[width=0.6\textwidth]{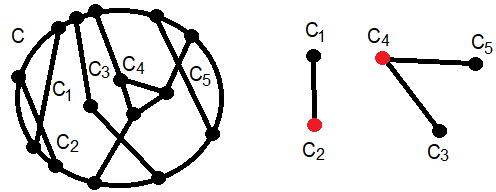}
	\caption{The simple cycle $C$, five $C$-components $C_1, C_2, ... , C_5$ (left) and bipartite interplacement graph corresponded to the cycle $C$ (right).}\label{fig:apalg1}
\end{figure}

Consider a basic algorithm for testing planarity property for 2-connected graphs:

\begin{enumerate}
    \item  \bf{Auslander-Parter algorithm}.\\
    This algorithm tested planarity property of a 2-connected graph $G$ by checking interplacement graphs for each cycle $C$ for bipartitness.\\
    \emph{Description:}
    \begin{enumerate}
        \item\label{enum:apalg0} Find any simple cycle $C$, e.g. by adding an edge to spanning tree.
        \item\label{enum:apalg1} Find all connected components of the graph $G\setminus C$ and add to them edges which doesn't belong to the cycle $C$ but connected vertices of $C$. These will be all $C$-components.
        \item\label{enum:apalg2} Construct interplacement graph and check it for bipartitness. If interplacement graph is not bipartite, then the graph $G$ is not planar.
        \item\label{enum:apalg3} For each $C$-component $C'$ construct a new cycle in the subgraph $G' = C\bigcup C'$ by changing the path through $C$ between any two consecutive (corresponding to some direction on $C$) attachments of $C'$ to the path between them in $C'$ (see figure~\ref{fig:apalg2}) and go recursively to the step~\ref{enum:apalg1} for the graph $G'$.
        \item The recursion terminates when the cycle $C$ has just one $C$-component in $G'$ and it is a path. 
    \end{enumerate}

\begin{figure}[H]
\vspace{-15pt}
    \centering
	\includegraphics[width=0.58\textwidth]{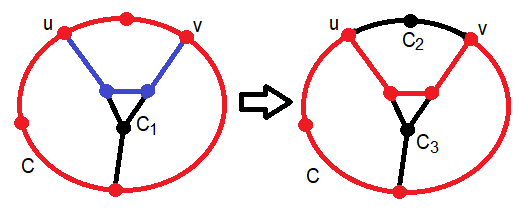}
	\caption{A simple cycle $C$ (left, red), $C$-component $C_1$ (left, black and blue) and a path between consecutive attachments of $C_1$ in $C$-component $C_1$ (left, blue). New cycle $C$ (right, red) and new $C$-components $C_2, C_3$ (right, black) corresponded to the step~\ref{enum:apalg3} of Auslander-Parter algorithm.}\label{fig:apalg2}
\end{figure}

\begin{lm}\label{lm:ausparter}
    A 2-connected graph is planar iff. for each simple cycle $C$:
    \begin{enumerate}
        \item the interplacement graph is bipartite,
        \item for each $C$-component $C'$ the subgraph $C\bigcup C'$ is planar.
    \end{enumerate}
\end{lm}
\begin{proof}
    If all conditions held then the planar representation can be constructed by reverse steps of recursion in Auslander-Parter algorithm.
\end{proof}

\NB If a graph is planar (if algorithm doesn't terminate earlier) the number of vertices which belong to $int(C)$ is decreasing after each step of the recursion. Hence the depth of recursion is $O\bigl(\|V\|\bigr)$. Since the graph is planar, for this graph $\|E\|\leq 3 \|V\|-6$ and thus, the number of $C$-components (and thus vetrices of the interplacement graph) is $O\bigl(\|E\|\bigr) = O\bigl(\|V\|\bigr)$. Therefore, the complexity of bipartite testing algorithm for the interplacement graph is $O\bigl(\|V\|^2\bigr)$. Thus, the complexity of Auslander-Parter algorithm is $\mathbf{O\bigl(\|V\|^3\bigr)}$.

\end{enumerate}

\subsection{3-connectivity}\label{subsc:3conn}

\begin{lm}\label{lm:plan2conn}
    The minimal non-planar subgraph of a non-planar graph is 2-connected.
\end{lm}
\begin{proof}
Let's denote this subgraph by $G$.
\begin{enumerate}
    \item Let's prove that $G$ is 1-connected or just connected. Assume the contrary. Then, this subgraph $G$ has at least two connected components. Since $G$ is minimal and non-planar, each component should be planar and therefore, subgraph $G$ is also planar. This holds contradiction.
    \item Let's prove that $G$ is 2-connected. Assume the contrary. Then, there exists a vertex $v$, such that by deleting $v$ graph becomes not connected and thus has at least two connected components. Let's denote one of the components by $G_1$ and by $G_2$ all the rest. Since the subgraph $G$ is minimal and non-planar, the induced subgraphs on $V(G_1)\bigcup \{v\}$ vertices and on $V(G_2)\bigcup \{v\}$ vertices are planar. By using corollary~\ref{cor:outface} representations of these subgraphs can be modified such that the vertex $v$ would belong to the boundary of outer faces of these subgraphs. Therefore, the union of representations of these subgraphs will be planar and it is representation of the whole graph $G$. This holds contradiction.   
\end{enumerate}
\end{proof}

\begin{lm}\label{lm:k5k33}
    If the minimal non-planar subgraph of a non-planar graph doesn't contain subdivisions of complete graphs $K_5$ and $K_{3,3}$ then this subgraph is 3-connected.
\end{lm}
\begin{proof}
    Let's assume the contrary. Denote the minimal non-planar subgraph by $G$. By lemma~\ref{lm:plan2conn} the subgraph $G$ is 2-connected, then, there exist two vertices $u$ and $v$, such that by deleting them $G$ becomes not connected. Let's denote the connected components of the $G\setminus\{u, v\}$ by $H_1, H_2, ..., H_k$. 
    
    Let $M_i$ be an induced subgraph on vertices $V(H_i)\bigcup\{u,v\}$. Let's prove that each $M_i$ is also connected. Assume the contrary: $M_i$ is not connected. Thus, $H_i = M_i\setminus\{u,v\}$ has connections to only one vertex. Then, the graph $G$ becomes not connected after deleting this vertex. This contradicts to the 2-connectivity of $G$. Therefore, each $M_i$ has connections to both vertices $u$ and $v$ and thus, connected.
    
    Assume that all $M_i+(u,v)$ (either $M_i$ with the edge $(u,v)$ or just $M_i$ if $u$ and $v$ are adjacent) are planar, then one can combine them for a whole planar representation by putting them inside each other (see figure~\ref{fig:k5k33}). Therefore, $\bigcup_{i = 1}^k \bigl(M_i+(u,v)\bigr) = G+(u,v)$ is planar. This contradict with non-planarity of $G$. Thus, there exists non-planar $M_j+(u,v)$.

\begin{figure}[H]
    \centering
	\includegraphics[width=0.2\textwidth]{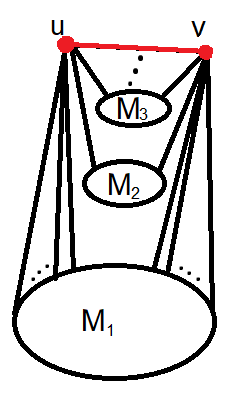}
	\caption{Planar representation of $\bigcup_{i = 1}^k \bigl(M_i+(u,v)\bigr)$ for lemma~\ref{lm:k5k33}.\label{fig:k5k33}}
\end{figure}

    Since $G$ is the minimal non-planar subgraph that contains no subdivisions of graphs $K_5$ and $K_{3,3}$, the graph $M_j+(u,v)$ should contain subdivisions of graphs $K_5$ or $K_{3,3}$. Since all $M_i$ are connected, for any $k\neq j$ there exists path $P_{uv}$ between $u$ and $v$ in $M_k$ and therefore, the graph $M_j\bigcup P_{uv}\sub G$ contains subdivisions of graphs $K_5$ or $K_{3,3}$. This holds contradiction since $G$ doesn't contain subdivisions of graphs $K_5$ and $K_{3,3}$.
\end{proof}

Now let's prove the Pontryagin-Kuratowski theorem:

\begin{thm}[Pontryagin-Kuratowski]\label{thm:ponkur}
    A connected graph is planar iff. it doesn't contain subdivisions of complete graphs $K_5$ and $K_{3,3}$.
\end{thm}
\begin{proof}
\emph{Necessarity} holds from lemmas~\ref{lm:plan1} and~\ref{lm:plan2}.\\
\emph{Sufficiency.} Assume the contrary: there exists non-planar graph that doesn't contain subdivisions of complete graphs $K_5$ and $K_{3,3}$. Let's consider it's minimal non-planar subgraph $G$. By lemma~\ref{lm:k5k33} this subgraph $G$ is 3-connected.

Consider two adjacent vertices $u$ and $v$. Since $G$ is minimal non-planar subgraph, $G' = G\setminus(u,v)$ is planar. By corollary~\ref{cor:3connminus} the subgraph $G'$ is at least 2-connected and by Menger's theorem there exist two vertex independent paths from $u$ to $v$. Thus, there exists a simple cycle in planar representation of $G'$ that contains vertices $u$ and $v$. Let $C$ be the simple cycle that contains $u$ and $v$ in planar representation of $G'$ with the maximum edges in $int(C)$.

If there exists a vertex that belongs to the $out(C)$ one can construct the cycle that contains more edges in inner part than in $int(C)$ (see figure~\ref{fig:widercycle}). Thus all vertices of the planar representation $G'$ belong to $int(C)$. 

\begin{figure}[H]
\vspace{-12pt}
    \centering
	\includegraphics[width=0.3\textwidth]{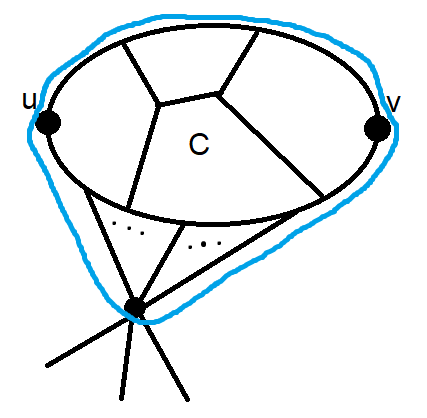}
	\caption{A cycle (blue) with more edges in inner part than in $int(C)$.}\label{fig:widercycle}
\end{figure}

Let's give some additional definitions which are necessary for the proof: let's denote by \emph{bridge of cycle $C$} in planar representation of $G'$ a simple path starting and ending in vertices of $C$ that doesn't contain another vertices and edges of $C$ (this \textquote{bridge of cycle} definition is not general definition in graph theory we use this definition only in this theorem. The general another definition of a bridge will be introduced in section~\ref{subsc:bridge}.). If bridge laying in $out(C)$ then it is called \emph{outer bridge}, if it is in $int(C)$ then it is called \emph{inner} bridge. Since the subgraph $G'$ is connected and all vertices except $V(C)$ belong to $int(C)$, all connected components of $G'\setminus C$ are inner $C$-components.

Now all preparations are done. Let's prove this theorem step by step: 
\begin{enumerate}
    \item If there is no inner (outer) bridges in $G'$ then there exists planar representation of $G$. This holds a contradiction, thus \emph{there exist at least one inner and one outer bridge}.
    \item If all inner bridges in $G'$ are not skew with $(u,v)$ then also there exists planar representation of $G$. This holds a contradiction, thus \emph{at least one inner bridge is skew with $(u,v)$}. Let's denote this inner bridge by $\mathbf{B_1}$.
    \item Since there is no vertices belongs to $out(C)$, \emph{all outer bridges are edges}.
    \item If there exists an outer bridge in $G'$ that is not skew with $(u,v)$ one can extend the cycle $C$ such that it will contain more edges (like in figure~\ref{fig:widercycle}). This holds a contradiction, thus \emph{all outer bridges are skew with $(u,v)$}. Let's denote an outer bridge by $\mathbf{B_2}$.
    \item \emph{Any two inner $C$-components are not skew} otherwise they intersect and form one whole $C$-component.
    \item\label{enum:case1} If a $C$-component is skew with $(u,v)$ and outer bridge then $G$ contains a subdivision of $K_{3,3}$ (see figure~\ref{fig:pk2}). This holds a contradiction, thus \emph{either $C$-component skew with an outer bridge or $C$-component skew with $(u,v)$}.
    \item\label{enum:case2} If there doesn't exist bridges from $u$ to $v$, then, let's divide $C$-components of the graph $G'$ by two groups: the first is $(u,v)$ and all inner $C$-components which are not skew with $(u,v)$ and the second --- outer bridges and inner $C$-components which are skew with $(u,v)$. The interplacement graph of these $C$-components (including the edge $(u,v)$) is bipartite by this division and the property~\ref{enum:case1} and since all $C$-components are planar ($G'$ is planar) then $G$ is planar by the lemma~\ref{lm:ausparter}. This holds a contradiction, thus \emph{there exists a bridge from $u$ to $v$}. Let's denote it by $\mathbf{B_3}$.
\end{enumerate}

Since the subgraph $G'$ is planar, the bridge $B_1$ should intersect with $B_3$ in some points. Consider the cases of relative positions of $B_1$, $B_2$ and $B_3$:
\begin{enumerate}
    \item The bridges $B_1$ and $B_2$ don't intersect.
    \item The bridges $B_1$ and $B_2$ intersect in 1 point.
    \item The bridges $B_1$ and $B_2$ intersect at least in 2 points.
\end{enumerate}
And
\begin{enumerate}[label=\alph*.]
    \item Bridges $B_1$ and $B_3$ intersect in 1 point. 
    \item Bridges $B_1$ and $B_3$ intersect at least in 2 points.
\end{enumerate}

For cases 1ab, 2ab and 3b the subgraph $G$ contains a subdivision of $K_{3,3}$ (see figure~\ref{fig:pk2}). For the case 3a it is easy to see that the subgraph $G$ contains a subdivision of $K_5$. This holds the contradiction and ends the proof.

\begin{figure}[H]
\vspace{-12pt}
    \centering
	\includegraphics[width=0.8\textwidth]{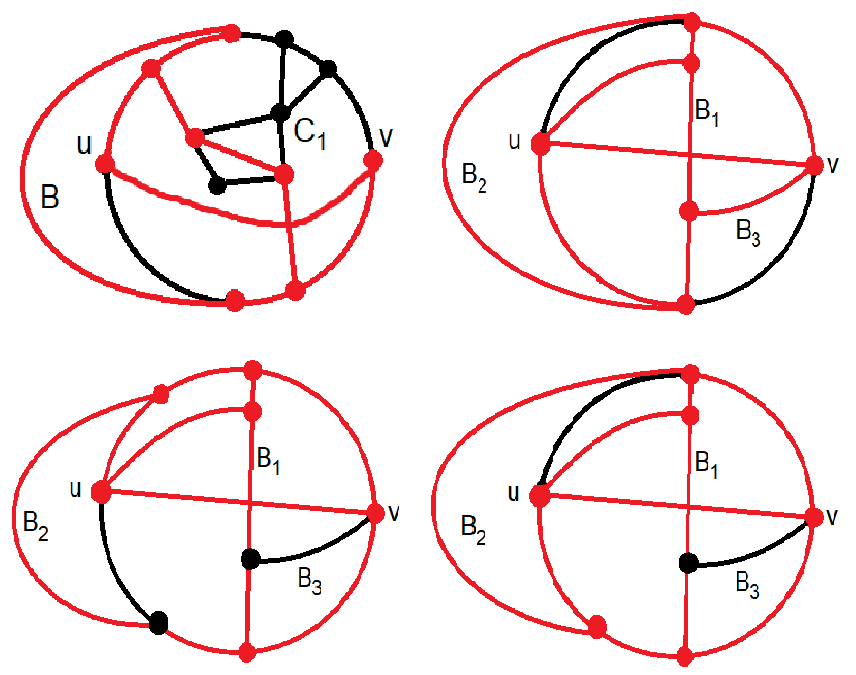}
	\caption{The subdivisions of $K_{3,3}$ in the Pontryagin-Kuratowski theorem for the case~\ref{enum:case1} (left, top), for the cases 1ab (right, top), 2ab (left, bottom), 3b (right, bottom) marked as red.}\label{fig:pk2}
\end{figure}

\end{proof}

Another result about a connection between planarity and connectiveness is
\begin{thm*}[Steinitz]
    A connected graph is planar and 3-connected iff. it is graph of a convex polyhedron. 
\end{thm*}
We will prove this theorem in subsection~\ref{subsec:st}.

\subsection{4-connectivity}\label{subsc:4conn}



First let's introduce famous Thomassen theorem for $C$-components

\begin{thm}[Thomassen]\label{thm:thomas}
    Let $C_{out}$ be a simple cycle corresponded to outer face of 2-connected planar graph $G$. Consider a vertex $v$ and an edge $e$ of $C_{out}$ and any another vertex $u$. Then there exists simple cycle $C$  that contains $u, v$ and the edge $e$ such that
    \begin{enumerate}
        \item Each $C$-component has at most 3 vertices of attachments,
        \item Each $C$-component containing an edge of $C_{out}$ has at most 2 vertices of attachments.
    \end{enumerate}
\end{thm}

The proof for this theorem is complicated, thus we just give a reference~\cite{Thomassen} and simply prove the theorem about 4-connected graphs


\begin{cor}(Whitney).
    Any planar 4-connected graph is Hamiltonian.
\end{cor}
\begin{proof}
    Since the outer face should contain at least 3 vertices, let's choose an edge $e\in C_{out}$ not adjacent to $v\in C_{out}$ and $u\notin C_{out}$. The simple cycle $C$ corresponded to Thomassen theorem~\ref{thm:thomas} should be Hamiltonian cycle (otherwise there exists $C$-component with $\|V(C)\|\geq 1$ that has at most 3 attachments and it holds a contradiction with 4-connectivity). 
\end{proof}


\section{Duality}

\subsection{Dual graphs}\label{subsec:dp}

\begin{defin}
    \bf{Dual graph} $D(G)$ corresponded to a planar representation of a graph $G$ is the graph constructed as following: in each face $F_u$ of the graph $G$ add corresponded vertex $u$ for dual graph, if two faces $F_u$ and $F_v$ adjacent to each other, then add the edge $(u,v)$.
\end{defin}

The dual graphs are defined only for planar graphs, thus we omit the planar properties in this subsection for short.\\   

\NB Dual graph can contain loops and multi-edges (see figure~\ref{fig:duals}).

\NB For isomorphic graphs dual graphs can be non-isomorphic (see figure~\ref{fig:duals} left).

\NB Non-isomorphic graphs can have isomorphic dual graphs (see figure~\ref{fig:duals} right).\\

We see that in the case then dual graph is not simple there are many non-trivial non-intuitive difficult things. In this case 1-cycles (loops) and 2-cycles (multi-edges) can be considered as simple-cycles and thus consider planar representations for not only simple graphs.

\begin{figure}[H]
\vspace{-12pt}
    \centering
	\includegraphics[width=1.0\textwidth]{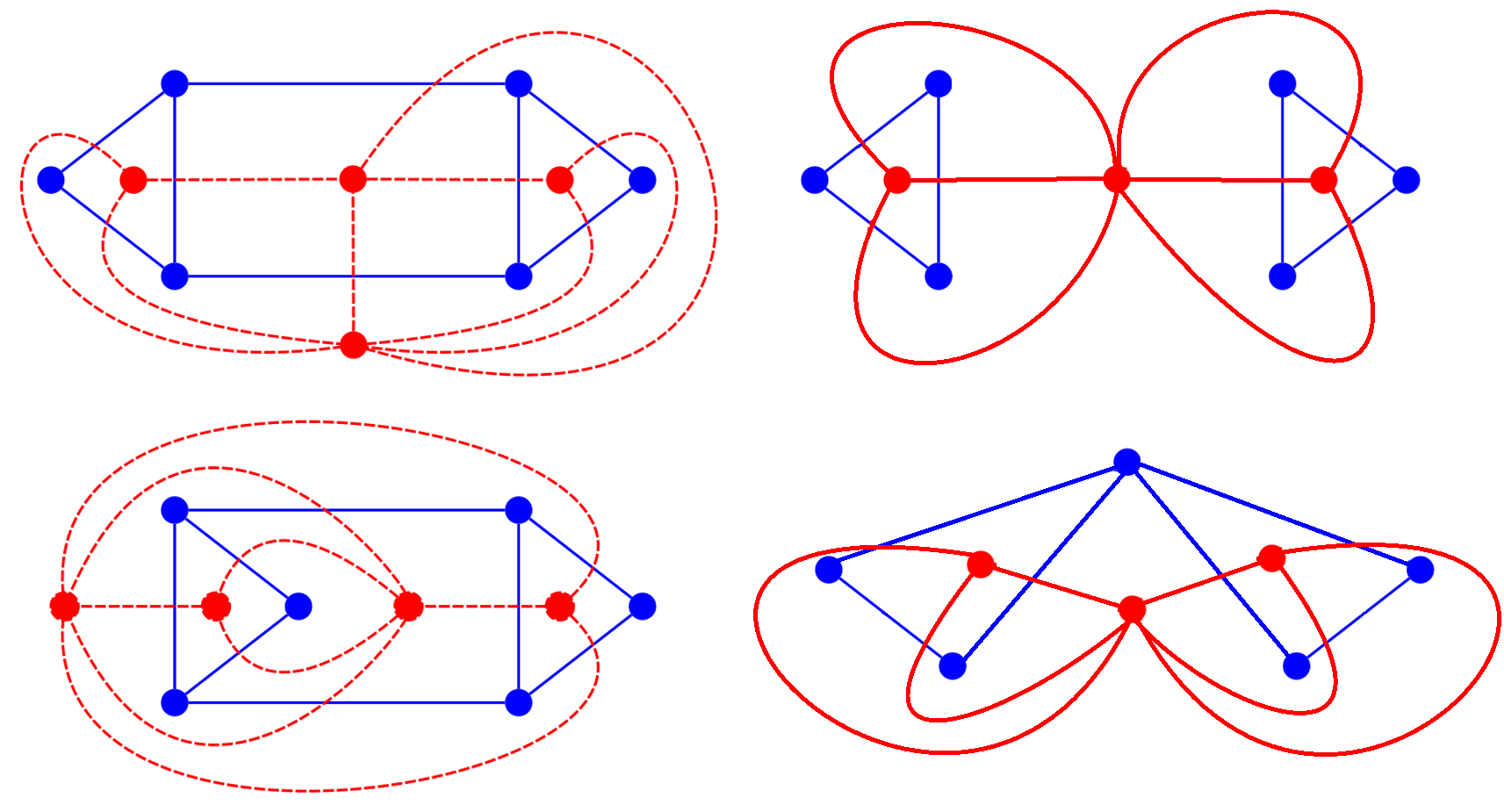}
	\caption{Not isomorphic dual graphs for two different planar representations of one simple graph (left). Isomorphic dual graphs for non-isomorphic simple graphs (right).}\label{fig:duals}
\end{figure}

\begin{lm}
    Dual graph is connected.
\end{lm}
\begin{proof}
    For any two vertices in dual graph $D(G)$ there exists a curve started from the first point and ended in the second point that doesn't contain vertices of $G$. The sequence of faces and edges traversed by this curve corresponds to the path between these two points in $D(G)$.
\end{proof}

\begin{lm}\label{lm:dualpl}
    Dual graph is planar.
\end{lm}
\begin{proof}
    Any edge $(u,v)$ of a dual graph corresponds to the adjacency property of two faces $F_u$ and $F_v$ and thus can be matched with the edge of the graph $G$ that belongs to $F_u\bigcap F_v$. Therefore, one can redraw this representation of the dual graph such that $u\in int(F_u)$ and $v\in int(F_v)$ and $(u,v)\sub F_u\bigcup F_v$ and moreover, there exists small neighborhood $U(u,v)\sub F_u\bigcup F_v$. 

    Assume the contrary, that dual graph is not planar. Then, for any representations of the dual graph there exist two edges $e_1$ and $e_2$ which intersect. Hence there exists neighborhood of $e_1$ and $e_2$ such that  $U(e_1)\sub F_{u_1}\bigcup F_{v_1}$ and $U(e_2)\sub F_{u_2}\bigcup F_{v_2}$ respectively. Thus, there exists common open set $W = U(e_1)\bigcap U(e_2)$ such that $W\sub (F_{u_1}\bigcup F_{v_1})\bigcap(F_{u_2}\bigcup F_{v_2})$. It holds contradiction since faces of a planar representation can intersect only by edges but not by open sets.
\end{proof}

\begin{lm}
    Planar graph $G$ is Hamiltonian iff. there exist a representation and a simple cycle $C$ such that the parts of the dual graph $D(G)$ without the vertex corresponded to the outer face which belong to $int(C)$ and $out(C)$ are trees.
\end{lm}
\begin{proof}
    The existence of vertex of the graph $G$ that belongs to $int(C)$ or $out(C)$ is equivalent to the existence of a cycle in the dual graph $D(G)$.
\end{proof}

\begin{lm}
    Let's $G$ be connected graph. Dual graph $D(G)$ is simple iff. $G$ is 3-edge connected. 
\end{lm}
\begin{proof}
    A loop in the dual graph $D(G)$ corresponds to a bridge in the graph $G$. Therefore, the absence of loops corresponded to 2-edge connectivity of the graph $G$. Multi-edges in $D(G)$ correspond to adjacent faces, which have several common adjacent edges. Thus, the additional absence of multi-edges is corresponded to 3-edge connectivity of the graph $G$.
\end{proof}

\begin{cor}
    Let's $G$ be connected graph. Dual graph $D(G)$ has no loops iff. $G$ is 2-edge connected (or equivalently $G$ has no bridges).
\end{cor}

\begin{lm}[Squared duality]\label{lm:dualsquare}
    For 2-edge connected graph $G$ with no loops holds
    $$D(D(G)) = G.$$ 
\end{lm}
\begin{proof}
    First, let's provide the matching between vertices of the graph $G$ and faces of dual graph $D(G)$. Consider a vertex $u$ of the graph $G$. Denote adjacent edges by $(u, u_i)$ for $i = 1, 2, ... ,k$ and the faces which contain $(u, u_i), (u, u_{i+1})$ for $i = 1, 2, ... ,k-1$ and $(u, u_k), (u, u_1)$ by $F_1, F_2, ..., F_k$ sequentially. These faces correspond to vertices $v_1, v_2, ..., v_k$ of the dual graph respectively. Since the graph is 2-edge connected with no loops, dual graph $D(G)$ also has no loops. Therefore, $F_i\not = F_{i+1}$ for $i = 1, 2, ... ,k-1$ and $F_k\not = F_1$. Thus, each pair $F_i, F_{i+1}$ and $F_k, F_1$ intersect for $i = 1, 2, ..., k-1$ by the edge $(u, u_{i+1})$ and $(u, u_1)$ respectively and thus there is matching between edges $(u, u_{i+1})$ and $(v_i, v_{i+1})$ and $(u, u_1)$ and $(v_k, v_1)$. This matching continues to a matching of the vertex $u$ and a simple cycle $v_1 v_2 ... v_k$ in the dual graph (see figure~\ref{fig:dualmatch1}). Thus, we matched every vertex with a simple cycle (or face).

\begin{figure}[H]
    \centering
	\includegraphics[width=0.7\textwidth]{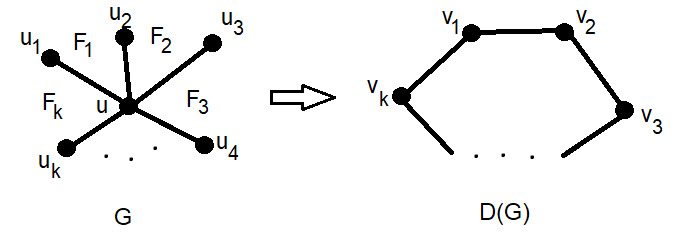}
	\caption{Matching between multi-angle and simple cycle for the graph $G$ and the dual graph $D(G)$ in lemma~\ref{lm:dualsquare}.}\label{fig:dualmatch1}
\end{figure}
    
    Now let's consider a vertex $v$ of the dual graph $D(G)$. This vertex corresponds to some face $F_v$ of the graph $G$. The boundary of the face $F_v$ is a simple cycle (since there are no loops in $G$). Let's denote this cycle by $u_1 u_2 ... u_k$. For each vertex $u_i$ of the graph $G$ there exists a matching face $F_i$ of the dual graph such that these faces adjacent sequentially to each other due to the matching property of edges of $D(G)$ and $G$ (each edge $(u_i, u_{i+1})$ corresponded to edge $F_i\bigcap F_{i+1}$). Since all faces $F_i$ have sequentially intersections in $int(F_v)$ and there are no another edges which belong to $int(F_v)$, the intersection $\bigcap_{i = 1}^k F_i = v$ (see figure~\ref{fig:dualmatch2}).  

\begin{figure}[H]
    \centering
	\includegraphics[width=0.7\textwidth]{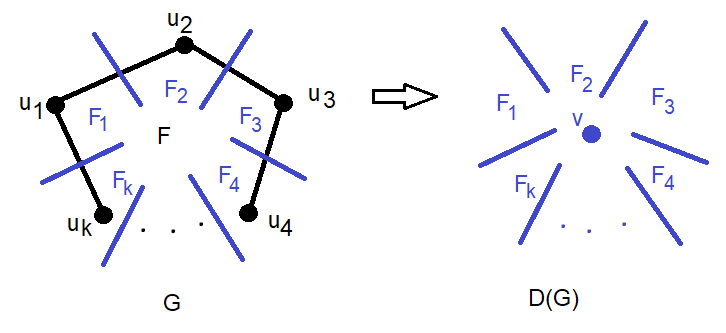}
	\caption{Matching between multi-angle and simple cycle for graph $D(G)$ (blue) and graph $G$ (black) in lemma~\ref{lm:dualsquare}.}\label{fig:dualmatch2}
\end{figure}

    Therefore, any multi-angle of the graph $G$ (the set of faces which contain correspondent vertex) matches bijectively with a face of the graph $D(G)$ and vice versa. Thus, for every vertex $D\bigl(D(v)\bigr) = v$ and for every edge $e$ adjacent to this vertex $D\bigl(D(e)\bigr) = e$. This can be continued for a whole graph and thus, $D\bigl(D(G)\bigr) = G$.
\end{proof}

\begin{cor}
    For 2-edge connected graph $G$ with no loops holds
    $$\|V\bigl(D(G)\bigr)\| = \|F(G)\|,\quad\|E\bigl(D(G)\bigr)\| = \|E(G)\|,\quad\|F\bigl(D(G)\bigr)\| = \|V(G)\|.$$
\end{cor}
\begin{proof}
    By definition: $\|V\bigl(D(G)\bigr)\| = \|F(G)\|$. By lemma~\ref{lm:dualpl}: $\|E\bigl(D(G)\bigr)\| = \|E(G)\|$. Using the correspondence in lemma~\ref{lm:dualsquare}: $\|F\bigl(D(G)\bigr)\| = \|V(G)\|$.
\end{proof}

\NB These lemma and corollary also hold for any connected graphs but the proof is more complex because of loops.

\begin{cor}
    Any non-isomorphic 2-edge connected planar graphs with no loops have non-isomorphic dual graphs.
\end{cor}
\begin{proof}
    It holds from the bijection in lemma~\ref{lm:dualsquare}.
\end{proof}

\NB If dual graph has loops this corollary \emph{doesn't hold}. The example can be constructed as in figure~\ref{fig:duals} right-bottom, if split the connected vertex to an edge between two triangles and then move the obtained loop of the dual graph such that a leaf will appear in the dual of dual graph.

\begin{lm}\label{lm:2conn}
      Dual graph $D(G)$ has no loops, 2-connected and planar iff. graph $G$ has no loops, 2-connected and planar.
\end{lm}
\begin{proof}
    Let's prove \emph{sufficient} condition. The \emph{necessary} condition will hold by the lemma~\ref{lm:dualsquare}. 

    Assume the contrary: the dual graph $D(G)$ has a vertex $v$ such that it becomes not connected after deleting $u$. Consider two connected components $G_1$ and $G_2$ in $D(G)$. Let's $M_1$ and $M_2$ be induced subgraphs on vertices $V(G_1)\bigcup \{v\}$ and $V(G_2)\bigcup \{v\}$ respectively. These subgraphs $M_1$ and $M_2$ should contain cycles $f_1$ and $f_2$ respectively, otherwise one of them is a tree and it holds a contradiction using the lemma~\ref{lm:dualsquare}. Let's denote the vertices of $D\bigl(D(G)\bigr) = G$ corresponded to faces $f_1$ and $f_2$ by $u_1\in G$ and $u_2\in G$ (see figure~\ref{fig:2conn3conn}). Let's denote the vertex corresponded to outer face of $D(G)$ by $u$. Since any inner faces of $M_1$ have no adjacent edges with any inner face of $M_2$, any path from vertex $u_1$ to $u_2$ should contain $u$. This holds a contradiction with 2-connectivity of $G$.
\end{proof}

\begin{lm}\label{lm:3conn}
    Dual graph $D(G)$ is simple, 3-connected and planar iff. graph $G$ is simple 3-connected and planar.
\end{lm}
\begin{proof}
    The proof for this theorem is the same as for the previous lemma, but in this case one should consider two faces: outer face that corresponds to the vertex $u\in G$ and \textquote{middle} face $f$ that corresponds to the vertex $w\in G$ (see figure~\ref{fig:2conn3conn}).
\end{proof}

\begin{figure}[H]
\vspace{-12pt}
    \centering
	\includegraphics[width=0.8\textwidth]{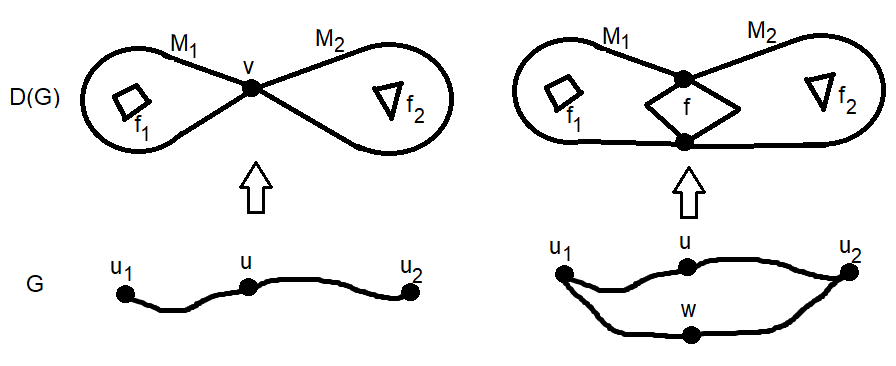}
	\caption{Graph $G$ and dual graph $D(G)$ for lemmas~\ref{lm:2conn} (left) and \ref{lm:3conn} (right).}\label{fig:2conn3conn}
\end{figure}

\subsection{Steinitz theorem}\label{subsec:st}

Let's start with several definitions and lemmas:

\begin{defin}
    \bf{Contraction of an edge $(u,v)$} is the operation that transform the graph $G$ to the graph $G'$ by merging vertices $u$ and $v$ to a new vertex $w\in V(G')$: $V(G') = \bigl(V(G)\setminus \{u, v\}\bigr)\bigcup \{w\}$ (see figure~\ref{fig:op}).
\end{defin}

\begin{defin}
    \bf{Minor of a graph $G$} is a graph that is obtained from $G$ by contraction and deleting operations of edges.
\end{defin}

\begin{defin}
    \bf{In series} reduction is a contraction of an edge adjacent to the vertex of degree two. \bf{In parallel} reduction is deleting multiple edges between to vertices. \bf{SP-reduction} of graph is a sequence of series and parallel reductions obtained for this graph (see figure~\ref{fig:op}).
\end{defin}

\begin{figure}[H]
\vspace{-12pt}
    \centering
	\includegraphics[width=1.0\textwidth]{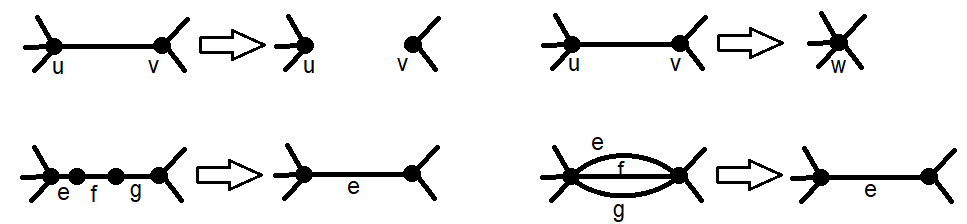}
	\caption{Deleting (left-top) and contraction (right-top) operations. In series (left-bottom) and in parallel (right-bottom) reductions.}\label{fig:op}
\end{figure}

\begin{defin}
    \bf{Grid graph $G(k, l)$} is the graph on the plane with vertices coordinates $\{(a,b):0\leq a\leq k-1, 0\leq b\leq l-1\}$ such that two vertices are connected if either their $x$-coordinates differ by 1 and $y$ coincide or $y$-coordinates differ by 1 and $x$ coincide.  
\end{defin}

\begin{lm}\label{lm:minorgrid}
    Any planar graph is a minor of a grid graph. 
\end{lm}
\begin{proof}
    Consider some representation of a planar graph $G$ on the plane $\rR^2$. We can split any vertex such that $G$ will have a vertices of maximum degree 4 (see figure~\ref{fig:spl}). Let's $U_r(v)$ be a disk in $\rR^2$ with center in $v$ and radius $r$, $\partial U_r(v)$ be the boundary of the disk and $U_r(e) = \bigcup\limits_{p\in e} U_r(p)$, where $p\in \rR^2$ are points of an edge $e$. 

\begin{figure}[H]
\vspace{-12pt}
    \centering
	\includegraphics[width=0.6\textwidth]{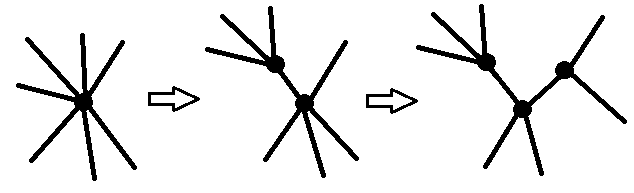}
	\caption{Splitting in small neighbourhood of vertex with degree more than 4 to vertices which has maximum degree 4 for lemma~\ref{lm:minorgrid}.}\label{fig:spl}
\end{figure}

    Let's find $\eps_1, \eps_2$ and $\eps_3$ such that:
    \begin{enumerate}
        \item Consider for every vertex $v_j$ a point $v'_j\in U_{\eps_3}(v_j)$ with rational coordinates. For any $j$: $U_{\eps_3}(v_j)\sub U_{\eps_1}(v'_j)$.
        \item For any $j$ the number of intersections of $\partial U_{\eps_1}(v'_j)$ and an edge $e$ should be equal to 1 for $e$ adjacent to $v_j$ and 0 otherwise.
        \item For any edge $e$ all regions $U_{\eps_2}(e)\setminus\Bigr(\bigcup\limits_{j=1}^n U_{\eps_1}(v'_j)\Bigl)$ are connected and disjoint.
    \end{enumerate}

    Since $U_{\eps_3}(v_j)\sub U_{\eps_1}(v'_j)$, the vertex $v_j\in U_{\eps_1}(v'_j)$. Let's denote by $e_{i,j}$ the intersections of edges $e_i$ and $\partial U_{\eps_1}(v'_j)$. Consider the points with rational coordinates $p_{i,j}\in \partial U_{\eps_1}(v'_j)$ such that $p_{i,j}\in U_{\eps_2}(e_{i,j})$. Since all vertices have the maximum degree 4, any $\partial U_{\eps_1}(v'_j)$ can contain maximum 4 different points $p_{i,j}$.
    
    Let's scale the coordinates of $\rR^2$ by multiplication for the product of denominators of all corresponding points. Thus all corresponded points will have integer coordinates. For any region corresponded to $U_{\eps_1}(v'_j)$ there exists the natural number $C_j$ for scaling coordinates of $\rR^2$ such that after scaling there exist less or equal four vertex independent paths from $v'_j$ to $p_{i,j}\in \partial U_{\eps_1}(v'_j)$ with consequent vertices in $U_{\eps_1}(v'_j)$ with either their $x$-coordinates differ by 1 and $y$ coincide or $y$-coordinates differ by 1 and $x$ coincide. There exist natural numbers $D_i$ with the same properties for any region $U_{\eps_2}(e_i)\setminus\Bigl(\bigcup\limits_{j=1}^n U_{\eps_1}(v'_j)\Bigr)$ and two vertices $p_{i,j}$ and $p_{i,k}$ corresponded to the edge $e_i$. By scaling, finally for the product of all $C_j$ and $D_i$ and shifting coordinates one can obtain required grid graph.
\end{proof}

\begin{defin}
    \bf{$\Delta Y$ operation} is an operation replacing a triangle that bounds a face by 3-star that connects the same vertices or vice versa (see figure~\ref{fig:dyop}). If a triangle transforms to 3-star it is called \bf{$\Delta$-to-$Y$ operation} and the reverse is called \bf{$Y$-to-$\Delta$ operation}.
\end{defin}

\begin{figure}[H]
\vspace{-12pt}
    \centering
	\includegraphics[width=0.6\textwidth]{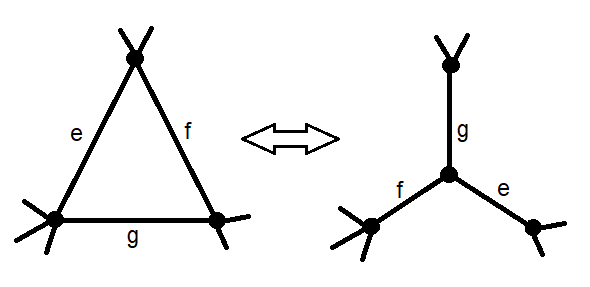}
	\caption{$\Delta Y$ operation.}\label{fig:dyop}
\end{figure}

Because $\Delta Y$ operations can cause multi-edges in a new graph, thus in this subsection all graphs can be considered as not simple by default.  

\begin{lm}\label{lm:ytod}
    \begin{enumerate}
        Let $v\in G$ is a vertex of degree 3 that has three non-parallel adjacent edges $(e, f, g)$.
        \item If $G$ has no loops, planar and 2-connected, then the result of $Y$-to-$\Delta$ operation to $(e, f, g)$ and then series of $SP$-reductions also has no loops, planar and 2-connected.
        \item Let $G$ be simple, planar, 3-connected graph and not $K_4$, then the result of $Y$-to-$\Delta$ operation to $(e, f, g)$ and then series of $SP$-reductions is also simple, planar, 3-connected.
    \end{enumerate}
\end{lm}
\begin{proof}
    By Mengers theorem~\ref{thm:menger} there exist 2 or 3 vertex independent paths between any two vertices in $G$. Paths which don't contain vertex $v$ after $Y$-to-$\Delta$ operation will not change. There can be only one path from these independent sets that contains vertex $v$ and thus edges $e$ and $f$ without loss of generality. After $Y$-to-$\Delta$ operation this path will walk through the new edge $g$.
\end{proof}

\begin{lm}\label{lm:dtoy}
    \begin{enumerate}
        Let $(e, f, g)$ be three non-parallel edges which form a triangle.
        \item If $G$ has no loops, planar and 2-connected, then the result of $\Delta$-to-$Y$ operation to $(e, f, g)$ and then series of $SP$-reductions is also has no loops, planar and 2-connected.
        \item Let $G$ be simple, planar 3-connected graph, then the result of $\Delta$-to-$Y$ operation to $(e, f, g)$ and then series of $SP$-reductions is also simple, planar and 3-connected.
    \end{enumerate}
\end{lm}
\begin{proof}
    Let's note that $Y$-to-$\Delta$ operation for a graph $G$ corresponds to $\Delta$-to-$Y$ for the dual graph $D(G)$ and vice versa. Also in series reduction for the graph $G$ corresponds to in parallel reduction for the dual graph $D(G)$ and vice versa. The multi-edges correspond to the object in dual graph that can be reduced by in series reduction to just one edge. By using this facts and lemmas~\ref{lm:2conn} and \ref{lm:3conn} this lemma is led to the previous one. 
\end{proof}

\begin{defin}
    A 2-connected graph $G$ is \bf{$\Delta Y$ reducible} to the graph $G'$ if it can be transformed to the graph $G'$ by sequence of $\Delta Y$ operations and $SP$-reductions.
\end{defin}

\begin{lm}\label{lm:grid2k4}
    Any grid graph $G(k,l): k, l\geq 3$ is $\Delta Y$ reducible to the graph $K_4$.
\end{lm}
\begin{proof}
    Let's use two following operations (figure~\ref{fig:twoop}) in the grid graph $G(k,l)$.

\begin{figure}[H]
\vspace{-5pt}
    \centering
	\includegraphics[width=0.55\textwidth]{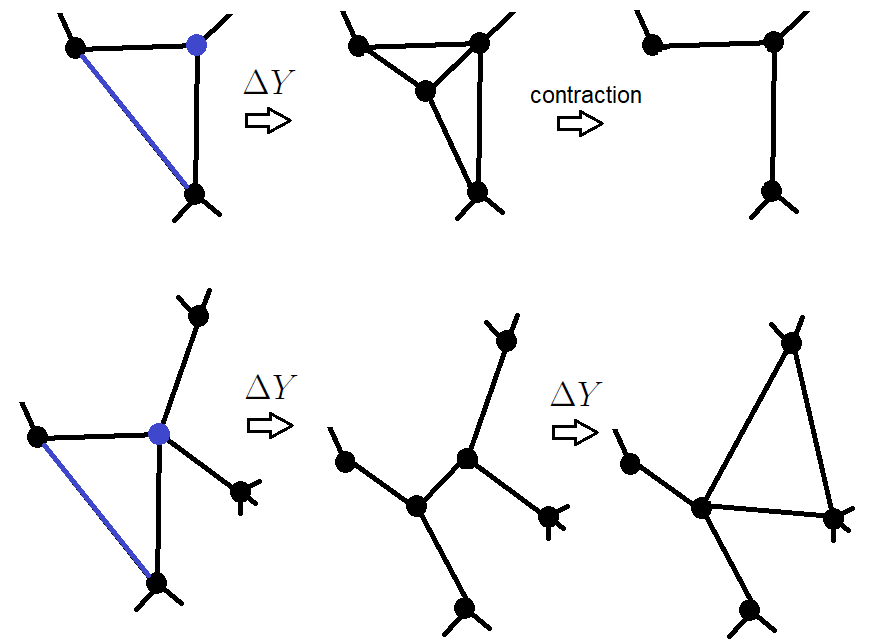}
	\caption{Two operations in lemma~\ref{lm:grid2k4}: \emph{operation 1} for triangle with adjacent vertex with degree 3 (top) and \emph{operation 2} for triangle with adjacent vertex with degree 4 (bottom).}\label{fig:twoop}
\end{figure}

First, let's use in series reduction to transform the left-bottom and right-top squares to triangles. Then, transfer the left-bottom triangle to the top row by series of operation 2. If by this the multi-edges are obtained, use in parallel reduction, if it is not, use operation 1 to delete transferred triangle. Use this procedure for second square in the bottom row. By this procedure one can delete all squares of the bottom row. The symmetrical procedure can be used for deleting columns from left to right. Therefore, this grid graph can be reduced to the grid graph $G(2,2)$. See figure~\ref{fig:grid2k4} to obtain the graph $K_4$ from $G(2,2)$.

\begin{figure}[H]
    \centering
	\includegraphics[width=0.9\textwidth]{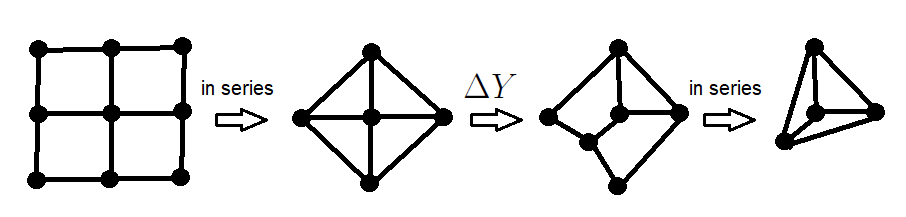}
	\caption{The series of reductions to obtain the graph $K_4$ from $G(2,2)$ for lemma~\ref{lm:grid2k4}.}\label{fig:grid2k4}
\end{figure}
    
\end{proof}

Let's $C_2$ be a circle with two vertices and two multi-edges between them.

\begin{lm}\label{lm:g2c2}
    If a planar graph with no loops $G$ is $\Delta Y$ reducible to the graph $C_2$, then so is for every 2-connected minor without loops of $G$.
\end{lm}
\begin{proof}
    Let's prove it by induction for the number of $\Delta Y$ reductions necessary to obtain the graph $C_2$:
    
    The base: any 2-connected minor of $C_2$ is $C_2$.
    
    Now, let for every $G$ that can be $\Delta Y$ reduced to the graph $C_2$ by $k$ steps, every 2-connected minor is also $\Delta Y$ reducible to $C_2$. Let's prove it for $G$ that can be $\Delta Y$ reduced to the graph $C_2$ by $k+1$ steps. Consider the graph $G'$ after one step of $\Delta Y$ reduction. Let's $H$ be a 2-connected minor without loops of $G$.   

        \begin{enumerate}
            \item If $H$ doesn't contain the part that is changed by current $\Delta Y$ reduction from $G$ to $G'$ then $H$ is 2-connected minor of $G'$. Therefore, $H$ is $\Delta Y$ reducible to $C_2$ by the inductive assumption.
            \item If $H$ contains whole part that is changed by current $\Delta Y$ reduction from $G$ to $G'$, then this reduction can be used in $H$ to obtain the graph $H'$. By the lemmas~\ref{lm:ytod} and \ref{lm:dtoy} $H'$ will be 2-connected minor without loops of $G'$. Therefore, $H'$ and hence $H$ are $\Delta Y$ reducible to $C_2$.
            \item If $H$ doesn't contain whole part, then the current reduction is $\Delta$-to-$Y$ or $Y$-to-$\Delta$. Let's denote the corresponding edges by $(e,f,g)$, and $e$ or $e, f$ will be the part of the graph $H$ without loss of generality. Let's consider cases:
            \begin{enumerate}
                \item The current reduction is $\Delta$-to-$Y$. 
                \begin{enumerate}
                    \item If $H$ can be obtained from $G$ by one contraction of the edge $g$ then $H$ will be the minor of $G'$ by contraction of corresponded edges $e, f$ in $G'$. 
                    \item If $H$ can be obtained from $G$ by contraction of two edges or more, then $H$ will contain a loop. This holds a contradiction.
                    \item If $H$ can be obtained from $G$ by one deletion of the edge $g$ then $H$ will be the minor of $G'$ by contraction the corresponded edge $g$ in $G'$.
                    \item If $H$ can be obtained from $G$ by deletion of the edges $f, g$ then $H$ will be the minor of $G'$ by the deletion of the corresponded edge $e$ and contraction the corresponded edge $f$ or $g$ in $G'$.
                \end{enumerate}
                 \item The current reduction is $Y$-to-$\Delta$. 
                 \begin{enumerate}
                     \item If $H$ can be obtained from $G$ by contraction of any edges then $H$ will contain a loop. This holds a contradiction.
                     \item If $H$ can be obtained from $G$ by deletion of one edge $g$ then let's $H'$ be the graph that is obtained by $H$ by in series reduction of the part $\{e, f\}$. Since $H'$ is a minor of $G'$, the minor $H'$ and hence $H$ are $\Delta Y$ reducible to $C_2$. 
                     \item If $H$ can be obtained from $G$ by deletion of two or more edges then $H$ will be not 2-connected. This holds a contradiction.
                 \end{enumerate}
            \end{enumerate}
        \end{enumerate}
        This ends the proof.
\end{proof}

\begin{cor}
    Any simple 3-connected planar graph $G$ is $\Delta Y$ reducible to the graph $K_4$. 
\end{cor}
\begin{proof}
    By the lemma~\ref{lm:minorgrid} $G$ is a minor of a grid graph $G(k,l)$. This grid graph $G(k,l)$ is $\Delta Y$ reducible to the graph $K_4$ by the lemma~\ref{lm:grid2k4}. Let's note that $K_4$ is also $\Delta Y$ reducible to $C_2$. Hence, $G$ is $\Delta Y$ reducible to $C_2$ by the lemma~\ref{lm:g2c2}. 
    
    Let's prove that $G$ is $\Delta Y$ reducible to $K_4$ by the induction of the number edges in $G$:
    \begin{enumerate}
        \item The base. A simple 3-connected planar graph with the smallest number of edges is $K_4$ by the lemma~\ref{lm:kld} for $(k,l,d)$ graphs.
        \item  Let's any simple 3-connected planar graph with $k$ edges is $\Delta Y$ reducible to $K_4$. Let's prove it for $G$ with $k+1$ edges. Consider the $\Delta Y$ reduction of $G$ to $C_2$. Since $C_2$ has only two edges, there will be step with $SP$-reduction in $\Delta Y$ sequence. Let's $G'$ be the graph that is obtained from $G$ after the first $SP$-reduction. By the lemma~\ref{lm:3conn} $G'$ is simple, planar, 3-connected and has less edges than $G$. Then $G'$ and hence $G$ are $\Delta Y$ reducible to $K_4$ by the inductive assumption.
    \end{enumerate}
\end{proof}

\begin{thm}[Steinitz]\label{thm:st}
    Let's $G$ be simple connected graph. Then, $G$ is planar and 3-connected iff. it is graph of a convex polyhedron. 
\end{thm}
\begin{proof}\ 
\begin{enumerate}
    \item \emph{Necessarity.} Let's prove the following: let $G'$ be a graph that obtained from $G$ by one $\Delta Y$ operation and $SP$-reduction and $G'$ be a graph of a convex polyhedron $P(G')$ then $G$ is also a graph of a convex polyhedron $P(G)$. Let's consider cases:
    \begin{enumerate}
        \item If the current $\Delta Y$ operation is $\Delta$-to-$Y$ operation then the polyhedron $P(G)$ can be obtained from the polyhedron $P(G')$ by the section of the plane $\pi$ (see figure~\ref{fig:dtoypoly}).  

\begin{figure}[H]
    \centering
	\includegraphics[width=0.9\textwidth]{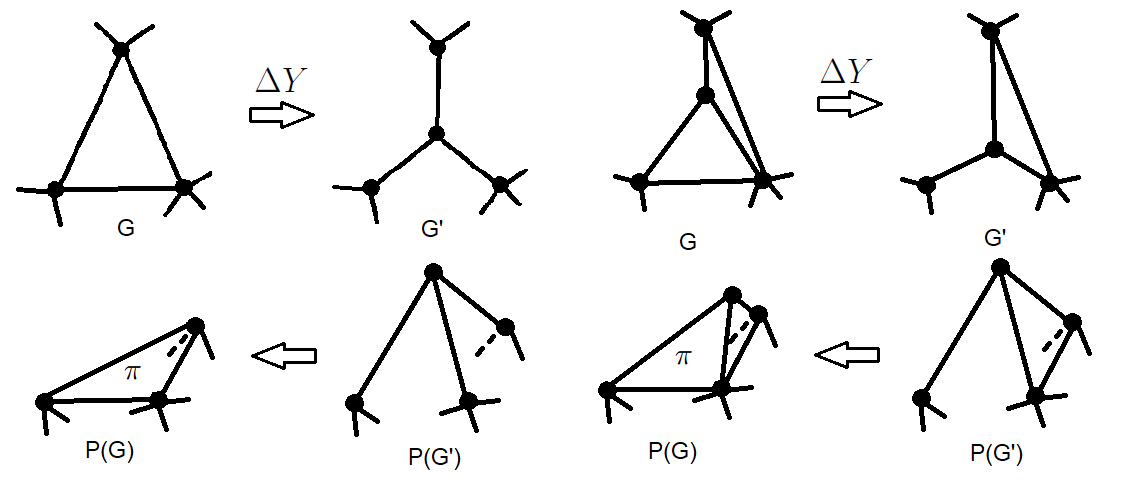}
	\caption{$\Delta$-to-$Y$ operation without $SP$-reduction (left-top) and corresponded section of the polyhedron $P(G')$ by the plane $\pi$ (left-bottom). $\Delta$-to-$Y$ operation with $SP$-reduction (in series, right-top) and the corresponded section of the polyhedron $P(G')$ by the plane $\pi$ (right-bottom).}\label{fig:dtoypoly}
\end{figure}

    \item\label{enum:st0} If the current $\Delta Y$ operation is $Y$-to-$\Delta$ operation. Consider the corresponded triangle $T\in P(G')$. This triangle is adjacent to three faces $F_1, F_2, F_3$. If faces $F_1, F_2$ and $F_3$ intersect in one point then the polyhedron $P(G')$ is tetrahedron, $G'$ is $K_4$ and polyhedron $P(G)$ can be constructed by gluing two tetrahedrons to each other. Thus $P(G)$ is convex and corresponds to the graph $G$.
    
    Now let's faces $F_1, F_2$ and $F_3$ don't intersect in one point. Then, do the reverse operation like in figure~\ref{fig:dtoypoly} left, but with some deformation: compress the triangle in the face $T$ and deform $P(G')$ such that planes corresponded to $F_1, F_2$ and $F_3$ become intersected in one point (see figure~\ref{fig:deformp}) and extend current faces to this intersection. Obtained polyhedron will be convex and correspond to the graph $G$. 

\begin{figure}[H]
    \centering
	\includegraphics[width=0.25\textwidth]{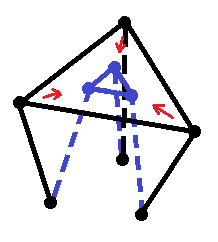}
	\caption{Deformation of the polyhedron $P(G')$ corresponded to the case~\ref{enum:st0} of the Steinitz theorem.}\label{fig:deformp}
\end{figure}

    \end{enumerate}
    Let's note that in each case all transformations been done preserve convexity of the polyhedron.
    \item \emph{Sufficiency.} In the Euler theorem~\ref{thm:eulerpoly} the planarity property was proved. Let's prove that by deletion of any two vertices $u, v$ the graph $G$ remains connected.
    \begin{enumerate}
        \item Consider the case when $u,v$ are vertices of the common face $\pi$. Let's prove that for any two vertices $\notin\pi$ there exists path between them that doesn't intersect the face $\pi$. Consider the path $P$ between any two vertices in $G$. Let's transform this path by following:
        \begin{enumerate}
            \item\label{thm:st1} Consider the case when $P$ contains an edge $e\in\pi$. Let's denote the another adjacent face (not $\pi$) by $F$. Transform the path $P$ to $(P\setminus e)\bigcup(F\setminus e)$. Since faces $F$ and $\pi$ can intersects only by one edge (property~\ref{poly:ev} of polyhedrons), the number of edges in the path $P$ which belongs to $\pi$ will decreased by 1. Do this for every edge in $P\bigcap\pi$. 
            \item\label{thm:st2} If $P$ contains a vertex $v\in\pi$. Consider the multi-angle of the polyhedron corresponded to vertex $v$ (or all faces which contain $v$). Let's denote the adjacent to $v$ edges in $P$ by $e_1$ and $e_n$. For the edge $e_1$ there exists a face $F_1\neq\pi$ that contains $e_1$. Let's do the transformation like in case~\ref{thm:st1} for the edge $e_1$. By this transformation the edge $e_1$ transforms to the path through the face $F_1$. Let's denote the new edge of $F_1$ adjacent to the vertex $v$ by $e_2$ and continue this procedure for $e_2$ and so on. In the end transform the path $P$ to $\bigl(P\setminus \{e_{n-1}, e_n\}\bigr)\bigcup\bigl(F_{n-1}\setminus\{e_{n-1}, e_n\}\bigr)$.
        \end{enumerate}
        By this transformations one can modify the path between any two vertices $\notin\pi$ by a path that doesn't intersect the face $\pi$. Therefore, there exists path between any two vertices in $G\setminus\{u,v\}$.
        \item Let's vertices $u, v$ don't belong to the same face. Consider the path between any two vertices that contains the vertex $u$ without loss of generality. Transform the path using procedure in the case~\ref{thm:st2}. Since vertices $u, v$ don't belong to the same face, the faces $F_1, F_2, ..., F_{n-1}$ don't contain the vertex $v$ and hence, the new transformed path doesn't contain the vertices $u$ and $v$. Therefore, there exists a path between any two vertices that doesn't contain the vertices $u$ and $v$.
    \end{enumerate}
\end{enumerate}
\end{proof}

\subsection{Dual polyhedrons}

\begin{defin}
    \bf{Dual polyhedron $D(P)$} for a convex polyhedron $P$ is defined by following:
    \begin{enumerate}
        \item Consider the 3-dimensional coordinates of vertices of polyhedron $P$. Each face $v_1 v_2 ... v_k \in P$ match with the vertex $u = \frac 1 k \sum_{i = 1}^k v_i.$
        \item If any two faces of $P$ are adjacent then add an edge in the dual graph $D(P)$ between two vertices which are matched to these faces.
    \end{enumerate}
\end{defin}

\NB A dual polyhedron is unique by this construction.

\NB The graph of dual polyhedron $D(P)$ is dual to the graph of a polyhedron $P$.

\begin{lm}\label{lm:dualint}
    $int\bigl(D(P)\bigr)\sub int(P)$.
\end{lm}
\begin{proof}
    By using convex property:
    $$\forall u, v\in V\bigl(D(P)\bigr)\Rightarrow u, v\in\d P\Rightarrow (u,v)\sub int(P) \Rightarrow D(P) \sub P \Rightarrow int\bigl(D(P)\bigr)\sub int(P).$$
\end{proof}

\begin{lm}\label{lm:dpc}
    The dual polyhedron is convex.
\end{lm}
\begin{proof}
    Since the graph of dual polyhedron $D(P)$ is dual to the graph of a polyhedron $P$, this matching can be continued in the same way like in the lemma~\ref{lm:dualsquare} to the matching between faces and multi-angles of the polyhedron $P$ and the dual polyhedron $D(P)$ such that edges of faces are corresponded to edges of multi-angles, and vertices of faces are corresponded to faces of the multi-angles (see figure~\ref{fig:dualpoly}). 
    
    Now let's assume the contrary: there exist two faces $F_1, F_2\in D(P)$ such that the edge $(u,v):u\in F_1, v\in int\bigl(D(P)\bigr)$ intersect the face $F_2$. The vertex $v\in int(P)$ by lemma~\ref{lm:dualint}. By the correspondence below there are a multi-angles at vertices $h_1$ and $h_2$ which match to the faces $F_1$ and $F_2$ (see figure~\ref{fig:dualpoly} left). Hence, this edge $(u, v)$ intersects with the polyhedron $P$ and $u\in P$ and $v\in int(P)$. This holds a contradiction.
\end{proof}

\begin{figure}[H]
    \centering
    \vspace{-10pt}
	\includegraphics[width=0.8\textwidth]{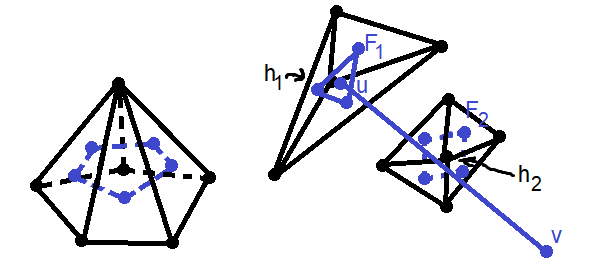}
	\caption{For the lemma~\ref{lm:dpc}: correspondence between a face of dual polyhedron $D(P)$ (blue, left) and a multi-angel of polyhedron $P$ (black, left). Two faces $F_1$ and $F_2$, the edge $(u,v)$ of dual polyhedron $D(P)$ (blue, right) and corresponded multi-angels at vertices $h_1$ and $h_2$.}\label{fig:dualpoly}
\end{figure}

Consider the duality relationship for platonic solids (see figure~\ref{fig:dgplatonic}):

\begin{enumerate}
    \item Tetrahedron $\rightarrow$ Tetrahedron,
    \item Octahedron $\rightarrow$ Cube,
    \item Icosahedron $\rightarrow$ Dodecahedron,
    \item Cube $\rightarrow$ Octahedron,
    \item Dodecahedron $\rightarrow$ Icosahedron.
\end{enumerate}

\begin{figure}[H]
    \centering
    \vspace{-12pt}
	\includegraphics[width=0.78\textwidth]{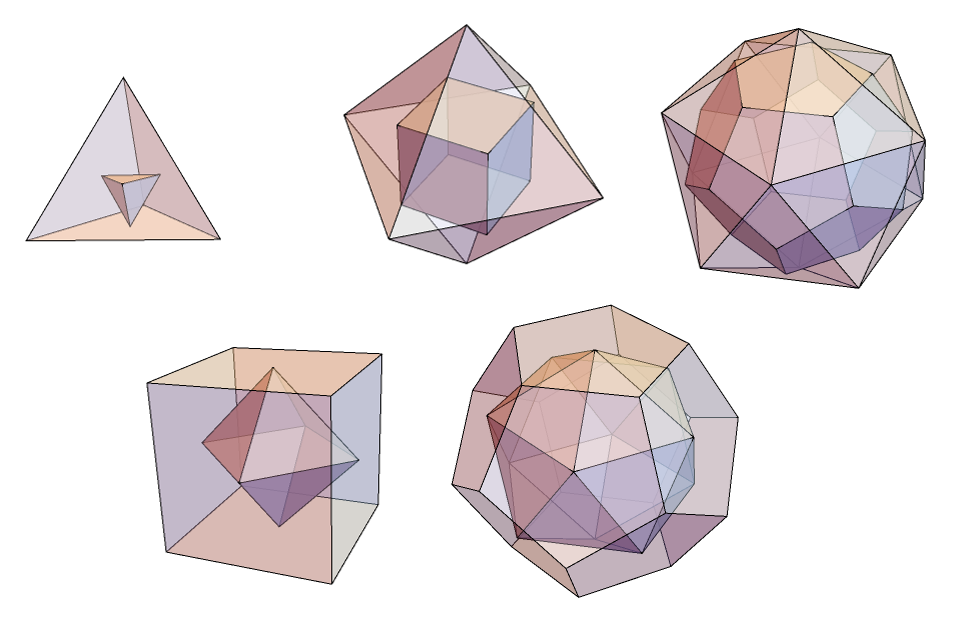}
	\caption{Dual polyhedrons for platonic solids.\label{fig:dgplatonic}}
\end{figure}

\begin{thm}[Whitney]
    Any planar 3-connected graph has unique dual graph up to isomorphism. 
\end{thm}
\begin{proof}
    Let's assume the contrary: a planar 3-connected graph $G$ has two non-isomorphic dual graphs $D_1$ and $D_2$. By the Steinitz theorem~\ref{thm:st} any 3-connected graph $G$ is a graph of a convex polyhedron $P(G)$. Any dual graph $D(G)$ is 3-connected by the lemma~\ref{lm:3conn} thus it is also a graph of a convex polyhedron $P\bigl(D(G)\bigr)$. Since $D_1$ and $D_2$ are not isomorphic, the graphs of two polyhedrons $P(D_1)$ and $P(D_2)$ are not isomorphic. By the lemma~\ref{lm:dualsquare} graphs of dual polyhedrons $D\bigl(P(D_1)\bigr)$ and $D\bigl(P(D_2)\bigr)$ are isomorphic to the graph $G$ and thus isomorphic. Let's denote this isomorphism by $i_1:D(P\bigl(D_1)\bigr)\rightarrow D\bigl(P(D_2)\bigr)$. By matching in the lemma~\ref{lm:dualsquare} this isomorphism defines isomorphism of their duals: $i'_1:D^2\bigl(P(D_1)\bigr)\rightarrow D^2\bigl(P(D_2)\bigr)$, and dual polyhedron $D^2\bigl(P(D_1)\bigr)$ is isomorphic to $P(D_1)$ and $D^2\bigl(P(D_2)\bigr)$ is isomorphic to $P(D_2)$. Let's denote these isomorphisms by $i_2:D^2\bigl(P(D_1)\bigr)\rightarrow P(D_1)$ and $i_3:D^2\bigl(P(D_2)\bigr)\rightarrow P(D_2)$. The composition $i_3\circ i'_1\circ i_2^{-1}$ is isomorphism $P(D_1)\rightarrow P(D_2)$.
\end{proof}

\subsection{Algorithms}\indent

\bf{Dual graph constructing algorithm}.

This algorithm constructed the dual graph $D(G)$ corresponding to a planar representation of a connected graph $G$.\\

\emph{Description:}
\begin{enumerate}
    \item Initialize 2-dimensional array $a = \bigl[a_1[v], a_2[v]\bigr]$ corresponded to edges of a dual graph $D(G)$ with dimensions $2\times\|E(G)\|$.    
    \item Find all bridges using bridge-detecting algorithm.
    \item \label{enum:dgalg1}  $e\in G$ and use any algorithm of searching shortest path $P$ in $G\setminus e$ between two vertices adjacent to the edge $e$ to find a simple cycle $C = P_1\bigcup e$ in $G$.
    \item Check that this cycle is new cycle by checking whether all $a[j]$ corresponded to the edges of $C$ don't contain the same vertex. If it is not new cycle find another shortest path $P_2$ (by deleting one of adjacent edges to $e$ in $P_1$) in $G\setminus e$ again and change $C = P_2\bigcup e$. 
    \item\label{enum:dgalg2} For every edge $e_j\in C$ add new vertex $v$ either in $a_1[j]$ or in $a_2[j]$ if $a_1[j]$ is not empty.
    \item Do~\ref{enum:dgalg1}-\ref{enum:dgalg2} for every another edge in $G$.
    \item For each empty $a_2[j]$ add a vertex $v_0$ that corresponds to the outer face of the graph. 
    \item For every not assigned bridge $b$ find the path $P$ from this bridge to a vertex that contains an edge $e_i$ with not empty $a_2[i]$. For every bridge $e_j\in P\bigcup{b}$ add in $a_1[j]$ and $a_2[j]$ the value $a_2[i]$.
\end{enumerate}

\NB The complexity of this algorithm using breadth first search or depth first search for connected graphs equals to $\mathbf{O\bigl(\|E\|^2\bigr)}$.\\

\begin{defin}
    The dual graph $D(G)$ is called \bf{simple*} iff. it is simple after the deletion of vertex corresponded to the outer face of
    $G$.
\end{defin}\noindent
Consider the algorithm corresponded to simple* dual graphs:\\

\bf{Lee algorithm.}

This algorithm finds the shortest path through inner faces between any two inner faces of a planar representation of a graph with simple* dual graph.\\

\emph{Description:}
\begin{enumerate}
    \item Construct the dual graph $D(G)$ for the current graph $G$ and delete the vertex corresponded to outer face of the current graph $G$.
    \item Mark the vertices corresponded to the same depth by this depth (or distance from starting point) using breadth first search algorithm.
    \item Algorithm terminates when the end point is marked.
\end{enumerate}

To find the shortest path walk through the vertices corresponded to reverse order of marked numbers (from ending point to starting, see figure~\ref{fig:lee}).\\

\NB The complexity of this algorithm (without constructing dual graph) equals to the complexity of breadth first search for connected graphs $\mathbf{O\bigl(\|E\|\bigr)}$.

\begin{figure}[H]
    \centering
	\includegraphics[width=0.5\textwidth]{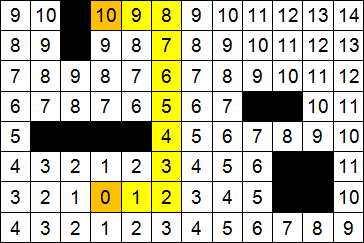}
	\caption{Result of the Lee algorithm.\label{fig:lee}}
\end{figure}

\emph{Modifications of Lee algorithm:}
\begin{enumerate}
    \item In this algorithm any constructing shortest path algorithm can be used. The complexity will be equal to the complexity of corresponded algorithm.
    \item One can consider different time travelling (and also infinite time travelling) from one face to another in the planar representation. Thus, the additional weights can be added in the dual graph. And in this case the shortest path algorithms for weighted graphs can be used. 
    \item\label{enum:modlee} The breadth first search algorithm can be applied to the starting point and the ending point in parallel. In this case the shortest path can be constructed when these two algorithms intersect by a vertex.
    \item If one has some information about the dual graph or path (for example, it is known that the path should contain some vertex). Lee algorithm can be used from this vertex or done in parallel for the starting vertex, the ending vertex and the current vertex like in~\ref{enum:modlee}. 
\end{enumerate}


\section{Flows on graphs}\label{sect:flow}

\subsection{Decomposition of flows}

\begin{defin}
    A \bf{flow network} is directed graph $G$ with positive weights $w(u,v)$ with two highlighted vertices \bf{$s$ (source)} and \bf{$t$ (terminal)} and \bf{capacity function} $c:V(G)\times V(G)\rightarrow \rR^+$ defined by following:
    $$c(u,v) =
    \begin{cases}
        w(u,v) & (u,v)\in E(G)\\
        0 & \text{otherwise}
    \end{cases}.$$
\end{defin}

Let's consider a flow network $G$.

\begin{defin}
    A \bf{flow} is a function $f:V(G)\times V(G)\rightarrow\rR$ satisfied the following properties:
    \begin{enumerate}
        \item $f(u,v) = -f(v,u)$ (skew symmetry),
        \item $f(u,v)\leq c(u,v)$ (capacity constraint),
        \item $\forall u\in V(G)\setminus\{s,t\}: \sum\limits_{v\in V(G)} f(u,v) = 0$ (conservation of a flow).
    \end{enumerate}
    A \bf{flow $f$ value} is $F = \sum\limits_{v\in V(G)} f(s,v)$.
\end{defin}

\NB By skew symmetry a value $F$ of the flow $f$ can be less than 0.

\begin{st}[Sum of flows]\label{st:sumfl}
    Consider a flow network $G$ with two flows: $f_1$ and $f_2$. If $f_1(u,v)+f_2(u,v)\leq c(u,v)$ for any edge $(u,v)$ then $f_1+f_2$ also a flow and $F(f_1+f_2)= F_1+F_2$.
\end{st}
\begin{proof}
     It easy holds from definition. 
\end{proof}

\begin{thm}[Decomposition of flow]\label{thm:decflow}
     For a positive flow $f: F > 0$ there exist at most $\|E\|$ simple paths $P_i$ from $s$ to $t$ and simple cycles $C_i$ such that $$f = \sumt_i f_{P_i}+ \sumt_i f_{C_i} \text{ and } F = \sumt_i f(P_i),$$ where $f(P) = \mint_{e\in P,\ f(e) > 0} f(e)$ and $$f_P(u,v) = 
    \begin{cases}
        f(P) & (u,v)\in P,\\
        -f(P) & (v,u)\in P,\\
        0 & \text{otherwise.}
    \end{cases}
    $$
\end{thm}
\begin{proof}
    Since the flow $f$ is positive, there exists an edge $(s, v_1)$ with positive flow $f(s, v_1)$. By the skew symmetry $f(v_1, s) = - f(s, v_1)$ and thus by conservation of a flow there exists an edge $(v_1, v_2)$ with positive flow $f(v_1, v_2)$. Then the same holds for the edge $(v_1, v_2)$. If these procedure reaches $t$ then denote this path by $P_1$. If these procedure reaches some visited vertex then denote the corresponded visited simple cycle by $C_1$. Let's denote this cycle or path just by $P$.

    Consider a flow $f'(u,v) = 
    \begin{cases}
        f(P) & (u,v)\in P,\\
        -f(P) & (v,u)\in P,\\
        0 & \text{otherwise.}
    \end{cases}
    $ 

    By decreasing our flow $f$ by this flow $f'$ new flow becomes 0 at at least one edge. Then all edges from the source $s$ will be with 0 flow. This procedure can be obtained for all the rest vertices with positive adjacent edges. Thus by doing it at most $\|E\|$ times we obtain the decomposition. 
\end{proof}

This theorem simple implies \bf{the decomposition algorithm}.\\

\NB Since the number of steps for constructing a cycle or a path in the theorem~\ref{thm:decflow} less or equal the number of vertices in the graph and the number of such cycles and paths is $O\bigl(\|E\|\bigr)$, the complexity of decomposition algorithm is $\mathbf{O\bigl(\|V\|\|E\|\bigr)}$.

\begin{defin}
    A \bf{residual network} $G_f = \bigl(V(G), E_f\bigr)$ is a flow network with edges $E_f$ with non-negative \bf{residual capacity} $c_f(u,v) = c(u,v)-f(u,v)$. 
\end{defin}

\begin{lm}[Two flows]\label{lm:2fl}
    Let $f$ and $h$ be two flows in a flow network $G$ then $h$ can be represented as a sum $f+f'$, where $f'$ is a flow in the residual network $G_f$ with the value $F' = H-F$.
\end{lm}
\begin{proof}
     It easy holds from definition. 
\end{proof}

\begin{lm}[Subtraction of flows]\label{lm:flowminus}
    Let $f$ and $g$ be flows with equal values in a flow network $G$. Then $g = f+\sumt_i f_{C_i},$ where $f_{C_i}$ are flows in the residual network $G_f$ among cycles $C_i$.
\end{lm}
\begin{proof}
    By the previous lemma there exists a flow $f'$ in $G_f$ such that $g = f+f'$. Since $F' = G-F = 0$, the decomposition algorithm will obtain only cycles $C_i$ in the residual network $G_f$.
\end{proof}

\begin{defin}
    A \bf{$s,t$-cut $(A, B)$} is a distribution of the vertex set $V(G)$ by two parts $A,B$ such that:
    \begin{enumerate}
        \item $A\bigcup B = V(G)$,
        \item $A\bigcap B =$ \O,
        \item $s\in A,\ t \in B$.
    \end{enumerate}
\end{defin}

\begin{defin}
    A \bf{capacity of the $s,t$-cut $(A, B)$} is 
    $$C(A, B) = \sum\limits_{u\in A}\sum\limits_{v\in B} c(u,v).$$
\end{defin}

\begin{defin}
    A \bf{flow through the $s,t$-cut $(A, B)$} is 
    $$F(A, B) = \sum\limits_{u\in A}\sum\limits_{v\in B} f(u,v).$$
\end{defin}

\begin{lm}[$s,t$-cut flow]\label{lm:stcut}
    A flow through any $s,t$-cut is equal to the flow value.
\end{lm}
\begin{proof}
    By the skew symmetry: $\sum\limits_{u\in A}\sum\limits_{v\in A} f(u,v) = 0$.
    $$\sum\limits_{u\in A}\sum\limits_{v\in B} f(u,v) =\sum\limits_{u\in A}\sum\limits_{v\in B} f(u,v) + \sum\limits_{u\in A}\sum\limits_{v\in A} f(u,v) = \sum\limits_{u\in A}\sum\limits_{v\in V(G)} f(u,v) = $$
    $$ = \sum\limits_{v\in V(G)} f(s,v) + \sum\limits_{u\in A\setminus\{s\}}\sum\limits_{v\in V(G)} f(u,v)\overset{\text{conservation low}}{=} F.$$
\end{proof}

\begin{cor}
    The sum of flows from the source is equal to the sum of flows to the terminal.
\end{cor}

\subsection{Max-flow}\label{subsc:flowalg}

\begin{defin}
    \bf{Max-flow} is a flow with the maximum possible value in a flow network. 
\end{defin}

\NB The value of the max-flow should be non-negative.

\begin{lm}[Necessary and sufficient condition of positive max-flow value]\label{lm:nsmaxflow}
    The max-flow value is positive iff. there exists a path from the source to the terminal.
\end{lm}
\begin{proof}
    \begin{enumerate}
        \item \emph{Sufficiency.} Let $c_P$ be the minimum capacity value among the path $P$. Let's define a flow: 
        $$f(u,v) = 
        \begin{cases}
            c_P & (u,v)\in P,\\
            -c_P & (v,u)\in P,\\
            0 & otherwise.
        \end{cases}$$
        The flow value is $F= c_P > 0$ thus the max-flow value is positive.
        \item \emph{Necessarity.}  Let's assume the contrary: there is no path from the source to the terminal. Let's denote by $A$ a set of vertices which are reachable from the source and $B$ will be all the rest. Thus $c(A, B) = 0$ and by $s,t$-cut flow lemma~\ref{lm:stcut}: 
        $$0 < F = F(A, B) \leq c(A, B) = 0$$ holds contradiction.
    \end{enumerate}
\end{proof}

\NB For the max-flow a flow value is positive is equivalent to flow value is not zero.

\begin{lm}[About a path in residual network]\label{lm:resnetpath}
    A flow $f$ is max-flow iff. there is no path from the source to the terminal in the residual network $G_f$.     
\end{lm}
\begin{proof}
    \begin{enumerate}
        \item \emph{Necessarity.} Let $f$ be the max-flow. Assume the contrary: there exists a path $P$ from the source to the terminal in the residual network $G_f$. Let $c_{f,P}$ be the minimum capacity value in $G_f$ among the path $P$. Let's define a flow: 
        $$f'(u,v) = 
        \begin{cases}
            c_{f,P} & (u,v)\in P,\\
            -c_{f,P} & (v,u)\in P,\\
            0 & otherwise.
        \end{cases}$$
        The flow value $F' = c_{f,P}$. Since $f(u,v)+f'(u,v)\leq f(u,v)+c_f = c(u,v)$ the sum $f+f'$ is a flow and $F+F'> F$. This holds a contradiction with max-flow assumption.
        \item \emph{Sufficiency.} Assume the contrary: there exists max-flow $f'\neq f$. By the two flows lemma~\ref{lm:2fl} the flow $f'-f$ is a flow in residual network $G_f$ with the flow value $F'-  F> 0$. Thus the maximal flow on $G_f$ is positive and by the lemma~\ref{lm:nsmaxflow} there exists a path in $(G_f)_{f'-f} = G_{f'}$. This holds the contradiction by necessarity condition.
    \end{enumerate}
\end{proof}

\begin{defin}
    \bf{Min-cut} is a $s,t$-cut $(A,B)$ with the minimum capacity $(A,B)$. 
\end{defin}

\begin{thm}[Ford–Fulkerson]\label{thm:fordfulk}
    The max-flow value equals to the min-cut capacity value.
\end{thm}
\begin{proof}
    Let $f$ be the max-flow and $(A_{min}, B_{min})$ be the  min-cut. By the $s,t$-cut lemma~\ref{lm:stcut}: $F = F(A_{min}, B_{min})\leq c(A_{min}, B_{min})$.
    
    By the lemma about a path in residual network~\ref{lm:resnetpath} there is no paths between the source $s$ and the terminal $t$ in the residual network $G_f$. Let's denote by $A$ the subset of vertices of $G_f$ which are reachable from $s$ and by $B$ all the rest. 
    $$0 = c_f(A,B) = c(A,B)-f(A,B)\imp F = f(A, B) = c(A, B)\geq c(A_{min}, B_{min})\imp $$
    $$\imp F =  c(A_{min}, B_{min}).$$
\end{proof}

Consider algorithms for constructing max-flow:

\begin{enumerate}
    \item \bf{Ford–Fulkerson algorithm}.\\
This algorithm finds max-flow in a flow network $G$ from the source $s$ to the terminal $t$.\\
\emph{Description:}
\begin{enumerate}
    \item Initialize max-flow $f_{max} = 0$.
    \item\label{enum:ff1} Find any path $P$ from $s$ to $t$ in $G$.
    \item Find the minimum capacity $c_p$ among this path $P$.
    \item\label{enum:ff2} Construct the flow $f$ with $c_p$ among the path as in the lemma~\ref{lm:nsmaxflow}.
    \item Add new flow $f$ to the max-flow $f_{max}$. Do steps~\ref{enum:ff1}-~\ref{enum:ff2} for residual network $G_{f_{max}}$ and find a new flow $f$.
\end{enumerate}

\NB The complexity of this algorithm depends of a weights distribution of a network.

\NB In some special cases this algorithm may work very slow (see ex~\ref{ex:ff}) and moreover, not terminate. \\

\begin{ex}\label{ex:ff}
    Let's consider several steps of Ford–Fulkerson algorithm:

\begin{figure}[H]
    \centering
	\includegraphics[width=1.0\textwidth]{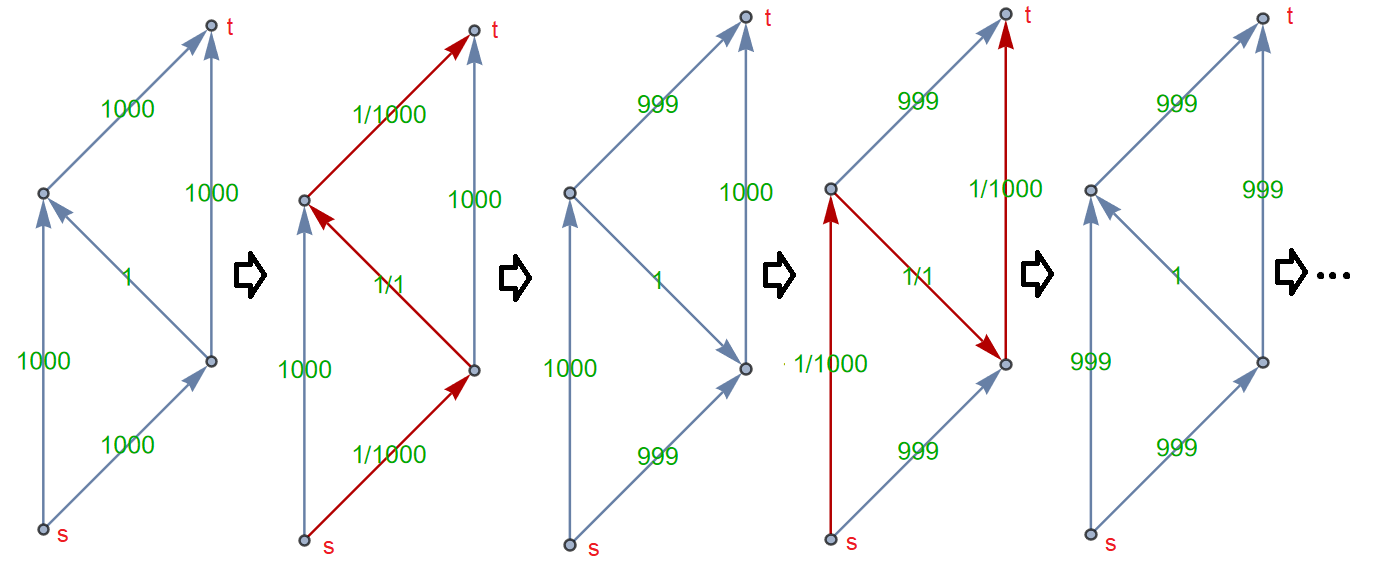}
	\caption{Steps of Ford–Fulkerson algorithm. First: the flow network $G$. Second: the flow $f$ versus capacity among the path (red). Third: the residual network $G_f$ and so on.}\label{fig:ff}
\end{figure}

Since the path construction algorithm is not specified, the path can be chosen like in this figure~\ref{fig:ff} and thus this algorithm will do about 2000 cycles. It is easy to see that by choosing \textquote{right} paths it can be terminated after 2 steps. 

\end{ex}

To specify the path construction algorithm for example BFS algorithm can be chosen and corresponded algorithm is called

\item \bf{Edmonds–Karp algorithm}.\\
This algorithm finds max-flow in a flow network $G$ from the source $s$ to the terminal $t$.\\
\emph{Description:}\\
This algorithm works the same as Ford–Fulkerson algorithm but at the step~\ref{enum:ff1} BFS algorithm is being used for searching the path with the minimum number of edges. 

\NB The complexity of this algorithm corresponds to the weights of a network and also may work very slow.

But the problem of complexity will stay almost the same. The next algorithm solves this problem:

\item \bf{Dinic algorithm.}

To explain next algorithm let's first give some definitions:

\begin{defin}
    A \bf{level graph} corresponded to the root $s$ of a directed unweight graph $G$ is a subgraph $G_L = \bigl(V(G), E_L\bigr)$ such that $(u,v)\in E_L\sub E(G) \iff \dist(s, v) = \dist(s, u)+1$. 
\end{defin}

\NB A level graph can be constructed using BFS algorithm starting from the root $s$. Note that $BFS(v) = \dist(s, v),\ \forall v\in V(G)$.

\begin{defin}
    A \bf{block flow} is the max-flow in the level graph $G_L$ corresponded to the source $s$.
\end{defin}

\begin{st}
    A level graph is acyclic graph and all paths from the root $s$ to any vertex $v$ have the same length.
\end{st}
\begin{proof}
    Easily holds from definition.
\end{proof}

Now let's consider the \bf{Dinic algorithm.}\\ \\
\emph{Description:}\\
This algorithm finds max-flow in a flow network $G$ from the source $s$ to the terminal $t$ by using a residual network for a block flow.

\begin{enumerate}
    \item\label{enum:din1} Construct $G_L$ for a graph $G$ without weights using BFS algorithm. If there is no path from $s$ to $t$ then terminates.
    \item\label{enum:din2} Construct a block flow $f$ in $G_L$ using Edmonds–Karp algorithm.
    \item Do \ref{enum:din1} and \ref{enum:din2} for the residual network $G_f$.
\end{enumerate}

\begin{thm}[Block flow]\label{thm:blockflow}
    Consider a block flow $f$. Let's denote the minimum number of edges from $s$ to a vertex $v$ in $G$ and in $G_f$ by $\dist(s, v)$ and $\dist'(s, v)$ respectively. Then, $\dist'(s,t)\geq \dist(s,t)+1$. 
\end{thm}
\begin{proof}
    Let's denote by $len(P)$ the length of a path $P$. First let's note that if there is a path $P'$ in $G_f$ from $s$ to $t$ then either $\len(P') > \dist(s,t)$ or it contains a new edge $e\notin E(G)$. Thus it is sufficient to prove that for any path $P'$ in $G_f$ from $s$ to $t$ that contains a new edge $e\notin E(G)$ the length $\len(P')\geq \dist(s,t)+1$.

    Let's prove more: if the path $P'$ in $G_f$ from $s$ to any vertex $v$ contains new edges $e_1, e_2, ..., e_k\notin E(G)$ then $\len(P')\geq \dist(s,v)+2k$. Let's prove it using induction by $k$:
    \begin{enumerate}
        \item The base. If the path $P'$ in $G_f$ from $s$ to a vertex $v$ contains a new edge $e_1 = (w_1,u_1)$ then the inverse edge $(u_1,w_1)$ belongs to $G$ and $f(u_1,w_1) = c(u_1,w_1)$ for the block flow $f$. By decomposition algorithm and previous statement there exists a path in the level graph $G_L$ that contains this edge $(u_1,w_1)$. Thus $\dist(s, w_1) = \dist(s, u_1)+1$.\\
        
        Let's denote the part of the path $P'$ from the vertex $a$ to $b$ by $P'_{ab}$. Since $P'_{s w_1}$ and $P'_{u_1 v}$ are paths in $G$, the length $\len(P'_{s w_1})\geq \dist(s, w_1)$ and $\len(P'_{u_1 v})\geq \dist(u_1, v)$. Thus $\len(P'_{s u_1})\geq \dist(u_1)+2\imp \len(P')\geq \dist(s, v)+2$.
        \item Let this statement holds for any $j\leq k-1$ then let's prove it for $j = k$. Since $e_k = (w_k,u_k)$ is a new edge, there exists a path in the level graph $G_L$ that contains the inverse edge $(u_k,w_k)$ and thus $\dist(s, w_k) = \dist(s, u_k)+1$. Since $P'_{u_k v}$ is a path in $G$, the length $\len(P'_{u_k v})\geq \dist(u_k, v)$.\\
        
        By the induction statement the length $\len(P'_{s w_k})\geq \dist(s, w_k)+2(k-1)$. Thus $\len(P'_{s u_k})\geq \dist(s, u_k)+2k\imp \len(P')\geq \dist(s,v)+2k$. 
    \end{enumerate}

    Hence, if a path $P'$ in $G_f$ from $s$ to $t$ contains a new edge $e\notin E(G)$ then the length $len(P')\geq \dist(s,t)+1$.
\end{proof}

\NB By the block flow theorem~\ref{thm:blockflow} the minimum path length at each step of the Edmonds–Karp algorithm increases at least by 1 and thus the number of steps is $O\bigl(\|V\|\bigr)$. By the decomposition theorem~\ref{thm:decflow} the number of steps in Edmonds–Karp algorithm for a level graph $G_L$ is $O\bigl(\|V\|\|E\|\bigr)$ (like in the decomposition algorithm). Thus the complexity of Dinic algorithm is $\mathbf{O\bigl(\|V\|^2 \|E\|\bigr)}$.\\

\NB One can improve the complexity of Dinic algorithm using the dynamic tree structure. In that case it will be $\mathbf{O\Bigl(\|V\| \|E\| \log\bigl(\|V\|\bigr)\Bigr)}$.

\end{enumerate}

\subsection{Proof of Menger's theorem}\label{subsc:menger}

\begin{lm}[Integer flow]
    If all capacities of a flow network are integer then there exists a max-flow that is also integer on every edge.
\end{lm}
\begin{proof}
    By using any max-flow construction algorithm one can obtain the integer max-flow.
\end{proof}

\begin{lm}
    If all capacities of a flow network are either 0 or 1 and the max-flow value equal to $K$ then the paths from the source $s$ to the terminal $t$ are $K$-edge independent.  
\end{lm}
\begin{proof}
    Let's construct the integer max-flow from the previous lemma. By using decomposition algorithm any paths $P_i$ will be edge independent and since the flow among each path equals to 1 then the number of paths $P_i$ equals to $K$.
\end{proof}

\begin{thm}[Menger]\label{thm:menger}
    A graph is $\k$-vertex (edge) connected iff. for any two vertices there exist $\k$-vertex (edge) independent paths.
\end{thm}
\begin{proof}
    \begin{enumerate}
    \item Let's prove first for edges.
        \begin{enumerate}
        \item \emph{Sufficiency.} If there exist $\k$-edge independent paths thus there is no $\k-1$-edge cuts.
        \item \emph{Necessarity.} Let the graph $G$ be connected after a deletion of any $\k-1$ edges. Let's denote two vertices in the condition by $s$ and $t$. Let's define weight of every edge in $G$ by 1. Consider $G$ as the flow network and let $f$ be the max-flow. By the Ford-Fulkerson theorem~\ref{thm:fordfulk}: $F = c(A_{min}, B_{min})$ and equal to the number of edges from vertices of $A_{min}$ to vertices of $B_{min}$. Since the graph $G$ will be connected after a deletion of any $\k-1$ edges and $F$ is integer, the value $F \geq \k$. Thus by the previous lemma the paths from $s$ to $t$ are $\k$-edge independent.
        \end{enumerate}
    \item Let's prove the vertex condition.
        \begin{enumerate}
        \item \emph{Sufficiency.} The same as for edges.
        \item For \emph{necessarity} condition let's give additional definition:

        \begin{defin}
            Let $A$ and $B$ be two subsets of vertices of the graph $G$. A set $S\sub V(G)$ is called \bf{$AB$-separator} if there is no path in $G$ starting from any vertex of $A$ and ending in any vertex of $B$ after the deletion of the set $S$.
        \end{defin}

        Let's prove that for any $A\sub V(G)$ and $B\sub V(G)$ if any $AB$-separator consist of at least $\k$ vertices then there exist $\k$-vertex independent paths from vertices of $A$ to vertices $B$ (let paths can also be consisted of just one vertex). Let's prove it by the induction of the number of edges in $G$: 
        \begin{enumerate}
            \item The base. If $G$ has no edges then the vertices in $A\bigcap B$ are also paths.
            \item Let the induction statement holds for the graph $G' = G\setminus(v_1, v_2)$. Let's prove it for $G$. Let's assume that $AB$-separator in $G$ consists of $\k$ vertices. If $AB$-separator in $G'$ consists also of $\k$ vertices then by the induction statement there exist $\k$-vertex independent paths. Thus there exists $AB$-separator $S$ in $G'$ that consists of $\k-1$ vertices. Let $v_1$ (without loss of generality) is reachable from some vertex of $A$ in $G'\setminus S$ (otherwise $S$ is a $AB$-separator in $G$). \\
            
            Let's denote this set $S\bigcup \{v_1\}$ by $S_1$. The set $S_1$ is $AB$-separator in $G$ (otherwise $S$ is not a $AB$-separator in $G'$). The $AS_1$-separator is also $AB$-separator in $G$ and thus consists of at least $\k$ vertices. Thus by the induction statement there exist $\k$-vertex independent paths from vertices of $A$ to vertices $S_1$. By providing the same proof for $S_2 = S\bigcup \{v_2\}$ one can obtain $\k$-vertex independent paths in $G$ from vertices of $S_2$ to $B$. Since $\|S_1\| = \|S_2\| = \k$, there exist $\k$-vertex independent paths from $A$ to $B$ in $G$.
        \end{enumerate}
        \end{enumerate}
    \end{enumerate}
\end{proof}

\subsection{Minimum-cost flow}

\begin{defin}
    A \bf{cost} function $a$ is any function $a:V(G)\times V(G)\rightarrow \rR$ on a flow network $G$. \bf{The total cost} of the flow $f$ with cost function $a$ is 
    $$p(G) = \sumt_{(u, v)\in V(G)\times V(G),\ f(u,v) > 0} a(u,v) f(u,v).$$
\end{defin}

\begin{defin}
    A flow is called \bf{min-cost-flow} if $f$ has the minimum total cost $p$ among the all another flows with the flow value $F$.
\end{defin}

Let $G$ be a flow network with a flow $f$ and cost function $a$.

\begin{lm}[Necessary and sufficient condition of min-cost-flow]\label{lm:nsmincost}
    A flow $f$ is min-cost-flow iff. there are no negative cycles corresponding to the cost function $a$ in the residual network $G_f$.
\end{lm}
\begin{proof}
    \begin{enumerate}
        \item \emph{Necessarity.} Assume the contrary: there exists negative cycle $C$ in the residual network $G_f$. Let $c_{f,C} = \min_{e\in C} c_f(e)$. Let's define a flow 
        $$f_C(u,v) = 
        \begin{cases}
            c_{f,C} & (u,v)\in C,\\
            -c_{f,C} & (v,u)\in C,\\
            0 & otherwise.
        \end{cases}
        $$
        The total cost of $f+f_C$ will be less than the total cost of $f$. This holds contradiction.
        \item \emph{Sufficiency.} Let $f$ be the flow with no negative cycles corresponding to the cost function $a$ in the residual network $G_f$ and $f'$ min-cost-flow among all flows with the flow value $F$. The total cost $p(f')\leq p(f)$. By the subtraction of flows lemma~\ref{lm:flowminus}: $f' = f+\sumt_i f_{C_i} \imp p(f')\geq p(f) \imp p(f') = p(f)$.  
    \end{enumerate}
\end{proof}

\begin{thm}\label{thm:deltaflow}
    Let $f$ be the min-cost-flow, $P$ be the path with minimum sum of costs among all paths from the source $s$ to the terminal $t$ in the residual network $G_f$ and $c_{f,P}$ be the minimum residual capacity among the path $P$. Let
    $$ f_\delta(u,v) = 
    \begin{cases}
        \delta & (u,v)\in P\\
        -\delta & (v,u)\in P \ \ ,\quad \text{ for any }\delta: 0\leq\delta\leq c_{f,P}.\\
        0 & otherwise
    \end{cases}
    $$
    Then $f+f_\delta$ is min-cost-flow among all flows with the value $F+\delta$. 
\end{thm}
\begin{proof}
    Let $g$ be a min-cost-flow among all flows with value $F+\delta$. Then by the two flows lemma~\ref{lm:2fl}: $g = f+f'$, where $f'$ is a flow in the residual network $G_f$ with a flow value $\delta$. By the decomposition theorem~\ref{thm:decflow} $f' = \sumt_i f'_{P_i}+ \sumt_i f'_{C_i}$. 
    
    By previous lemma any cycle $C_i$ is non-negative cycle corresponding to the cost function $a$. If $C_i$ is a positive cycle then the total cost of $g$ can be decreased and it holds contradiction. Thus $\forall i\text{ and } e\in C_i \text{ hold } a(e) = 0$. Therefore, 

    $$p(f') = \sumt_i p(f'_{P_i}) = \sumt_i \sumt_{(u, v)\in P_i} a(u,v) f'_{P_i}(u,v) = \sumt_i f'(P_i) \sumt_{(u, v)\in P_i} a(u,v) \geq $$
    $$\geq\sumt_i f'(P_i) \sumt_{(u, v)\in P} a(u,v) = \delta \sumt_{(u, v)\in P} a(u,v) = p(f_\delta).$$
\end{proof}

Before introducing min-cost-max-flow construction algorithm let's define the flow among a path $P$ in the residual network by 
    $$f' = \begin{cases}
        c_{f,P} & (u,v)\in P,\\
        -c_{f,P} & (v,u)\in P,\\
        0 & otherwise,
        \end{cases}$$ 
        where $c_{f,P}$ is the minimum residual capacity among the path $P$.

Now, let's introduce

\begin{enumerate}
    \item \bf{Min-cost-max-flow construction using Ford-Fulkerson algorithm.}\\
    \emph{Description:}
    \begin{enumerate}
        \item Initialise $f = 0$.
        \item\label{enum:mincost1} Find a negative cycle $C$ corresponding to the cost function $a$ in the residual network $G_f$ using Belman-Ford algorithm.
        \item Increase $f$ by the flow $f'$ among the cycle $C$ and go to step~\ref{enum:mincost1}.
        \item\label{enum:mincost2} When there will be no negative cycles in the resulting residual network $G_f$, find the path $P$ with minimum sum of costs from the source $s$ to the terminal $t$ in this network $G_f$.
        \item Increase $f$ by the flow $f'$ among the path $P$ and go to step~\ref{enum:mincost2}.
    \end{enumerate}
    \NB Since this algorithm is a modification of Ford-Fulkerson algorithm, the complexity is also depends of the distribution of capacities and price function.
    \item \bf{Min-cost-max-flow construction using Dinic algorithm.}\\
    \emph{Description:}
    \begin{enumerate}
        \item Find max-flow $f$ using Dinic algorithm.
        \item\label{enum:mincost3} Find a negative cycle $C$ corresponding to the cost function $a$ in the residual network $G_f$ using Belman-Ford algorithm.
        \item Increase $f$ by the flow $f'$ among the cycle $C$ and go to step~\ref{enum:mincost3}.
    \end{enumerate}
    \NB The complexity of this algorithm is $\mathbf{O\bigl(\|V\|\|E\|^2+\|V\|^2\|E\|\bigr) = O\Bigl(\|V\|\|E\|\bigl(\|E\|+\|V\|\bigr)\Bigr)}.$
\end{enumerate}


\section{Local and global characteristics of graph}

\subsection{Centralities with local knowledge}

Let's consider an undirected graph $G$. 

\begin{defin}
    \bf{Centrality} is a function defined for each vertex of a graph that contains some information of a graph structure.
\end{defin}

Let's denote
\begin{itemize}
    \item[$\bullet$] by $\cN(v)$ the set of vertices which adjacent to a vertex $v$,
    \item[$\bullet$] by $\overline{\cN}(v) = \cN(v)\bigcup {v}$,
    \item[$\bullet$] by $\overline{f}(x_1, x_2, ... , x_k)$, where $f$ is any function $V\times V \times ... \times V\rightarrow\rR$ the restriction of this function on $\overline{\cN}(v)$ (for example $\overline L(x,y)$ will be the average shortest path between $x$ and $y$ restricted to subgraph $\overline \cN(v)$),
     \item[$\bullet$] by $d_i = deg(v_i)$,
     \item[$\bullet$] by $n = \|V(G)\|,\;\; m = \|E(G)\|$,
     \item[$\bullet$] by $X(i) = X(v_i)$ for any $X$ --- set or function corresponding to vertex $v_i$.
\end{itemize}
    
First let's give some general centralities which contain local information of the graph $G$: 

\begin{enumerate}
    \item The simplest example is the \bf{degree centrality $d_i = \deg{v_i} = \|\cN(i)\|$}.
    \item Let's denote by $G\bigl(\cN(v)\bigr)$ the induced subgraph on vertices $V\bigl(\cN(v)\bigr)$. and by $MC(v)$ the largest connected component in $G\bigl(\cN(v)\bigr)$. \bf{Maximal neighborhood component $MNC(v)$} is the number of vertices in $MC(v)$.
    \item \bf{Density of maximal neighborhood component} $DMNC(v) = \frac {\bigl\|E\bigl(MC(v)\bigr)\bigr\|} {\bigl\|V\bigl(MC(v)\bigr)\bigr\|^\eps}$, for some $\eps\in [1,2]$.
    \item \bf{Local cluster coefficient} $c_i = c(v_i) = \frac {\text{number of edges in }\cN(i)} {\text{maximum possible number of edges in }\cN(i)}= \frac{2 \bigl\|E\bigl(\cN(i)\bigr)\bigr\|} {d_i(d_i-1)}$. 
\end{enumerate}

\noindent Let's represent the list of general global characteristics of a graph:

\begin{enumerate}
    \item The simplest example is \bf{diameter} $diam(G) = \max_{s,t} \dist(s,t)$.
    \item \bf{Density} $D(G) = \frac {\text{number of edges in }G} {\text{maximum possible number of edges in }G}= \frac{2 m} {n(n-1)}$.
    \item \bf{Global efficiency} $E_{glob}(G) = \frac 1 {n (n-1)} \sumt_{s\neq t} \frac 1 {\dist(s,t)}$.
    \item \bf{Average shortest path length} $L(G) = \frac 1 {n(n-1)} \sumt_{s\neq t} \dist(s,t)$.
    \item \bf{Average clustering coefficient} $C_{WS}(G) = \frac 1 {n} \sumt_{i\in V(G)} c_i = \frac 1 {n} \sumt_{i\in V(G)}  \frac{2 \bigl\|E\bigl(\cN(i))\bigr)\bigr\|} {d_i(d_i-1)}$.
     \item \bf{Global clustering coefficient} 
     
    $C(G) = \frac {\text{number of closed triplets in $G$}} {\text{number of all triplets in $G$}} = \frac {3\times\text{number of triangles $G$}} {\frac 1 2 \sumt_{i\in V(G)} d_i (d_i-1)}  = \frac {\sumt_{i, j, k\in V(G)} a_{ij} a_{jk} a_{ki}} {\sumt_{i\in V(G)} d_i (d_i-1)}$.
    \item \bf{Small world coefficient} $SW(G) = \frac {C_{WS}(G)} {C_{WS}(G_{r})}$, where $G_{r}$ is a random graph  $(n, m)$.
\end{enumerate}

\NB It turns out that most real life networks are satisfied the small world property $SW(G)\gg 1$. \\


\NB Random graph $(n, m)$ is also called \emph{Erdős–Rényi graph}.\\

Often in the literature there is common mistake: the average clustering coefficient or Watts-Strogatz coefficient $C_{WS}(G)$ is called global clustering coefficient, but they are different in most cases. Let's show some examples:

\begin{ex}
    Consider \bf{windmill graph $W(n, k)$} --- the graph that is constructed of $n$ copies of complete graph $K_n$ and one additional vertex that has connections to all vertices of these graphs. Following by~\cite{Estrada} for windmill graphs $W(n, k)$ hold:
    \begin{enumerate}
        \item $\forall n,k\geq 2,\;\;C_{WS}\bigl(W(n, k)\bigr) > C\bigl(W(n, k)\bigr),$
        \item $\underset{n\rightarrow\infty}{\lim} C_{WS}\bigl(W(n, k)\bigr) = 1,\;\underset{n\rightarrow\infty}{\lim} C\bigl(W(n, k)\bigr) = 0.$
    \end{enumerate}\par \ \par
    Let's prove this. First,\par\ \par
    $C\bigl(W(n, k)\bigr) = \frac {3 n\ \bigl(\frac {k(k-1)} 2+ \frac {k(k-1)(k-2)} 6\bigr)} {\frac 1 2 \bigl(n k^2(k-1)+n k (n k -1)\bigr)} =  \frac {k^2-1} {k^2-k+nk-1}.$\par\ \par
    $c_{\text{central vertex}} =\frac {n{k(k-1)} } {n k (n k -1)} = \frac {k-1} {n k-1},$\par
    $C_{WS}\bigl(W(n, k)\bigr) = \frac 1 {nk+1} (\frac {k-1} {n k-1}+nk) = $\par
    $\qquad\qquad\qquad\qquad\ \, =  \frac {k-1+n^2 k^2 - n k} {n^2 k^2 -1} = 1-k \frac {n-1} {n^2 k^2-1}.$\par\ \par
    Hence, $\underset{n\rightarrow\infty}{\lim} C_{WS}\bigl(W(n, k)\bigr) = 1,\;\underset{n\rightarrow\infty}{\lim} C\bigl(W(n, k)\bigr) = 0.$\par\ \par\ \par\ \par\ \par
    
\begin{figure}[h]
\vspace{-220pt}
\hspace{280pt}
\begin{minipage}{8cm}
\hspace{50pt}
	\includegraphics[width=0.5\textwidth]{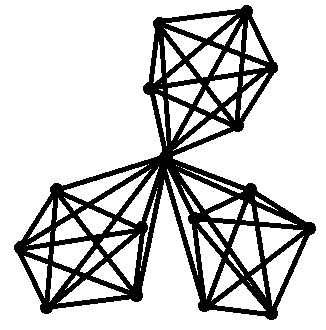}
	\caption{Windmill graph $W(3, 5)$.}
	\label{fig:sample}
 \end{minipage}
\end{figure}

    Second, compare the average clustering coefficient and the global clustering coefficient:
    $$\vspace{5pt}C_{WS}\bigl(W(n,k)\bigr)\vee C\bigl(W(n,k)\bigr)$$
    $$\vspace{5pt}\frac {k-1+n^2 k^2 - n k} {n^2 k^2 -1}\vee \frac {k^2-1} {k^2-k+nk-1}$$
    $$k^3 n^3- k^3 n^2 - k^3 n - k^2 n^2 + 2 k^2 n + k^3 -k^2\vee 0$$
    $$\vspace{5pt}k^2 (n-1)^2 \bigl(k (n-1)-1\bigr) \vee 0$$
    For $n\geq 2$ and $k\geq 2: k^2 (n-1)^2 \bigl(k (n-1)-1\bigr) > 0$. Therefore, $C_{WS}\bigl(W(n, k)\bigr) > C\bigl(W(n, k)\bigr).$ 
\end{ex}

\begin{ex}
    Let's consider \bf{wheel graph $W(k)$} --- the graph constructed of $k$-vertex cycle and one vertex in the center that has connections to all vertices of cycle. If $k = 3$ it is easy to see that $C_{WS}\bigl(W(3)\bigr) = C\bigl(W(3)\bigr) = 1$. If $k = 4$ then $C_{WS}\bigl(W(4)\bigr) = C\bigl(W(4)\bigr) = \frac 2 3$.
    
    For the case $k\geq 5$ holds
    \begin{enumerate}
        \item $C_{WS}\bigl(W(k)\bigr) > C\bigl(W(k)\bigr),$
        \item $\underset{k\rightarrow\infty}{\lim} C_{WS}\bigl(W(k)\bigr) = \frac 2 3,\;\underset{k\rightarrow\infty}{\lim} C\bigl(W(k)\bigr) = 0.$
    \end{enumerate}\par \ \par
    Let's prove it. First, let's calculate for this graph average clustering coefficient and global clustering coefficient: \par\ \par
    
    $C\bigl(W(k)\bigr) = \frac {3 k} {\frac 1 2 \bigl(6 k+k(k-1)\bigr)} = \frac 6 {k+5},$\par
    $C_{WS}\bigl(W(k)\bigr) = \frac 1 {k+1} \bigl(\frac {2n} {k(k-1)}+\frac 2 3 k\bigr) = \frac {2 (k^2-k+3)} {3 (k^2-1)}.$\par\ \par
    Hence, $\underset{k\rightarrow\infty}{\lim} C_{WS}\bigl(W(k)\bigr) = \frac 2 3,\;\underset{k\rightarrow\infty}{\lim} C\bigl(W(k)\bigr) = 0.$\par\ \par 
    Second, compare these two coefficients:\par\ \par
    $\vspace{5pt}C_{WS}\bigl(W(k)\bigr)\vee C\bigl(W(k)\bigr),$\par
    
\begin{figure}[h]
\vspace{-145pt}
\hspace{280pt}
\begin{minipage}{8cm}
\hspace{50pt}
    \center
	\includegraphics[width=0.4\textwidth]{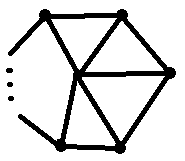}
	\caption{Wheel graph $W(k)$.}
	\label{fig:sample}
 \end{minipage}
\end{figure}
    $\vspace{5pt}\frac {2 (k^2-k+3)} {3 (k^2-1)}\vee\frac 6 {k+5},$\par 
    $k^3-5 k^2-2 k+24\vee 0$\par
    $(k^2-k+3)(k+5)\vee 9(k^2-1)$ \par
    $\vspace{5pt}(k+2)(k-3)(k-4) > 0$\par
    Therefore, $C_{WS}\bigl(W(k)\bigr) > C\bigl(W(k)\bigr).$ 
\end{ex}

\begin{ex}\vspace{15pt}
    Let's consider \bf{nested triangles graph $T(n)$} with $n+2$ triangles --- the graph constructed of 2 outer triangles and $n$ subsequently nested triangles with connections between correspondent vertices (see figure~\ref{fig:nestr}). For these graphs:
    \begin{enumerate}
        \item $C_{WS}\bigl(T(n)\bigr) > C\bigl(T(n)\bigr),$
        \item $\underset{n\rightarrow\infty}{\lim} C_{WS}\bigl(T(n)\bigr) = \frac 1 6,\;\underset{n\rightarrow\infty}{\lim} C\bigl(T(n)\bigr) = 0.$
    \end{enumerate}\par \ \par
    Let's prove this.\par\ \par
    $C\bigl(T(n)\bigr) = \frac {6} {\frac 1 2 (36+36 n)} = \frac 1 {3(n+1)},$\par
    $C_{WS}\bigl(T(n)\bigr) = \frac 1 {3(n+2)} \Bigl(6\,\frac 1 3+3 n\,\frac 1 6\Bigr) = \frac {n+4} {6(n+2)}.$\par\ \par
    Hence, $\underset{n\rightarrow\infty}{\lim} C_{WS}\bigl(T(n)\bigr) = \frac 1 6,\;\underset{n\rightarrow\infty}{\lim} C\bigl(T(n)\bigr) = 0.$
    
\begin{figure}[h]\label{fig:nestr}
\vspace{-200pt}
\hspace{280pt}
\begin{minipage}{8cm}
\hspace{50pt}
    \center
	\includegraphics[width=0.7\textwidth]{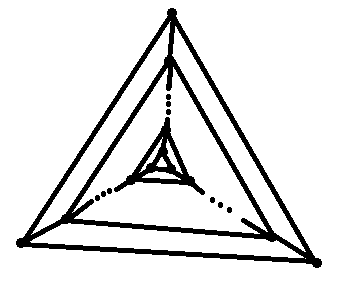}
	\caption{Nested triangles graph $T(n)$.}
	\label{fig:sample}
 \end{minipage}
\end{figure}\ \par

    Compare these two coefficients:
    $$\vspace{5pt}C_{WS}\bigl(T(n)\bigr)\vee C\bigl(T(n)\bigr)$$
    $$n^2+5n+4\vee 2n+4$$
    $$\vspace{5pt} n(n+3) > 0$$
    
    Hence, $C_{WS}\bigl(T(n)\bigr) > C\bigl(T(n)\bigr).$
\end{ex}

Let's prove a theorem (following by~\cite{Tuzh}) about relations between average local and global clustering coefficient.

\begin{thm}\label{thm6} Consider simple graph $G$. Let's $\forall i,j\in V(G),\; i\leq j$ hold $d_i\leq d_j \Rightarrow c_i\leq c_j$, then 
$$C_{WS}(G)\leq C(G).$$
\end{thm}
\begin{proof}
Consider re-numeration of vertices such that $\forall i\leq j:d_i\leq d_j$.
One can rewrite
$$
c_i = \frac {\sumt_{j, k\in V(G)} a_{ij} a_{jk} a_{ki}} {d_i (d_i-1)}.
$$
Indeed,
$$ 
a_{ij} a_{jk} a_{ki} = 
\begin{cases}
	1 & \text{if there exists edge between vertices $j$ and $k$ which adjacent to vertex $i$}, \\
	0 & \text{otherwise}. 
\end{cases}
$$
Therefore, 
$$
C_{WS}(G) = \frac 1 n \sumt_{i\in V(G)} \frac {\sumt_{j, k\in V(G)} a_{ij} a_{jk} a_{ki}} {d_i (d_i-1)}.
$$

Denote by $x_i = d_i (d_i-1)$. The number of edges $\|E\bigl(\cN(i)\bigr)\| = \frac 1 2 \sumt_{j, k\in V(G)} a_{ij} a_{jk} a_{ki}$ and the maximum number of edges in subgraph $\cN(i)$ equals to $\frac {d_i(d_i-1)} 2$, thus $x_i\geq 2,\; 0\leq c_i\leq 1.$ 

\vspace{10pt} Using Chebyshev's sum inequality ($d_i\leq d_j \Rightarrow x_i\leq x_j \text{ and } c_i\leq c_j$):
$$
\frac 1 n \sumt_{i \in  V(G)} x_i\;  C_{WS}(G) = \Bigg(\frac 1 n \sumt_{i \in  V(G)} x_i\Bigg) \Bigg(\frac 1 n \sumt_{i \in  V(G)} c_i\Bigg) \leq \frac 1 n  \sumt_{i \in  V(G)} x_i c_i = \frac 1 n \sumt_{i, j, k\in V(G)} a_{ij} a_{jk} a_{ki}.
$$
Therefore,
$$
C_{WS}(G) \leq \frac {\sumt_{i, j, k\in V(G)} a_{ij} a_{jk} a_{ki}} {\sumt_{i \in  V(G)} d_i(d_i-1)}= C(G).
$$
The equality holds when $\forall i,j\in V(G): d_i = d_j$ (i.e. for regular simple graphs) or when $\forall i,j\in V(G): c_i = c_j$.
\end{proof}

\begin{cor}
Let's $\forall i,j\in V(G),\; i\leq j$ hold $d_i\leq d_j \Rightarrow c_i\geq c_j$, then 
$$C_{WS}(G)\geq C(G).$$
\end{cor}
The proof is the same as in theorem~\ref{thm6} using Chebyshev's sum inequality.

\begin{cor}
For simple regular graphs $G$:
$$C_{WS}(G) = C(G).$$
\end{cor}

We see that for many cases $C_{WS}(G)\geq C(G)$, but it is not very hard to come up with an example then $C_{WS}(G) < C(G).$ Consider the complete graph $K_n$ and glue to each vertex cycle of the length 4. Due to the symmetry by the theorem~\ref{thm6}: $C_{WS}(G) < C(G).$

\subsection{Centralities with global knowledge}

Now let's provide the list of general centralities contained a global information of an undirected graph $G$: 

\begin{enumerate}
    \item \bf{Betweenness centrality} $BC(v) = \sumt_{s\neq t\neq v} \frac {\sigma_{st}(v)} {\sigma_{st}}$, where $\sigma_{st}$ is the total number of shortest paths from $s$ to $t$ and $\sigma_{st}(v)$ is the total number of shortest paths which contains vertex $v$.
    \item Let $T_s$ be a shortest path tree with the root $s$. \bf{Bottleneck} $BN(v) = \sumt_{s\in V(G)} p_s(v)$, where 
    $$p_s(v) = 
    \begin{cases}
        1 & \text{if the set of descendants of $v$ in $T_s$ contains more than $\|V(T_s)\|/4$ vertices,}\\
        0 & \text{otherwise.}
    \end{cases}$$
    \item \bf{Closeness centrality} $Clo(v) = \frac 1 {\sumt_{t\in V(G)} \dist(v, t)}$. 
    \item \bf{Eccentricity centrality} $EC(v) = \frac 1 {\max_{t\in V(G)} \dist(v, t)}$.
    \item \bf{Local efficiency} $E_{loc}(G) = \frac 1 {n} \sumt_{v\in V(G)} E_{glob}\bigl(\cN(v)\bigr).$
    \item \bf{Radiality} $Rad(v) = \frac {\sumt_{t\in V(G)} \bigl(diam(G)+1-\dist(v, t)\bigr)} {n-1}$.
    \item \bf{Stress} $Str(v) = \sumt_{s\neq t\neq v} \sigma_{st}(v)$, where $\sigma_{st}$ is the total number of shortest paths from $s$ to $t$ which contains vertex $v$.
\end{enumerate}

Let's note that if $A(G)$ is the adjacency matrix of $G$ and $E(G)$ is the efficiency matrix, where 
$$E = \{e_{ij}\} = 
\begin{cases}
    \frac 1 {\dist(v_i, v_j)} & i\neq j,\\
    0 & \text{otherwise,}
\end{cases}
$$
then 
$$E_{glob}(\cN(v)) = \frac {\sumt_{(v_i, v_j)\in E(\cN(v))} e_{i,j}} {d_i (d_i-1)}.$$
Thus 
$$
A(G)-E(G) = 
\begin{cases}
    0 
    \frac 1 {\dist(v_i, v_j)} & v_i\not\sim v_j,\\
    0 & \text{otherwise.}
\end{cases}
$$

Now let's prove some relationships between these measures following by~\cite{Strang}:

\begin{thm}\label{thm:cent2}
    $L(G)\geq 2-D(G).$
\end{thm}
\begin{proof}
    Let $\a$ be the number of pairs of vertices $v, w$ with $\dist(v,w) = 2$ and $\b$ be the number of pairs of vertices $v, w$ with $\dist(v,w) \geq 3$ and let's denote the number of vertices of $G$ by $n$ and the number of edges by $m$. Thus $2m+\a+\b = n(n-1)$.
    $$D(G) = \frac {2m} {n(n-1)},$$
    $$L(G) = \frac {2 m+2\a+t\b} {n(n-1)}, \text{ for some }t \geq 3,$$
    $$\frac {2 m+2\a+t\b} {n(n-1)}\geq \frac {2 m+2(\a+\b)} {n(n-1)} = \frac {4 m+2(\a+\b)} {n(n-1)}- \frac {2m} {n(n-1)} = 2-D(G).$$
\end{proof}

\begin{thm}\label{thm:cent3}
    $ 3-L(G)\leq 2 E_{glob}(G)\leq 1+D(G).$
\end{thm}
\begin{proof}
    Let's denotations be the same as in the previous theorem. 
    \begin{enumerate}
        \item Consider the case: $2 E_{glob}(G)\leq 1+D(G).$\\
            $$2 E_{glob}(G) = \frac 1 {n(n-1)} (4m+ \a+2\eps \b), \text{ for some }\eps < \frac 1 2,$$
            $$\hspace{-10pt}\frac 1 {n(n-1)} (4m+\a+2\eps \b)\leq \frac 1 {n(n-1)} (4m+\a+\b) = \frac {2m} {n(n-1)} + \frac {2m+\a+\b} {n(n-1)} = 1+D(G).$$
        \item Consider the case: $2 E_{glob}(G)\geq 3-L(G).$\\
        $$\frac 1 {n(n-1)} (4m+\a+2\eps \b)\vee 3-\frac {2 m+2\a+t\b} {n(n-1)},$$
        $$4m+\a+2\eps\b\vee 3 (2m+\a+\b)-2m-2\a-t\b,$$
        $$\b(2\eps+t-3)\geq 0.$$
    \end{enumerate}
\end{proof}

\NB The equality in the previous theorems~\ref{thm:cent2}-\ref{thm:cent3} is obtained in the case then $diam\bigl(\cN(v)\bigr)\leq 2$ for any vertex $v$.

\begin{lm}
    $L\bigl(\cN(i)\bigr) = 2-c_i$.
\end{lm}
\begin{proof}
    $$L\bigl(\cN(i)\bigr) = \frac 1 {d_i(d_i-1)} \sumt_{s,t\in \cN(i), s\neq t} \dist(s,t) = $$
    $$ = \frac 1 {d_i(d_i-1)} \Bigg(\sumt_{(s,t)\in E(\cN(i))} \dist(s,t)+ \sumt_{s,t\in \cN(i), (s,t)\notin E(\cN(i))} \dist(s,t)\Bigg) = $$
    $$ = \frac 1 {d_i(d_i-1)} \Bigg(2 \|E\bigl(\cN(i)\bigr)\| + \sumt_{(s,i), (i, t)\in E(G), (s,t) \notin E(G)} \dist(s,t) \Bigg) = $$ 
    $$ = \frac 1 {d_i(d_i-1)} \Bigg(2 \|E\bigl(\cN(i)\bigr)\|+2 \Bigl(d_i(d_i-1) - 2 \| E\bigl(\cN(i)\bigr)\|\Bigr)\Bigg) = 2 - c_i.$$
    Note that shortest paths are calculated for $L$ in the whole graph $G$.
\end{proof}

\begin{cor}\label{thm:cent1}
    $E_{loc}(G)= \frac 1 2 (1+C_{WS}(G))$.
\end{cor}
\begin{proof}
Let's give two alternative proofs:
\begin{enumerate}
    \item Let's note that $D(\cN(i)) = C(i)$. By theorem~\ref{thm:cent3}: 
    $$3-L(\cN(i))\leq 2 E_{glob}(\cN(i))\leq 1+D(\cN(i)),$$
    $$3-(2-c_i)\leq 2 E_{glob}(\cN(i))\leq 1+c_i,$$
    note shortest paths are calculated for $L$ and $E_{glob}$ in the whole graph $G$. Averaging by $i$ we obtain equality.
    \item Let's rewrite local cluster coefficient formula:
$$c_i = \frac {\sumt_{(s, t)\in E(\cN(i))} 1} {d_i(d_i-1)},$$

$$\frac 1 2 (1+c_i) = \frac 1 2  \frac {\sumt_{(s, t)\in E(\cN(i))} 1 + \sumt_{(s, t)\in E(\cN(i))} 1+ \sumt_{s,t\in V(\cN(i)), (s, t)\notin E(\cN(i))} 1} {d_i(d_i-1)} = $$
$$ = \frac {\sumt_{(s, t)\in E(\cN(i))} 1 + \sumt_{s,t\in V(\cN(i)),(s, t)\notin E(\cN(i))} \frac 1 2} {d_i(d_i-1)} = \frac {\sumt_{s, t\in V(\cN(i))} \frac 1 {\dist(s,t)} } {d_i(d_i-1)} = E_{glob}\bigl(\cN(i)\bigr).$$
Averaging by $i$ we obtain equality.
\end{enumerate}
\end{proof}

\section{Spectral properties}
In this section theorems and lemmas are provided for undirected graphs unless otherwise specified. Let's first review some important and well-known results from Linear Algebra.

\begin{thm}
[Courant-Fischer-Weyl Theorem]
\label{thm:min-max}
Let $A$ be a symmetric matrix of order $n$, and denote by $\l_1(A)\leq\l_2(A)\leq ... \leq \l_n(A)$ the eigenvalues of $A$. Let $\mathcal{H}^k$ be the collection of all $k$-dimensional linear subspaces of
$\rR^n$. Then 
\begin{align*}
\lambda_k(A)&= \min_{X \in { \mathcal{H}^{k}}} \max_{x\in X\backslash \{0\}}   
\frac{(A x,x)}{(x,x)} 
=  \max_{X' \in \mathcal{H}^{n-k+1}} \min_{ x \in X'\backslash \{0\}}  
\frac{(A x,x)}{(x,x)} 
\\&.
\end{align*}
\end{thm}

\begin{thm}[Cauchy’s Interlacing Theorem \cite{HJ12,Fisk05}]
\label{thm:Cauchy-Interlacing}
Let $A$ be an $n\times n$ symmetric real matrix, and $B$ is a principle $m\times m$ submatrix of $A$, 
 then the eigenvalues of $B$ interlace the eigenvalues of $A$, i.e.,
 $$\alpha_k\le \beta_k\le \beta_{k+n-m},\;k=1,\cdots,m,$$
 where $\alpha_1\le\cdots\le \alpha_n$ are the eigenvalues of $A$, and  $\beta_1\le\cdots\le \beta_m$ are the eigenvalues of $B$.
\end{thm}

Before presenting the following famous Perron-Frobenius Theorem, we give some definitions for matrices.

\begin{defin}
   A square $n\times n$ matrix $A=(a_{ij})$  is called \bf{reducible} if the indices $1, 2, \cdots, n$ can be divided into two disjoint nonempty sets $i_1,\cdots,i_k$ and $j_1,\cdots,j_\ell$ (with $k+\ell=n$) such that $a_{ij}=0$ for any $i\in \{i_1,\cdots,i_k\}$ and $j\in \{j_1,\cdots,j_\ell\}$. A square matrix which is not reducible is said to be \bf{irreducible}.  
\end{defin}

\begin{st}\label{st:connected-Perron}
The adjacency matrix of a connected graph is  symmetric, \textbf{irreducible}, and have nonnegative entries.
\end{st}
\begin{proof}
By the definition of adjacency matrix, it is clear to be symmetric and have nonnegative entries. 

Now we need to check that if the graph is connected, then its adjacency matrix must be irreducible. Suppose the contrary, that $G$ is connected but $A(G)$ is reducible.  Then the indices $1, 2, \cdots, n$ can be divided into two disjoint nonempty sets $i_1,\cdots,i_k$ and $j_1,\cdots,j_\ell$ (with $k+\ell=n$) such that $a_{ij}=0$ for any $i\in \{i_1,\cdots,i_k\}$ and $j\in \{j_1,\cdots,j_\ell\}$. 
This means that there is no edge with one end-point in  $i_1,\cdots,i_k$ and the other end-point in $j_1,\cdots,j_\ell$. Thus, $G$ is not connected, which contradicts to the assumption that $G$ is connected.
\end{proof}

\begin{thm}
[Perron-Frobenius Theorem \cite{HJ12}] \label{thm:Perron-Frobenius}
Let $A\in \R^{n\times n}$
be symmetric, irreducible, and have nonnegative entries. Then $A$ has an eigenvalue $\lambda$ that
is strictly positive. Furthermore, it has multiplicity one and its corresponding
eigenvector $\vec v$ has strictly positive entries.   Moreover, the largest eigenvalue of $A$ is positive and has multiplicity 1.
\end{thm}

  The proof is complicate and thus we only give a citation \cite{HJ12} for it.

\subsection{Adjacency and Laplacian matrix spectrum}

Let's prove three theorems about the properties of adjacency and Laplacian matrix spectrum for undirected graphs (following by~\cite{BellaB}). Let's denote by $\mu_{min}$ and $\mu_{max}$ the minimum and the maximum eigenvalues of adjacency matrix $A(G)$ respectively. First, let's notice that since the graph is undirected, the adjacency matrix is symmetric and thus Hermitian. Therefore, from the basic knowledge from Linear algebra the set $\{(Ax,x): \|x\| = 1\}$ is a closed interval $[a,b]\sub \R$, where $(\cdot, \cdot)$ is the standard scalar product and $a = \mu_{min},\; b = \mu_{max}$. 

If the graph is not empty then let's $(v_1, v_2)\in E(G)$ without loss of generality. Hence, $(Ax,x) > 0$ for $x = \bigl(\frac 1 {\sqrt{2}},\frac 1 {\sqrt{2}},0, ..., 0\bigr)$ and $(Ax,x) < 0$ for $x = \bigl(\frac 1 {\sqrt{2}},-\frac 1 {\sqrt{2}},0, ..., 0\bigr)$ and therefore, $\mu_{min}<0<\mu_{max}$.

Let's denote by $d_{min}$ and $d_{max}$ the minimum and maximum degree value of the graph $G$ respectively and let $\|V(G)\| = n$.

\begin{thm} Consider a connected undirected graph $G$ with adjacency matrix $A(G)$.
\begin{enumerate}
    \item For every eigenvalue $\mu$ of $A(G)$ holds $|\mu| \leq d_{max}$.
    \item The maximum degree $d_{max}$ is an eigenvalue of $A(G)$ iff. $G$ is regular graph. If $d_{max}$ is an eigenvalue then its multiplicity equals one. 
    \item If $-d_{max}$ is an eigenvalue of $A(G)$ then $G$ is regular and bipartite graph.
    \item If $G$ is bipartite then for every eigenvalue $\mu$ of $A(G): -\mu$ is also an eigenvalue with equal multiplicity.
    \item For the maximum eigenvalue $d_{min} \leq \mu_{max} \leq d_{max}$.
    \item For any induced subgraph $H\sub G: \mu_{min}(G)\leq \mu_{min}(H) \leq \mu_{max}(H)\leq \mu_{max}(G)$.
\end{enumerate}
\end{thm}
\begin{proof}
    \begin{enumerate}
        \item\label{am1} Let's $x = (x_1, x_2, ... x_n)$ be the eigenvector corresponding to the eigenvalue $\mu$ and $x_p = \max_{i} |x_i|$. Let's normalize this vector $x$ such that $x_p = 1$. Thus,
        $$|\mu| = |\mu x_p| = \Bigl|\sumt_{i=1}^{n} a_{pi} x_i\Bigr| \leq \sumt_{i=1}^{n} a_{pi} |x_i| \leq d_p \leq d_{max}.$$
        \item\label{am2} Consider $x$ and $x_p$ from~\ref{am1}. For any $i =1...n, |x_i|\leq 1$. Therefore,
        $$\mu = \mu x_p = \sumt_{i=1}^{n} a_{pi} x_i = \sumt_{(p,i)\in E(G)} x_i\leq d_p \leq d_{max}.$$
        If $d_{max}$ is eigenvalue of $A(G)$ then $d_p = d_{max}$ and $x_q = 1$ for any adjacent vertex $q$ with $p$. Thus, using the same procedure for $x_q$ imply $d_q = d_{max}$. Since the graph is connected, it holds for all vertices. Therefore, $G$ is regular graph. The reverse holds from the fact that the vector $\j$ of all ones is an eigenvector of $A(G)$ by~\ref{AGdeg} property. The multiplicity of $d_{max}$ holds from the Perron-Frobenius theorem~\ref{thm:Perron-Frobenius}.
        
        
        \item Using~\ref{am1} and~\ref{am2} we obtain $\sum_{i\sim p} x_i\geq -d_p \geq -d_{max}$. Hence for each $i\in V(G): i\sim p$ holds $x_i = -1$ and $d_i = d_{max}$. Using the same arguments for neighbours of $x_i$ we obtain $\forall k\in V(G): k\sim i$ holds $x_k = 1$ and $d_k = d_{max}$ (due to normalization of $x$ such that $x_i = 1$). Therefore, the vertex set $V(G)$ is divided by two groups with $x_i = 1$ and $x_i = -1$ and there is no edges within each of group. Thus, $G$ is bipartite and regular.  
        \item Let the vertex set divided by two groups $V(G) = V_1\bigsqcup V_2$ and there is no edges within each of group. Consider a map $b:\rR^n\rightarrow\rR^n$ such that $\forall x = (x_i)\in \rR^n,$
        $$ b(x_i) = 
        \begin{cases}
            -x_i & i\in V_1, \\
            x_i & i\in V_2.
        \end{cases}
        $$
        
        Let's show that if $\mu$ and $x$ are eigenvalue and eigenvector respectively then for any $i\in V(G)$ holds $\bigl(A\,b(x)\bigr)_i = -\mu \bigl(b(x)\bigr)_i$. Indeed, let $i\in V_1$ without loss of generality, then
        $$
            \bigl(A\,b(x)\bigr)_i = \sumt_{j\in V_1} a_{ij} (-x_j) + \sumt_{j\in V_2} a_{ij} x_j = \sumt_{j\in V_2} a_{ij} x_j = -\sumt_{j\in V_1} a_{ij} (-x_j)-\sumt_{j\in V_2} a_{ij} (-x_j) = 
        $$
        $$
            = -\bigl(A (-x)\bigr)_i = -\mu(-x_i) = -\mu \bigl(b(x)\bigr)_i. 
        $$
        Therefore, $b$ maps different eigenvectors corresponding to eigenvalue $\mu$ to different eigenvectors corresponding to eigenvalue $-\mu$. Therefore, these eigenvalues have equal multiplicity.
        \item From~\ref{am1} holds $\mu_{max}\leq d_{max}$. Then,
        $$
            \Bigg(A \Bigl(\frac 1 {\sqrt{n}}\;\j\Bigr), \frac 1 {\sqrt{n}}\;\j\Bigg) = \frac 1 n \sumt_{i = 1}^n \sumt_{j = 1}^n a_{ij} = \frac 1 n \sumt_{i = 1}^n d_i \geq d_{min}.
        $$
        Thus, $\mu_{max} = \maxt_{x\in \rR^n: \|x\| = 1} (Ax,x) \geq d_{min}$.
        \item Let $A'$ be the adjacency matrix for an induced subgraph $H$ and $y\in\rR^{n-1}:\|y\| = 1,$
        $$ 
            (Ay,y) = \maxt_{x\in \rR^n: \|x\| = 1} (Ax,x) = \mu_{max}(H).
        $$
        Consider $x = (y, 0)\in \rR^n$. Then, $\|x\| = 1$ and $(Ax, x) = (A'y, y) = \mu_{max} (H)$. Hence, $\mu_{max} (H)\leq \mu_{max} (G)$. The other inequality is proved similarly. 
    \end{enumerate}
\end{proof}

\vspace{10pt} Let's denote by $\l_1, \l_2, ... , \l_n$ the eigenvalues of Laplacian matrix $L$ for undirected graph $G$. Since $L$ is Hermitian matrix the same as adjacency matrix, all eigenvalues will be real and we can order them: $\l_1\leq\l_2\leq ... \leq \l_n$ and $\l_1 = \mint_{x\in \rR^n: \|x\| = 1} (Lx,x),\; \l_n = \maxt_{x\in \rR^n: \|x\| = 1} (Lx,x)$. In fact, if $x^1\in\rR^n: \|x^1\| = 1$ is eigenvector with eigenvalue $\l_1$, then 
$$
    \l_2 =  \mint_{x\in \rR^n: \|x\| = 1, (x,x^1) = 0} (Lx,x) = \mint_{x\in \rR^n: (x,x^1) = 0} \frac {(Lx,x)} {(x,x)}. 
$$
One can continue these procedure for all others eigenvalues. Let's prove  Laplacian matrix spectrum properties theorem:

\begin{thm}[Laplacian matrix spectrum properties]\label{thm:lapl} Consider undirected graph $G$ with $n$ vertices.
    \begin{enumerate}
        \item $\forall x\in\rR^n: (Lx, x)\geq 0$.
        \item The first eigenvalue $\l_1 = 0$ and $\j$ is eigenvector with eigenvalue $\l_1$.
        \item For the complete graph $K_n:$ $\l_2 = \l_3 = ... = \l_n = n$.
        \item For not complete graph $\l_2\leq \k(G)$, where $\k(G)$ is the connectivity of the graph.
    \end{enumerate}
\end{thm}
\begin{proof}
    \begin{enumerate}
        \item\label{Lp1} One can rewrite this as
        $$
            (Lx, x) = \sumt_{i=1}^n \Bigl(d_i x_i^2-\sumt_{i\sim j} x_i x_j\Bigr) = \sumt_{(i,j)\in E(G)} (x_i-x_j)^2\geq 0.
        $$
        \item The first part holds from~\ref{Lp1} and the second from the property~\ref{prop:sum} of Laplacian matrix.
        \item Let's do the same procedure like in Cayley theorem~\ref{thm:cayley} for the characteristic polynomial of Laplacian matrix $L$ of the complete graph $K_n$. Then, $\det(L-\l I) = -\l (n-\l)^{n-1}$. Let's note that by definition $\k(K_n) = n-1$.
        \item\label{Lp4} Let $S$ be a $\k(G)$-vertex cut. Denote by $H_1$ and $H_2$ induced subgraphs such that $G\setminus S = H_1\bigsqcup H_2$ and there are no edges between $H_1$ and $H_2$. Let $a = \|V(H_1)\|\text{ and } b = \|V(H_2)\|$. Consider a vector $x = (x_i)\in\rR^n$ such that
        $$
            x_i = 
            \begin{cases}
                b & i\in H_1, \\
                0 & i\in S, \\
                -a & i\in H_2.
            \end{cases}
        $$
        Then $(x, j) = 0$ and $(x,x) = ab^2+ba^2$. Denote $y = L x$. For the vector $\j$ holds $L b \j = 0$, therefore $y = L (x-b\j)$. Consider $i\in H_1$. Note that 
        $$
            \bigl(L (x-b\j)\bigr)_i = 0-\sumt_{k\sim i} a_{ki} (x_i-b) = -\sumt_{k\in S, k\sim i} (-b)\leq \k b.
        $$
        Analogously if $i\in H_2$ then $\bigl(L (x-b\j)\bigr)_i\geq -\k a$. Therefore,
        $$
            \l_2 = \mint_{z\in \rR^n: (z,\j) = 0} \frac {(Lz,z)} {(z,z)} \leq \frac {(Lx,x)} {(x,x)} \leq \frac {\k b^2 a+\k a^2 b} {ab^2+ba^2} = \k.
        $$
    \end{enumerate}
\end{proof}

Let's denote the set of edges between subsets $U\sub V(G)$ and $V\setminus U$ by $\d U$ so-called the \emph{boundary} of $U$. We will give a few words about why it is boundary and so forth in the section~\ref{sc:cheeger}, but now let's prove another property of the Laplacian matrix second eigenvalue~$\l_2$:

\begin{thm}
    Let $G$ be undirected graph with $n$ vertices. Then, $\forall U\sub V(G)$ holds
    $$
        \|\d U\|\geq \l_2 \frac {\|U\| \|V\setminus U\|} n.
    $$
\end{thm}
\begin{proof}
    Let $k = \|U\|$. Let's define $x = (x_i)\in \rR^n$ such as
    $$
        x_i = \begin{cases}
            n-k & i\in U,\\
            k & i\in V\setminus U.
        \end{cases}
    $$
    Then, $(x, \j) = 0$ and $(x,x) = k(n-k)^2+(n-k) k^2 = nk (n-k)$. Using the equation in~\ref{Lp1} in the previous theorem find $(Lx, x) = \|\d U\|\, n^2$. Thus, using the same arguments as in~\ref{Lp4} in the previous theorem obtain the inequality.
\end{proof}

\subsection{The normalized Laplacian spectrum}\label{sc:nls}
Let's define the normalized Laplacian for an undirected unweighted graph $G$ with no loops.
\begin{defin}
The \textbf{normalized Laplacian matrix} $\Delta$ is defined as $D^{-\frac12}LD^{-\frac12}$, where $L$ is the Laplacian matrix of $G$, and $D=\mathrm{diag}(\deg v_1,\cdots,\deg v_n)$ is the diagonal matrix consisting of degrees of $G$. Then 
$$
\Delta_{ij} = 
\begin{cases}
1, & \text{ if }i=j\\
-\frac{1}{\sqrt{d_id_j}} ,& \text{ if }i\ne j,(i,j)\in E(G)\\
0 ,& \text{ otherwise},\\
\end{cases}$$
where $d_i=\deg v_i$.
\end{defin}

We now list the important properties of $\Delta$: 
\begin{thm}[Properties of normalized Laplacian matrix]\ \\
\begin{enumerate}\vspace{-14pt}
\item Consider $\Delta$ as a linear operator  $\Delta:\rR^n\to \rR^n$, $\Delta$ is self-adjoint with respect to the scalar product $(\cdot,\cdot)$, i.e.,
$$
(x,\Delta y) = (\Delta x, y),
$$
for all $x,y \in \rR^n$ and standard scalar product 
$(x,y):=\sum_{i=1}^nx_iy_i$,

\item The operator $\Delta$ is non-negative, i.e., $\forall x\in \rR^n$,
$$
(\Delta x,x) \ge 0,
$$

\item For $x$ proportional to $(\sqrt{d_1},\cdots,\sqrt{d_n})$ holds $\Delta x =0$.

\item $\tr(\Delta) = n.$

\end{enumerate}
\end{thm}
\begin{proof}
   This properties are easy hold from similar theorem~\ref{thm:lapl} about Laplacian matrix. Let's note that for normalized Laplacian
   $$
   (\Delta x,x)=\sum_{(i,j)\in E(G)}\left(\frac{x_i}{\sqrt{d_i}}-\frac{x_j}{\sqrt{d_j}}\right)^2.
   $$
\end{proof}

The preceding properties have consequences for the eigenvalues of $\Delta$. We write them as $\lambda_k(\Delta)$ so that
  the eigenvalue equation becomes
$$
\Delta x^k = \lambda_k(\Delta) x^k,
$$
with $x^k$ being a corresponding \emph{eigenvector}. \\

\NB The eigenvalues and eigenvectors of the normalized Laplacian $\Delta$ can be rewritten by virtue of the (unnormalized) Laplacian in the following component-wise form:
$$(L x)_i=\lambda d_i\, x_i,\;\; i\in V.$$
From now on, we work on connected graphs. We can order the eigenvalues of $\Delta$ as
$$\lambda_1(\Delta)=0<\lambda_2(\Delta)\le \cdots\le \lambda_n(\Delta)$$
with the $<$ justified by

\begin{cor}[Properties of normalized Laplacian matrix spectrum]\ \\
\begin{enumerate}\vspace{-14pt}
\item $\lambda_i(\Delta)\ge0, \forall i = 1,\ldots,n$.
\item The smallest eigenvalue is $\lambda_1(\Delta) =0$. If $G$ is connected, then
$
\lambda_k(\Delta)~>~0
$ 
for $k>1$.
\item 
$$
   \sum_{k=1}^n \lambda_k(\Delta)=n. 
$$
\end{enumerate}
\end{cor}
\begin{proof}
    These arguments easily hold from the previous theorem. 
\end{proof}


\begin{lm}
Let $\mathcal{H}^k$ be the collection of all $k$-dimensional linear subspaces of
$\rR^n$. We have the following min-max characterization of the eigenvalues:
$$
\lambda_k= \max_{H_k \in  \mathcal{H}^{k-1}} \min_{x\in H_k^\bot\backslash \{0\}} 
\frac{(\Delta x,x)}{(x,x)} $$
and dually 
\begin{align*}
\lambda_k&= \min_{H_k \in { \mathcal{H}^{k}}} \max_{x\in H_k\backslash \{0\}}   
\frac{(\Delta x,x)}{(x,x)} =
\\&=  \max_{H_{n-k+1} \in \mathcal{H}^{n-k+1}} \min_{ x \in H_{n-k+1} \backslash \{0\}}  
\frac{(\Delta x,x)}{(x,x)} .
\end{align*}
\end{lm}

\begin{proof}
This is the well-known min-max theorem for the symmetric matrix $\Delta$. Precisely, taking $A=\Delta$ in Theorem \ref{thm:min-max}, we immediately obtain the desired min-max equalities.
\end{proof}

\begin{lm}
    
\begin{enumerate}[(a)]
  \item
Any eigenvector $x$ for some eigenvalue $\lambda \neq 0$ satisfies 
$$
\lambda =\frac{(\Delta x,x)}{(x,x)}.
$$
\item The second smallest eigenvalue of $\Delta$ is given by
$$
\lambda_2=\min_{\sum\limits_{i\in V}  \sqrt{d_i}x_i=0}  
\frac{(\Delta x,x)}{(x,x)} =\min_{ \sum\limits_{i\in V} d_i\,  y_i=0  }  \frac{\sum\limits_{(i,j)\in E(G)}
  (y_i-y_j)^2}{\sum_i d_i\ 
  y_i^2} 
$$
\item The largest eigenvalue of $\Delta$ is given by
 $$
\lambda_n=\max_{x \in \rR^n \setminus \{0\}}  
\frac{(\Delta x,x)}{(x,x)}  . 
$$
\item The eigenvalues satisfy 
$$
0\le \lambda \le 2. 
$$
\end{enumerate}
\end{lm}
\begin{proof}
    \begin{enumerate}[(a)]
        \item It follows from $\Delta x=\lambda x$ that $$
\frac{(\Delta x,x)}{(x,x)}=\frac{(\lambda  x,x)}{(x,x)}=\lambda \frac{(  x,x)}{(x,x)}=\lambda.
$$
        \item Since $(\sqrt{d_1},\cdots,\sqrt{d_n})$ is an eigenvector corresponding to $\lambda_1=0$, the first equality is due to min-max theorem of Rayleigh quotient in linear algebra. 
The second equality is then followed by taking $x_i=\sqrt{d_i}y_i$.
        \item This is also due to min-max theorem of Rayleigh quotient.
        \item Suppose that $x$ is an eigenvector corresponding to $\lambda$, then by (a), we have 
$$\lambda=\frac{(\Delta x,x)}{(x,x)}=\frac{\sumt_{(i,j)\in E(G)}\left(\frac{x_i}{\sqrt{d_i}}-\frac{x_j}{\sqrt{d_j}}\right)^2}{\sum_{i\in V}x_i^2}.$$
By Cauchy's inequality,
$$\sum_{(i,j)\in E(G)}\left(\frac{x_i}{\sqrt{d_i}}-\frac{x_j}{\sqrt{d_j}}\right)^2\le\sum_{(i,j)\in E(G)}2\left(\frac{x_i^2}{d_i}+\frac{x_j^2}{d_j}\right)=2\sum_{i\in V}\sum_{j\in V:j\sim i}\frac{x_i^2}{d_i}=2\sum_{i\in V}d_i\frac{x_i^2}{d_i} =2\sum_{i\in V}x_i^2$$
implying that $\lambda\le 2$.
    \end{enumerate}
\end{proof}

\begin{thm}
Among all the graphs with $n$ vertices, the complete graph has the largest possible $\lambda_2(\Delta)$ and the smallest possible $\lambda_n(\Delta)$.     
\end{thm}

\begin{proof}
In fact, since $\lambda_1(\Delta)=0$ and the trace $\sum_{i=1}^n\lambda_i(\Delta)=n$, we have $(n-1)\lambda_2(\Delta)\le \sum_{i=2}^n\lambda_i(\Delta)=n$ which means $\lambda_2(\Delta)\le \frac n {n-1}$. The equality holds iff. $\lambda_2(\Delta)=\lambda_3(\Delta)=\cdots=\lambda_n(\Delta)=\frac{n}{n-1}$. 
In this case, the eigenvalue $\frac n {n-1}$ has multiplicity $n-1$, and thus any vector $x$ orthogonal to the the constant vector $(1,1,\cdots, 1)$ (which serves an eigenvector corresponding to the smallest eigenvalue $\lambda_1(\Delta)=0$) must be an eigenvector corresponding to  $\lambda_2(\Delta)=\lambda_3(\Delta)=\cdots=\lambda_n(\Delta)=\frac{n}{n-1}$. This means that  every vector $x$ satisfying $\sum_{i=1}^n x_i=0$ must be an eigenvector. In particular, $y = (1,-1,0,\cdots,0)$ is an eigenvector corresponding to $\frac n {n-1}$. 
Thus, $1+\frac{1}{\sqrt{d_1d_2}}= \frac {(\Delta y,y)} {(y,y)}=\lambda= \frac{n}{n-1}$. This derives $d_1d_2=(n-1)^2$, and then together with $d_i\le n-1$, we have $d_1=d_2=n-1$. 
Similarly, we can obtain $d_i=n-1$ for any $i$, which means that $G$ is only possible to be a complete graph.     
\end{proof}

\NB This result indicates that  complete graphs can be uniquely characterized by the spectrum of normalized Laplacian. 
Below, we show a characterization of multi-partite graphs, meaning that in some sense, multi-partite graphs can be partially characterized by the spectrum of normalized Laplacian.

\begin{defin}
    An undirected graph is called \bf{$k$-partite} \ifof one can divide the vertices by $k$ disjoint groups such that any two vertices from the same group are not adjacent.
\end{defin}

In the same way one can define complete $k$-partite graph:

\begin{defin}
    A $k$-partite graph is called \bf{complete $k$-partite} graph if any two vertices from the different groups are adjacent.
\end{defin}

\begin{thm}
A graph with $\lambda_{t+1}(\Delta)>\lambda_{t}(\Delta)=...=\lambda_2(\Delta)=1$ must be a complete $(n-t+1)$-partite graph, where $2\le t\le n-1$.
\end{thm}
\begin{proof}
Denote by $\mu_1(A)\le \cdots\le \mu_n(A)$ the eigenvalues of $A$ ordered  non-decreasingly. 
We first prove an elementary observation  that for any given $k=1,\cdots,n$, the following statements hold: 
\begin{itemize}
    \item $\lambda_k(\Delta)> 1$ \ifof $\mu_{n-k+1}(A)< 0$
    \item $\lambda_k(\Delta)< 1$ \ifof  $\mu_{n-k+1}(A)> 0$ 
  \item  $\lambda_k(\Delta)= 1$ \ifof  $\mu_{n-k+1}(A)= 0$
\end{itemize}
In fact, by the relation $\Delta=I-D^{-\frac12}AD^{-\frac12}$, it is clear that if $D^{-\frac12}AD^{-\frac12} x=\eta x$ then $\Delta x=x-D^{-\frac12}AD^{-\frac12} x=(1-\eta)x$. This means that $1-\eta$ is an eigenvalue of $\Delta$, whenever $\eta$ is an eigenvalue of $D^{-\frac12}AD^{-\frac12}$. 
Now we write the eigenvalues of $D^{-\frac12}AD^{-\frac12}$ in non-decreasing order as $\eta_1(D^{-\frac12}AD^{-\frac12})\le \cdots\le \eta_n(D^{-\frac12}AD^{-\frac12})$, and then we have 
$\lambda_k(\Delta)=1-\eta_{n-k+1}(D^{-\frac12}AD^{-\frac12})$, $k=1,\cdots,n$. By the min-max theorem of Rayleigh quotient, it is clear that for any $i=1,\cdots,n$,
\begin{align*}
\eta_{i}(D^{-\frac12}AD^{-\frac12})&=\min_{X \in { \mathcal{H}^{i}}}\max_{x\in X}\frac{(D^{-\frac12}AD^{-\frac12}x,x)}{(x,x)} = 
\\&=\min_{X \in { \mathcal{H}^{i}}}\max_{x\in X}\frac{\sumt_{(u,v)\in E(G)}\frac{x_ux_v}{\sqrt{\deg u\cdot\deg v}}}{x_1^2+\cdots+x_n^2} \stackrel{y_v=\frac{x_v}{\deg v}}{=}
\\&=\min_{X \in { \mathcal{H}^{i}}}\max_{y\in X}\frac{\sumt_{(u,v)\in E(G)}y_uy_v}{\sumt_{v=1}^n\deg v\cdot y_v^2},
\end{align*}
$$\mu_{i}( A)=\min_{X \in { \mathcal{H}^{i}}}\max_{y\in X}\frac{\sumt_{(u,v)\in E(G)}y_uy_v}{\sum_{v=1}^n  y_v^2}.$$
Since for any $y\in \rR^n$, the quantity $$\frac{\sumt_{(u,v)\in E(G)}y_uy_v}{\sumt_{v=1}^n\deg v\cdot y_v^2}$$ lies between $$\frac{1}{\maxt_{v\in V(G)} \deg v} \frac{\sumt_{(u,v)\in E(G)}y_uy_v}{\sum_{v=1}^n  y_v^2}
~~\text{ and }~~ \frac{1}{\mint_{v\in V(G)} \deg v} \frac{\sumt_{(u,v)\in E(G)}y_uy_v}{\sum_{v=1}^n  y_v^2} .$$ From this, we can immediately get that $\mu_{i}( A)>0$ (or $<0$) if and only if $\eta_{i}(D^{-\frac12}AD^{-\frac12})>0$ (or $<0$). 
Indeed, all basic properties for the eigenvalue 0 of an  adjacency matrix can be translated into the language of the eigenvalue 1 of the corresponding normalized Laplacian. 

Now, let $G$ be a connected graph which is not complete. If $G$ is not a
complete multipartite graph, 
then $G$ has three vertices $u,v,w$ such that $u\sim v$ (i.e., $(u,v)\in E(G)$), $(u,w)\not\in E(G)$ and $(v,w)\not\in E(G)$. Consider a shortest path from $w$ to $\{u,v\}$,  we see
that $G$ has two possible types of induced subgraphs (up to symmetry):

\begin{center}
\begin{tikzpicture}[scale=2]
\draw (0,1)--(0,-1)--(1,0)--(2,0);
\draw[dashed] (2,0)--(2.5,-0.5);
\node (1) at (2.5,-0.5) {$\bullet$};
\node (1) at (2.3,-0.5) {$w$};
\node (1) at (0,1) {$\bullet$};
\node (1) at (0,-1) {$\bullet$};
\node (1) at (1,0) {$\bullet$}; 
\node (1) at (2,0) {$\bullet$}; 
\node (1) at (-0.2,1) {$u$};
\node (1) at (-0.2,-1) {$v$};
\node (1) at (1,0.2) {$w'$}; 
\node (1) at (2,0.2) {$w''$}; 
\end{tikzpicture}
~~~~~~~~~~~~~~~~~~
\begin{tikzpicture}[scale=2]
\draw (0,1)--(0,-1)--(1,0)--(2,0);
\draw (0,1)--(1,0);
\draw[dashed](2,0)--(2.5,-0.5);
\node (1) at (2.5,-0.5) {$\bullet$};
\node (1) at (2.3,-0.5) {$w$};
\node (1) at (0,1) {$\bullet$};
\node (1) at (0,-1) {$\bullet$};
\node (1) at (1,0) {$\bullet$}; 
\node (1) at (2,0) {$\bullet$}; 
\node (1) at (-0.2,1) {$u$};
\node (1) at (-0.2,-1) {$v$};
\node (1) at (1,0.2) {$w'$}; 
\node (1) at (2,0.2) {$w''$}; 
\end{tikzpicture}    
\end{center}

For the graph on the left hand side (i.e., the path graph on  the four vertices  $u,v,w',w''$), the adjacency matrix is 
$$B_1:=\left(\begin{array}{cccc}
   0  & 1& 0 &0 \\
   1  & 0& 1 &0 \\
     0  & 1& 0 &1 \\
        0  & 0& 1 &0 \\
\end{array}\right).$$
The characteristic polynomial for this matrix is $\mu^4-3\mu^2+1 = (\mu^2+\mu-1)(\mu^2-\mu-1)$, thus the second largest eigenvalue $\mu_3$ of the above adjacency matrix is $\mu_3= \frac{\sqrt{5}-1}{2}\approx 0.618$.  While for the graph on the right hand side (i.e., the graph induced on the four vertices  $u,v,w',w''$), the adjacency matrix is
$$B_2:=\left(\begin{array}{cccc}
   0  & 1& 1 &0 \\
   1  & 0& 1 &0 \\
     1  & 1& 0 &1 \\
        0  & 0& 1 &0 \\
\end{array}\right).$$
The characteristic polynomial for this matrix is $\mu^4-4\mu^2-2\mu+1 = (\mu+1) (\mu^3-\mu^2-3\mu+1)$, thus the second largest eigenvalue $\mu_3$ of the adjacency matrix satisfies $\mu_3\approx  0.3111$\footnote{It can be calculated using online program, e.g. \url{https://eigenvalues-calculator.bchrt.com/}}.  Since 
both of these two 
graphs on the four vertices  $u,v,w',w''$ have the second largest eigenvalue larger than 0, and either $B_1$ or $B_2$ is a principle submatrix of $A$, we have $$0<\min\{\mu_3(B_1),\mu_3(B_2)\}\le \mu_{3+n-4}(A)= \mu_{n-1}(A)$$ by interlacing (Theorem \ref{thm:Cauchy-Interlacing}). Thus, $\lambda_2(\Delta)<1$ which contradicts with the assumption. 
Hence, $G$ is a complete multipartite graph.

Suppose that $G$ is a complete multipartite graph on $k$ parts, where $k\ge 2$. Let's denote them by $V_1,\cdots,V_k$.  Then we find that $Ax=0$ for any $x$ with the property that $\sum_{j\in V_i}x_j=0$, $\forall i=1,\cdots,k$. Let's denote by $n_i=|V_i|$. For each $V_i$ the submatrix $A(V_j)$ has 0 eigenvalue with multiplicity equal to $n_j-1$. Thus, 0 is an eigenvalue of $A(G)$ and its multiplicity is $n-k$. 

Now let's prove that there is only one positive eigenvalue of $A$. Suppose that $\lambda>0$ and $Ax=\lambda x$. Then $\forall j\in V_i,\, i=1,\cdots,k$
\begin{equation}\label{eq:ck-partite}
(Ax)_j = \sum_{v\not\in V_i}x_v=\lambda x_j.
\end{equation}
Clearly, $x_j=x_{j'}$ for any $j,j'\in V_i$, and thus for simplicity, we write $y_i=x_j$ for $j\in V_i$. It follows from \eqref{eq:ck-partite} that
\begin{equation}\label{eq:ck-partite2}\sum_{i=1}^k n_iy_i=\sum_{v\in V}x_v=(n_i+\lambda)x_j=(n_i+\lambda)y_i,\,\forall j\in V_i,\, i=1,\cdots,k.
\end{equation}
Hence, $(n_1+\lambda)y_1=\cdots=(n_k+\lambda)y_k$. Since $n_i+\lambda>0$, and there exists $i\in\{1,\cdots,k\}$ such that $y_i\ne0$, we have $(n_1+\lambda)y_1=\cdots=(n_k+\lambda)y_k\ne 0$. 
Without loss of generality, we assume $(n_i+\lambda)y_i=c\ne 0$, $\forall i$. Then $y_i=c/(n_i+\lambda)$ and according to \eqref{eq:ck-partite2}, we obtain
$$ \sum_{i=1}^k n_iy_i=c\sum_{i=1}^k \frac{n_i}{n_i+\lambda}=(n_i+\lambda)y_i=c $$
implying that $f(\lambda)=1$ where 
$$f(t)=\sum_{i=1}^k \frac{n_i}{n_i+t},\; t\in[0,+\infty).$$
Clearly, 
$$f'(t)=-\sum_{i=1}^k\frac{n_i}{(n_i+t)^2}<0,\;\forall t\ge 0,$$
Since $f$ is continuous and  strictly decreasing on $[0,+\infty)$, $f(0)=k\ge 2$ and $\lim_{t\to+\infty}f(t)=0$, there is a unique point $\lambda\in (0,+\infty)$ satisfying $f(\lambda)=1$. (The root of the equation $f(\lambda)=1$ can be viewed as some kinds of Loomis's lemma). 
Therefore, combining Perron-Frobenius theorem (see Theorem \ref{thm:Perron-Frobenius}) and Statement \ref{st:connected-Perron}, $A$ has exactly one positive eigenvalue $\lambda$ whose multiplicity is 1.  
Then, we have that $\lambda_{n-1}(A)=\cdots=\lambda_{n-(n-k)}(A)=0>\lambda_{k-1}(A)$, which derives that $\lambda_{n+1-(k-1)}(\Delta)>\lambda_{n+1-k}(\Delta)=\cdots=\lambda_2(\Delta)=1$. Taking $k=n-t+1$, we obtain the desired relation in the theorem.
\end{proof}

\subsection{Cheeger inequality}\label{sc:cheeger} 

We consider a subset $S$ of the vertex set $V(G)$, and its complement
$\overline{S}=V(G)\setminus S$. We define the \emph{volume} of $S$ as
$$\vol (S)=  \sum_{v\in S} d_v,$$
and the \emph{boundary measure} of $S$ as
$$|\partial S|=\bigl|E(S,\overline{S})\bigr|=\Bigl\|\bigl\{(i,j)\in E(G):i\in S,j\not\in S\bigr\}\Bigr\|.$$
The reason that we call the set $\partial S:=\{(i,j)\in E(G);i\in S,j\not\in S\}$ the \emph{boundary} of $S$ is that each edge $(i,j)$ from $\partial S$ cross both $S$ and $\overline{S}$. This property is in the same spirit of boundary in topology -- a point is called a boundary point of a set $S$ if every (closed) neighbourhood of such point intersects both $S$ and the complement $\overline{S}$.

In fact, when we consider a graph $G$ as a one-dimensional simplicial complex \cite{Forman} (i.e., a set-family $\mathcal{K}$ of $V:=\{1,\cdots,n\}$ such that each $A\in \mathcal{K}$ is  a singleton $\{i\}$ or a set $\{i,j\}$ of cardinality 2 satisfying if $A\in \mathcal{K}$ and $A_1\subset A$, then $A_1\in \mathcal{K}$), the smallest closed star neighbourhood of 
an edge is the collection of its two end-vertices and such edge itself (i.e., for any edge $(i,j)$, its closed star neighbourhood can be expressed as $\{\{i\},\{j\},\{i,j\}\}$ and since we can identity $\{i,j\}$ with $(i,j)$, and identity $i$ with $\{i\}$, the neighbourhood will be simply written as $\{i,j,(i,j)\}\}$ under  the graph language). So, every neighbourhood of such edge $(i,j)\in\partial S$ intersects both  $S$ and the complement $\overline S$. 

\begin{center}
\begin{tikzpicture}[scale=2]
\draw[dashed,fill=orange!30](0,0) circle(0.3cm);
\draw (0,1)--(0,-1);
\node (1) at (0,0) {\color{orange}$\bullet$}; 
\node (1) at (-0.5,0.4) {$S$}; 
\node (1) at (0.5,0.4) {$\overline S$}; 
\end{tikzpicture}~~~~~~~~~~~~~~~~~~~~~~~~~~~~~~~~
\begin{tikzpicture}[scale=2]
\draw[very thick,orange!70] (-0.3,0)--(0.3,0);
\draw (0,1)--(0,-1);
\node (1) at (-0.3,0) {\color{orange}$\circ$}; 
\node (1) at (0.3,0) {\color{orange}$\circ$}; 
\node (1) at (-0.5,0.4) {$S$}; 
\node (1) at (0.5,0.4) {$\overline S$}; 
\end{tikzpicture}
\end{center}
\NB Note that in the left picture, in the original topological setting, any small neighbourhood (marked in orange) of a boundary point intersects the sets $S$ and $\overline{S}$. 
In the same spirit, on the right picture for the setting of graph, the smallest closed  neighbourhood (marked in orange) of any boundary edge (marked in orange) has one endpoint in  $S$ and the other in  $\overline{S}$.

\vspace{0.2cm}

In the case of normalized Laplacian we can give another estimation for the second eigenvalue $\l_2(\Delta)$. Let's 
introduce the \textbf{\emph{
Cheeger constant}}
$$
h:=  \min_{S\text{ nonempty proper subset of } V}\frac{\bigl|E(S,\overline{S})\bigr|}{\min \bigl\{\mathrm{vol}(S),\mathrm{vol}(\overline{S})\bigr\}}
$$
Now let us state a strong version of the  Cheeger inequality \cite{Chung}:
\begin{thm}\label{th:Cheeger}
$$
1-\sqrt{1- h^2}\le  \lambda_2(\Delta) \le 2 h.
$$
\end{thm}
\begin{proof}
Let the vertex set $V$ be divided  into the two disjoint sets $U,\overline{U}$ of
nodes, and let $U$ be the one with the smaller volume. Let's define $x = (x_i)\in \rR^n$ such as
    $$
        x_i = \begin{cases}
            1, & i\in U,\\
            -\a, & i\in \overline{U},
        \end{cases}
    $$
for positive $\alpha > 0$ such that the equation $\sum_{i\in V}x_i d_i =0$ holds, that is, $$\sum_{i\in U} d_i - \alpha\sum_{i \in \overline{U}} d_i  =0.$$
Since 
$\overline{U}$ is the subset with the larger volume 
$\sum_{i\in\overline{U}} d_i$, we have $\alpha \le 1$. Thus, for
our choice of $x$, the quotient  becomes 
\begin{eqnarray*}
  \lambda_2&=& \min \left\{\left.\frac{\sum\limits_{(i,j)\in E(G)}
  (x_i-x_j)^2}{\sum_i d_i\ 
  x_i^2}\right| \sum_i d_i\  x_i=0 \right\}\le\\
&\le&
\frac{(1+\alpha)^2\bigl|E(U,\overline{U})\bigr|}{\sum_{i \in U} d_i\  +\sum_{i
\in \overline{U}} d_i\ 
      \alpha^2} = \frac{(1+\alpha)^2\bigl|E(U,\overline{U})\bigr|}{\sum_{i \in U} d_i  + \sum_{i\in U} d_i  \, \alpha}=\frac{(1+\alpha)\bigl|E(U,\overline{U})\bigr|}{\sum_{i\in U} d_i}\le\\
&\le&
      2\frac{\bigl|E(U,\overline{U})\bigr|}{\sum_{i\in U}d_i}=
  2\frac{\bigl|E(U,\overline{U})\bigr|}{\mbox{vol}(U)}.
\end{eqnarray*}
Since this holds for
all such splittings of our graph $G$, we obtain 
the upper bound
$$ \lambda_2 \le 2 h.$$
Next, we turn to the hard part. Let $y$ be an eigenvector corresponding to $\lambda_2$. Since $\sum_i y_i=0$ and $y\ne0$, we have $\{i\in V:y_i>0\}\ne\varnothing$ and $\{i\in V:y_i<0\}\ne\varnothing$. 
Without loss of generality, we may assume that $$\sum_{i\in V:y_i>0}d_i\le \frac12\sum_{i\in V}d_i.$$
Let $x$ be the restriction of $y$ onto $\{i\in V:y_i>0\}$, that is, $$x_i=\begin{cases}
    y_i, &\text{ for }i\in V\text{ with }y_i>0\\
    0, &\text{ otherwise}.\\ 
\end{cases}$$
We shall first prove that $$R(x):=\frac{\sum\limits_{(i,j)\in E(G)}
  (x_i-x_j)^2}{\sum_i d_i\ 
  x_i^2}\le \lambda_2.$$
  For $i\in V$ with $x_i>0$, 
$$(L x)_i=\sum_{j\in V:(i,j)\in E}(x_i-x_j)\le\sum_{j\in V:(i,j)\in E}(y_i-y_j) =\lambda_2 d_iy_i=\lambda_2d_i x_i,
$$
meaning that 
$$(L x,x):=\sum_{i\in V} (L x)_ix_i=\sum_{i\in V:x_i>0} (L x)_ix_i\le \sum_{i\in V:x_i>0} \lambda_2 d_ix_ix_i=\lambda_2\sum_{i\in V}d_ix_i^2.$$
Thus, we have $R(x)\le \lambda_2$.

Also, it follows from 
$$\frac{\sum\limits_{(i,j)\in E(G)}
  (x_i+x_j)^2}{\sum_i d_i\ 
  x_i^2}=2-\frac{\sum\limits_{(i,j)\in E(G)}
  (x_i-x_j)^2}{\sum_i d_i\ 
  x_i^2} = 2-R(x)
  $$  
  that
  \begin{equation}\label{eq:R(x)(2-R(x))}
\left(\frac{\sum\limits_{(i,j)\in E(G)}
  \bigl|x_i^2-x_j^2\bigr|}{\sum_i d_i\ 
  x_i^2}\right)^2\le \frac{\sum\limits_{(i,j)\in E(G)}
   (x_i-x_j)^2\sum\limits_{(i,j)\in E(G)}
   (x_i+x_j)^2}{\bigl(\sum_i d_i\ 
  x_i^2\bigr)^2}=R(x)\bigl(2-R(x)\bigr).      
  \end{equation}

  Next we shall use a standard procedure in discrete analysis to show $$\frac{\sum\limits_{(i,j)\in E(G)}
  \bigl|x_i^2-x_j^2\bigr|}{\sum_i d_i\ 
  x_i^2}\ge h.$$
  In fact, taking $V_t=\{i\in V:x_i^2>t\}$ with $t\ge0$, and taking $M=\maxt_{i\in V}x_i^2$, we have 
  $$
  \int_0^M vol(V_t)dt=\int_0^M \sum_{i\in V_t}d_i dt
  =\int_0^M \sum_{i\in V}1_{t<x_i^2}d_i dt=\sum_{i\in V}\int_0^M 1_{t<x_i^2}d_i dt=\sum_{i\in V}x_i^2d_i, $$ where 
  $$1_{x_i^2>t}=\begin{cases}
    1, &\text{ for } t<x_i^2,\\
    0, &\text{ otherwise},\\ 
\end{cases}$$
  
  and
  \begin{align*}
\int_0^M |\partial V_t|dt&=\int_0^M \sum_{(i,j)\in E(G):i\in V_t,j\not\in V_t}1 dt
  =\int_0^M \sum_{(i,j)\in E(G)}1_{x_j^2\le t<x_i^2} dt = \\&=\sum_{(i,j)\in E(G)}\int_0^M 1_{x_j^2\le t<x_i^2} dt=\sum_{(i,j)\in E(G)} \bigl|x_i^2-x_j^2\bigr|.     
  \end{align*}
  Now, take $t_0\in[0,M]$ such that $$\frac{|\partial V_{t_0}|}{vol(V_{t_0})}=\min_{0\le t\le M}\frac{|\partial V_t|}{vol(V_t)}.$$
  Then, $vol(V_{t_0})\le vol(V_{0})=vol\bigl(\{i\in V:x_i>0\}\bigr)\le vol (V)/2$ which implies $vol(V_{t_0})\le vol(V\setminus V_{t_0})$ and thus 
  $$h\le \frac{|\partial V_{t_0}|}{\min\bigl\{vol(V_{t_0}),vol(V\setminus V_{t_0})\bigr\}}=\frac{|\partial V_{t_0}|}{vol(V_{t_0})}\le \frac{|\partial V_t|}{vol(V_t)}$$
   for any $0\le t\le M$. So,  
  $$|\partial V_t|\ge   vol(V_t) h  $$
and hence
$$\int_0^M |\partial V_t|dt\ge h \int_0^M vol(V_t)dt.$$
Therefore, together with all the above facts, we get
$$  \frac{\sum\limits_{(i,j)\in E(G)}
  \bigl|x_i^2-x_j^2\bigr|}{\sum_i d_i\ 
  x_i^2}=\frac{\int_0^M |\partial V_t|dt}{\int_0^M vol(V_t)dt} \ge h. $$
By \eqref{eq:R(x)(2-R(x))}, we have $$R(x)\bigl(2-R(x)\bigr)\ge h^2$$ 
which implies $$1-\sqrt{1-h^2}\le R(x)\le 1+\sqrt{1-h^2}.$$
We finally obtain 
$$\lambda_2\ge R(x)\ge 1-\sqrt{1-h^2}.$$
The whole proof is then completed.
\end{proof}

Note that $0\le h\le1$, and this yields $0\le 1-h^2\le \Bigl(1-\frac{h^2}{2}\Bigr)^2$ which is equivalent to  $h^2/2\le1-\sqrt{1-h^2}$. 
Then, by the strong  Cheeger inequality in Theorem \ref{th:Cheeger}, we immediately obtain the usual Cheeger inequality $h^2/2\le \lambda_2$.

\section{Graph $p$-Laplacian}

To illustrate the $p$-Laplacian on graphs, we need to recall the formulation of the classical $p$-Laplacian $L_p$ acting on a smooth function $f$:
$$L_p f=-\mathrm{div}\bigl(|\nabla f|^{p-2}\nabla f\bigr)$$
where $\nabla$ and $\mathrm{div}$ represent 
the gradient operator and the divergence operator, respectively.

Let  $G=(V,E)$ be a finite, undirected, simple graph with $n$ vertices and $m$ edges. 
Along this spirit, we define the gradient $\nabla:C(V)\to C(E)$ as $$\bigl(\nabla x\bigr) \bigl([i,j]\bigr)=x_j-x_i,$$
where 
$$C(V):=\{\text{real-valued functions on }V\}\cong \rR^n$$ 
and
$$C(E):=\{\text{real-valued functions on }E\}\cong \rR^m.$$
Now we can think of an undirected graph $G$ as a directed graph with two orientations, denoted by $[i,j]$ and $[j,i]$, on each edge $\{i,j\}$. 
And we shall fix an orientation. 
A vector field $\psi$ is a map $\psi:E\to \rR$ with the property $\psi([i,j])=-\psi([j,i])$ for every $i\sim j$. Note that this concept is similar to the flow introduced in Section \ref{sect:flow}.

The divergence $\mathrm{div}:C(E)\to C(V)$ is defined by $$\mathrm{div}\psi(i)=\sum_{j\in V:j\sim i}
\psi([i,j]).$$
We remark here that due to the anti-symmetry of $\psi$, for any nonempty proper subset $S\subset V$, $$\mathrm{div}\psi(S):=\sum_{(i,j)\in E(G):i\in S,j\not\in S}\psi([i,j])$$
indicates the `rate of flow' from $S$ to its complement $V\setminus S$. That is the reason why we call $\mathrm{div}$ the divergence on graphs.

For   $p>1$, the graph $p$-Laplacian $L_p:C(V)\to C(V)$ is defined by 
$$L_px=-\mathrm{div}(|\nabla x|^{p-2}\nabla x).$$
Clearly, if $p=2$, we obtain
 $L_2x=-\mathrm{div}(\nabla x)$ which is indeed the standard Laplacian matrix.

Let $\phi_p:\rR\to \rR$ be defined via $\phi_p(t)=|t|^{p-2}t$. 
 Then 
$$(L_px)_i=\sum_{j\in V:j\sim i}\phi_p(x_i-x_j),\;\forall i\in V.$$

By setting $\delta=\nabla$ and $\delta^*=-\mathrm{div}$, we have  $$L_p=\delta^*\phi_p\delta.$$
It should be noted that $\delta=B$ is actually the incidence matrix of $G$, and the adjoint operator $\delta^*$ is the transport $B^\top$. 

\subsection{Unnormalized version of the $p$-Laplacian  eigenproblem}
We say $({\lambda},{x})\in \rR\times (C(V)\setminus \{0\})$ is an eigenpair, if it satisfies the eigen-equation 
$$L_p {x}={\lambda}\,   
\phi_p(x)$$ i.e., 
\begin{equation}\label{eq:componentwise-pLap}
\sum_{j\in V:j\sim i}\phi_p(x_i-x_j)={\lambda}\,
\phi_p(x_i),\;\; i\in V.    
\end{equation}
The eigenvalues / eigenfunctions (also   called \emph{eigenvectors}) of $L_p$ are the critical values / critical points of the $p$-Rayleigh quotient 
$$
\mathcal{R}_p(x)=\frac{\|\nabla x\|_p^p}{\|x\|_{p}^p}:=\frac{\sumt_{(i,j)\in E(G)} |x_i-x_j|^{ p}}{\sum_{i\in V} |x_i|^{ p}}
$$

For $p=1$, the eigen-equation for $L_1$ reduces to the 
differential inclusion:
$$0\in \partial_x
\sum_{(i,j)\in E(G)}|x_i-x_j|
-\lambda\,\partial_x \sum_{i\in V} |x_i|.$$
And inspired by this inclusion condition, we 
define
$$L_1x:=\partial_x 
\sum_{(i,j)\in E(G)}|x_i-x_j|=\lim\limits_{p\to1, \hat{x}\to x}L_p \hat{x}$$
i.e., the limit points of $L_p \hat{x}$ for $p$ tends to 1, and $\hat{x}$ tends to $x$. 
However, for the sake of brevity in this section, 
we  focus only on $L_p$ when $p>1$.

\begin{st}\label{st:nontrivial-eigenvector}
If $x$ is an eigenvector 
corresponding to a nonzero 
eigenvalue of $L_p$, then 
$$\sum_{i=1}^n
\phi_p(x_i)=0.$$
\end{st}
\begin{proof}
Since $({\lambda},{x})$ is an eigenpair of $L_p$ with $\lambda\ne0$, 
$$\sum_{j\in V:j\sim i}\phi_p(x_i-x_j)={\lambda}\,
\phi_p(x_i),\;\; i=1,\cdots,n.$$
Summing up the above $n$ equations, we get
$$\sum_{i\in V}\sum_{j\in V:j\sim i}\phi_p(x_i-x_j)={\lambda}\,
\sum_{i\in V}\phi_p(x_i).$$
Note that the left hand side equals $$\sum_{(i,j)\in E(G)}\big(\phi_p(x_i-x_j)+\phi_p(x_j-x_i)\big)=0$$
where we used the easy fact that $\phi_p(t)+\phi_p(-t)=0$, $\forall t\in\rR$. 
This implies that the right hand side is zero, and  thus by $\lambda\ne0$, we have
$$\sum_{i\in V}\phi_p(x_i)=0.$$
\end{proof}

\begin{st}
For any eigenvalue $\lambda$ of $L_p$, $\lambda\le 2^{p-1}d$ where $d$ is the maximum degree of the graph. The equality $\lambda= 2^{p-1}d$ holds iff. $G$ has a regular bipartite component containing a vertex realizing the maximum degree.
\end{st}
\begin{proof}
    Since $|x_i-x_j|^p\le 2^{p-1}(|x_i|^p+|x_j|^p)$, we have
$$\sum_{(i,j)\in E(G)} |x_i-x_j|^p\mathop{\le}\limits^{(*)} 2^{p-1}\sum_{(i,j)\in E(G)}(|x_i|^p+|x_j|^p)=2^{p-1}\sum_{i\in V}d_i|x_i|^p\mathop{\le}\limits^{(**)} d2^{p-1}\sum_{i\in V}|x_i|^p.$$

The first inequality (*) holds \ifof $x_i=-x_j$ for any $(i,j)\in E(G)$, \ifof $G$ has a   bipartite connected component. 

The second inequality (**) holds \ifof $d_i=d$ for any $i\in V$ with  $x_i\ne0$. 
Combining the above 
facts on the equality case for (*) and (**), we get that  the whole equality $\sum_{(i,j)\in E(G)} |x_i-x_j|^p=d2^{p-1}\sum_{i\in V}|x_i|^p$ holds \ifof  $G$ has a regular bipartite component that contains vertices attaining 
maximum degree. 
It then follows from $$\lambda=\mathcal{R}_p(x)=\frac{\sumt_{(i,j)\in E(G)} |x_i-x_j|^p}{\sumt_{i\in V}|x_i|^p}$$
that the desired conclusion in the statement is true.
\end{proof}

Finally, we shall compute the $p$-Laplacian eigenvalues of complete graphs. This was done in a paper by Amghibech \cite{Amghibech}, but the original computation  contains a small mistake. Here we present a corrected proof.

\begin{thm}\label{thm:complete}
Let $G=(V,E)$ be a complete graph with $V=\{1,\cdots,n\}$. Then the nonzero eigenvalues of $L_p$ are    $n-\alpha-\beta+(\alpha^{\frac{1}{p-1}}+\beta^{\frac{1}{p-1}})^{p-1}$, where $\alpha,\beta\in\mathbb{Z}_+$ with $\alpha+\beta\le n$.
\end{thm}

\begin{st}\label{st:a1A-b1B}
Under the above setting, if $A$ and $B$ are two disjoint nonempty subsets of $V$, then there exist $a>0$ and $b>0$ such that $a\vec1_A-b\vec1_B$ is an eigenvector corresponding to the  eigenvalue $n-|A|-|B|+(|A|^{\frac{1}{p-1}}+|B|^{\frac{1}{p-1}})^{p-1}$ of $L_p$.
\end{st}

\begin{proof}
Note that if $a\vec1_A-b\vec1_B$ is an eigenvector of $L_p$, the corresponding  eigenequation \eqref{eq:componentwise-pLap} 
can be written in component-wise form as
$$\begin{cases}
|B|(a+b)^{p-1}+(n-|A|-|B|)a^{p-1}=\lambda a^{p-1}\\
|B|b^{p-1}-|A|a^{p-1}=0\\
|A|(a+b)^{p-1}+(n-|A|-|B|)b^{p-1}=\lambda b^{p-1}.
\end{cases}$$
By solving the equations above, we determine $a,b,\lambda$ as follows:  $a=c|B|^{\frac{1}{p-1}}$, $b=c|A|^{\frac{1}{p-1}}$ and $\lambda=n-|A|-|B|+(|A|^{\frac{1}{p-1}}+|B|^{\frac{1}{p-1}})^{p-1}$, where $c\ne0$. 
\end{proof}

\begin{st}\label{st:eigen-aAbB}
If $x$ is an  
eigenvector corresponding to a positive eigenvalue of $L_p$, then there exist disjoint nonempty subsets $A,B\subset V$, and $a,b>0$, such that $x=a\vec1_A-b\vec1_B$.
\end{st}

\begin{proof}
Let $(\lambda,x)$ be an eigenpair with $\lambda>0$. Without of loss of generality, we may assume  $x_1\le\cdots\le x_n$ due to the complete symmetry of $G$. 

By Statement \ref{st:nontrivial-eigenvector}, there exists $1\le\alpha<\beta\le n$ such that
$$x_1\le\cdots\le x_\alpha<0=x_{\alpha+1}=\cdots=x_{\beta-1}=0<x_{\beta}\le\cdots\le x_n.$$
Then, the eigenequation \eqref{eq:componentwise-pLap} 
becomes
\begin{equation}\label{eq:complete:eigen-system}
\begin{cases}
\lambda=-\sum_{k=1}^i(-1+\frac{x_k}{x_i})^{p-1}+\sum_{k=i}^n(1-\frac{x_k}{x_i})^{p-1}&\text{ for }i=1,\cdots,\alpha\\
\lambda=\sum_{k=1}^i(1-\frac{x_k}{x_i})^{p-1}-\sum_{k=i}^n(-1+\frac{x_k}{x_i})^{p-1}&\text{ for }i=\beta,\cdots,n\\
\sum_{k=1}^\alpha(-x_k)^{p-1}-\sum_{k=\alpha+1}^nx_k^{p-1} =0& 
\end{cases}    
\end{equation}
We shall prove that $x_1=\cdots=x_\alpha$, and we split  the rest of the 
proof into two cases:
\begin{itemize}
 \item $p>2$

 In this case, we use 
 the above equations to express $\lambda$ for $i=\alpha$ and for $i=1$. 
 \begin{align}
\lambda&\xlongequal{i=1\text{ in }\eqref{eq:complete:eigen-system}}\sum_{k=1}^n(1-\frac{x_k}{x_1})^{p-1} \notag
\\&\le \sum_{k=1}^\alpha(\frac{x_k}{x_\alpha}-\frac{x_k}{x_1})^{p-1}+\sum_{k=\alpha+1}^n(1-\frac{x_k}{x_1})^{p-1} \label{eq:alpha-1}\\
&= \sum_{k=\alpha+1}^n(\frac{x_k}{x_1}-\frac{x_k}{x_\alpha})^{p-1}+\sum_{k=\alpha+1}^n(1-\frac{x_k}{x_1})^{p-1} \label{eq:alpha-2}
\\&\le \sum_{k=\alpha+1}^n(1-\frac{x_k}{x_\alpha})^{p-1}\label{eq:alpha-3}
\\&\le -\sum_{k=1}^\alpha(-1+\frac{x_k}{x_\alpha})^{p-1}+\sum_{k=\alpha+1}^n(1-\frac{x_k}{x_\alpha})^{p-1}\xlongequal{i=\alpha\text{ in }\eqref{eq:complete:eigen-system}}\lambda\notag
 \end{align}
 where the inequality  \eqref{eq:alpha-1} is due to the fact that $1\le x_k/x_\alpha$, $k=1,\cdots,\alpha$, the equality \eqref{eq:alpha-2} is because of $\sum_{k=1}^\alpha(-x_k)^{p-1}=\sum_{k=\alpha+1}^nx_k^{p-1}$, the inequality  \eqref{eq:alpha-1} is based on the elementary inequality $a^t+b^t\le (a+b)^t$ whenever $a,b>0$, $t>1$, in which we take $a=\frac{x_k}{x_1}-\frac{x_k}{x_\alpha}$, $b=1-\frac{x_k}{x_1}$ and $t=p-1$. 
 So, all the inequalities are in fact equalities, which yield $$x_1=\cdots=x_\alpha.$$
 \item $1<p<2$

  In this case, we  use the  equations in \eqref{eq:complete:eigen-system} to represent $\lambda$ when $i=\alpha$ and $i=n$. We shall use the elementary inequality $a^t\ge (a+b)^t-b^t$ (or equivalently, $a^t+b^t\ge (a+b)^t$) whenever $a,b>0$, $0<t<1$, in which we shall take $t=p-1$.
\begin{align}
\lambda&\xlongequal{i=n\text{ in }\eqref{eq:complete:eigen-system}}\sum_{k=1}^n(1-\frac{x_k}{x_n})^{p-1}\notag
\\&\ge \sum_{k=1}^\alpha(\frac{x_k}{x_\alpha}-\frac{x_k}{x_n})^{p-1}-\sum_{k=1}^\alpha(-1+\frac{x_k}{x_\alpha})^{p-1}+\sum_{k=\alpha+1}^n(1-\frac{x_k}{x_n})^{p-1}\label{eq:p<2:alpha}
\\&=\sum_{k=\alpha+1}^n(\frac{x_k}{x_n}-\frac{x_k}{x_\alpha})^{p-1}+\sum_{k=\alpha+1}^n(1-\frac{x_k}{x_n})^{p-1}-\sum_{k=1}^\alpha(-1+\frac{x_k}{x_\alpha})^{p-1}\label{eq:p<2:alpha2}
\\&\ge \sum_{k=\alpha+1}^n(1-\frac{x_k}{x_\alpha})^{p-1}-\sum_{k=1}^\alpha(-1+\frac{x_k}{x_\alpha})^{p-1}\xlongequal{i=\alpha\text{ in }\eqref{eq:complete:eigen-system}}\lambda\label{eq:p<2:alpha3}
\end{align}
where the inequality \eqref{eq:p<2:alpha} is by $(1-\frac{x_k}{x_n})^{p-1}\ge (\frac{x_k}{x_\alpha}-\frac{x_k}{x_n})^{p-1}-(-1+\frac{x_k}{x_\alpha})^{p-1}$ for $k=1,\cdots,\alpha$, the equality \eqref{eq:p<2:alpha2} is because of $\sum_{k=1}^\alpha(-x_k)^{p-1}=\sum_{k=\alpha+1}^nx_k^{p-1}$, the inequality \eqref{eq:p<2:alpha3} follows from $(\frac{x_k}{x_n}-\frac{x_k}{x_\alpha})^{p-1}+(1-\frac{x_k}{x_n})^{p-1}\ge (1-\frac{x_k}{x_\alpha})^{p-1}$. 
Again, all the inequalities are actually 
 equalities, which imply  $x_1=\cdots=x_\alpha.$
\end{itemize}

If we change $x$ to $-x$, we obtain that $x$ is constant on $\{\beta,\beta+1,\cdots,n\}$; this achieves the proof of the remainder, that is, we have proved that $x=a\vec1_A-b\vec1_B$ for some $a,b>0$, where  $A=\{1,\cdots,\alpha\}$ and $B=\{\beta,\beta+1,\cdots,n\}$. 
\end{proof}

Combining Statements \ref{st:a1A-b1B} and \ref{st:eigen-aAbB}, we complete the proof of Theorem \ref{thm:complete}.
 It is clear that  when $p=2$, all the positive eigenvalues are equal to $n$.

\subsection{Normalized  $p$-Laplacian eigenvalue problem }

The normalized $p$-Laplacian is simply defined by $$\Delta_p=D^{-1}L_p$$ and the relevant eigenproblem is defined as
$$\Delta_p {x}={\lambda}\, 
\phi_p(x),$$
which can be written in the component-wise form like \eqref{eq:componentwise-pLap} as 
$$\frac{1}{d_i}\sum_{j\in V:j\sim i}\phi_p(x_i-x_j)={\lambda}\,
\phi_p(x_i),\;\; i\in V.    $$

Note that the definition of $\Delta_p$ involves the inverse of diagonal matrix formed by degrees. This fact implies that we have to work on  graphs with no isolated vertices. That is, we assume that every vertex is adjacent to some other vertices. 

The spectrum of $\Delta_p$ satisfies the following properties.
\begin{enumerate}
\item If $(\lambda, x)$ is an eigenpair of $\Delta_p$, then 
$$ R_p(x):= \frac{\sumt_{(i,j)\in E(G)} |x_i-x_j|^p}{\sumt_{i\in V}d_i|x_i|^p} =\lambda $$
\begin{proof}
Multiplying $x_i$ by 
$$\sum_{j\in V:j\sim i}\phi_p(x_i-x_j)={\lambda}\,
d_i\phi_p(x_i)$$
and summing up the $n$ equalities over $i=1,2,\cdots,n$, we get
$$ \sum_{i\in V}\sum_{j\in V:j\sim i}\phi_p(x_i-x_j)x_i={\lambda}\sum_{i\in V}
d_i\phi_p(x_i)x_i.$$
Note that the left hand side is 
$$ \sum_{(i,j)\in E(G)}(\phi_p(x_i-x_j)x_i+\phi_p(x_j-x_i)x_j)=\sum_{(i,j)\in E(G)}\phi_p(x_i-x_j)(x_i-x_j)=\sum_{(i,j)\in E(G)}|x_i-x_j|^p$$
while the right  hand side is 
$$\lambda\sum_{i\in V}
d_i\phi_p(x_i)x_i=\lambda\sum_{i\in V}
d_i|x_i|^p$$
due to the fact that $\phi_p(t)t=|t|^p$. 
The proof is then completed.
\end{proof}
\item The eigenvalues of $\Delta_p$ lie in the interval $[0,2^{p-1}]$. 
\begin{proof}
For any eigenvalue $\lambda$ of $\Delta_p$, and for any eigenvector $x$ corresponding to $\lambda$, we have 
$$0\le \lambda= \frac{\sum_{(i,j)\in E(G)} |x_i-x_j|^p}{\sum_{i\in V}d_i|x_i|^p} \mathop{\le}\limits^{(*)} \frac{2^{p-1}\sum_{(i,j)\in E(G)}(|x_i|^p+|x_j|^p)}{\sum_{i\in V}d_i|x_i|^p} =  \frac{2^{p-1}\sum_{i\in V}d_i|x_i|^p}{\sum_{i\in V}d_i|x_i|^p}=2^{p-1}.$$
\end{proof}
\item The largest eigenvalue is $2^{p-1}$ if and only if the graph has a  bipartite connected component.
\begin{proof}
The first inequality (*) reduces to an equality  \ifof $x_i=-x_j$ for any $(i,j)\in E(G)$, \ifof  $G$ has a   bipartite connected component. 
\end{proof}


\end{enumerate}

\newpage
\phantomsection

\end{document}